\documentclass[11pt, paper=a4, abstract]{scrartcl}

\usepackage[utf8]{inputenc}
\usepackage[backend=biber, style=apa]{biblatex}
\AtBeginBibliography{\footnotesize}
\usepackage{amsmath}
\usepackage{amssymb}
\usepackage{amsthm}
\usepackage{nicefrac}
\usepackage{scalerel} % for Kronecker sum
\usepackage{stackengine} % for Kronecker sum
\usepackage{xcolor}
\usepackage[normalem]{ulem} % tracked deletions
\usepackage{authblk} % for author formatting
\usepackage{csquotes}
\usepackage{setspace}
\usepackage{dsfont}
\usepackage{placeins}
\usepackage{longtable}
\usepackage{ltablex}   % combines longtable + tabularx
\keepXColumns
\usepackage{booktabs}
\usepackage{algorithm}
\usepackage{algpseudocode}
\usepackage{pdflscape}
\usepackage[
    font=footnotesize,
    labelfont={bf,sf},
    labelsep=period,
    format=plain,
]{caption}
\usepackage{graphicx}

\usepackage[left=2cm, top=2cm, right=2cm, bottom=2.5cm]{geometry}

\usepackage{mathtools}
\mathtoolsset{showonlyrefs}

\usepackage{hyperref}
\usepackage{aliascnt}
\hypersetup{
    colorlinks,
    linkcolor={blue!70!black},
    citecolor={blue!70!black},
    urlcolor={blue!70!black}
}
\providecommand{\tightlist}{%
  \setlength{\itemsep}{0pt}\setlength{\parskip}{0pt}}

\newcommand{\bsf}{\boldsymbol{f}}

\newcommand{\bsell}{\boldsymbol{\ell}}

\newcommand{\bsu}{\boldsymbol{u}}

\newcommand{\bsx}{\boldsymbol{x}}
\newcommand{\bsy}{\boldsymbol{y}}

\newcommand{\bsX}{\boldsymbol{X}}

\newcommand{\bfb}{\mathbf{b}}

\newcommand{\bfA}{\mathbf{A}}

\newcommand{\bfD}{\mathbf{D}}

\newcommand{\bfI}{\mathbf{I}}

\newcommand{\bfK}{\mathbf{K}}

\newcommand{\bfP}{\mathbf{P}}
\newcommand{\bfQ}{\mathbf{Q}}

\newcommand{\bsgamma}{\boldsymbol{\gamma}}
\newcommand{\bsdelta}{\boldsymbol{\delta}}

\newcommand{\bszeta}{\boldsymbol{\zeta}}
\newcommand{\bseta}{\boldsymbol{\eta}}

\newcommand{\bsnu}{\boldsymbol{\nu}}

\newcommand{\bspi}{\boldsymbol{\pi}}

\newcommand{\bstau}{\boldsymbol{\tau}}

\newcommand{\bsphi}{\boldsymbol{\phi}}

\newcommand{\bspsi}{\boldsymbol{\psi}}

\newcommand{\bstheta}{\boldsymbol{\theta}}

\newcommand{\bsGamma}{\boldsymbol{\Gamma}}

\newcommand{\bsLambda}{\boldsymbol{\Lambda}}

\newcommand{\bsOmega}{\boldsymbol{\Omega}}

\newcommand{\bfone}{\mathbf{1}}
\newcommand{\bfzero}{\mathbf{0}}

\newcommand{\calN}{\mathcal{N}}

\newcommand{\bbE}{\mathbb{E}}

\newcommand{\bbP}{\mathbb{P}}

\newcommand{\bbR}{\mathbb{R}}

\theoremstyle{plain}

\newaliascnt{lemma}{theorem}

\aliascntresetthe{lemma}
\newaliascnt{proposition}{theorem}
\newtheorem{proposition}[proposition]{Proposition}
\aliascntresetthe{proposition}
\newaliascnt{corollary}{theorem}
\newtheorem{corollary}[corollary]{Corollary}
\aliascntresetthe{corollary}

\theoremstyle{definition}
\newaliascnt{definition}{theorem}

\aliascntresetthe{definition}

\theoremstyle{remark}
\newaliascnt{remark}{theorem}

\aliascntresetthe{remark}

\newcommand{\dimann}[2]{\underset{\scriptscriptstyle \strut #1}{#2}}

\title{\Large Reconciling Interpretability with Covariate-Dependent Shape Flexibility in Penalized Transformation Models for Distributional Regression}

\date{}

\author[1]{Johannes Brachem%
\thanks{%
\texttt{brachem@uni-goettingen.de}}
}
\author[1]{Thomas Kneib}

\affil[1]{University of Göttingen, Germany}

\begin{document}

\maketitle

\begin{abstract}
A central challenge in distributional regression is to allow the shape of the conditional distribution of the response variable to vary flexibly with covariates while retaining directly interpretable effects on its mean and standard deviation.
We extend the penalized transformation model (PTM) family into a conditional-shape PTM, which assigns separate structured additive predictors to the conditional mean, standard deviation, and standardized distributional shape beyond location and scale. A covariate-dependent monotone transformation maps the standardized response to a fixed reference distribution, while affine standardization enforces mean zero and variance one for the induced standardized distribution at every covariate value. Thus, the first two predictors remain exactly the conditional mean and standard deviation. The shape predictor accommodates selected linear, nonlinear, group-specific, spatial, and interaction effects; regularization shrinks unsupported departures toward a nested reference-family location--scale model. We fit the PTM using mini-batched stochastic variational inference with model-aligned Gaussian blocks and staged optimization.
In simulations, the PTM recovers smooth mean and standard-deviation effects and a covariate-dependent transition from skewness to bimodality while suppressing unnecessary shape effects. Under a deliberately misspecified mixture design, it remains competitive with a structured additive Dirichlet-process mixture in test-set density and distribution-function accuracy, although all models show undercoverage and both flexible methods miss fine mixture features. Applications to $13{,}425$ Norwegian water-conductivity observations and $1{,}182{,}514$ German daily-temperature observations demonstrate selective group-specific, seasonal, and spatial shape variation. Predictive performance criteria favor the conditional-shape PTM over a fixed-shape PTM and a Gaussian location--scale model.
\end{abstract}


\medskip

\clearpage

\section{Introduction}

Regression analyses often focus on how covariates change the conditional
mean. In many applications, however, scientific interest extends to the entire
conditional distribution: variability may change, distributions may become
skewed or multimodal across groups, seasons, or locations, and risks may depend
on probabilities of extreme or threshold events. Distributional regression
addresses such questions by allowing multiple features of the response
distribution to vary with covariates. We present a transformation-based model
designed to combine flexibility beyond a fixed parametric response family with
direct response-scale interpretations of central tendency and dispersion, and
control over which covariates affect which distributional features. For an
absolutely continuous univariate response, separate structured additive
predictors govern the conditional mean, conditional standard deviation, and
standardized distributional shape. A covariate-dependent monotone
transformation maps the standardized response to a fixed reference
distribution, allowing the shape to be learned from the data while preserving the
exact meanings of the first two predictors. This separation also permits shape
flexibility to be restricted to scientifically supported effects and
regularized toward a simpler reference-family location--scale model.
Here, shape denotes the standardized distributional, or latent-variable,
shape remaining after location and scale have been separated.

\paragraph{Distributional regression.}
Structured additive distributional regression, including generalized additive
models for location, scale, and shape (GAMLSS), assigns structured additive
predictors to the parameters of a chosen response distribution
\parencite{Rigby2005-GeneralizedAdditiveModels,Klein2015-BayesianStructuredAdditivea,Umlauf2018-BamlssBayesianAdditive,Kneib2019-ModularRegressionLego}.
This provides one coherent conditional distribution, allows covariates to be
assigned selectively to distributional parameters, and supports nonlinear,
spatial, group-specific, and interaction effects. Interpretation can be direct
when the parameters have clear statistical meanings, but the attainable
shapes, support, and tail behavior remain constrained by the selected family.
Additive quantile regression instead targets selected conditional quantiles
without requiring a full parametric response distribution
\parencite{Waldmann2013-BayesianSemiparametricAdditive,Fasiolo2021-FastCalibratedAdditive}.
When quantile levels are fitted separately, the fits do not by themselves
define one likelihood-based conditional density and may cross
\parencite{Bondell2010-NoncrossingQuantileRegression}. Quantile regression is
therefore well suited when particular quantiles are the primary estimands,
whereas our target is a complete conditional distribution with distinct mean,
standard-deviation, and shape effects.

Flexible density-regression methods address this target more directly. The
single-weights Dirichlet-process mixture of
\textcite{Rodriguez-Alvarez2024-DensityRegressionDirichleta} specifies a
coherent conditional density through normal components whose means have
structured additive predictors, while the mixture weights and component
variances are shared across covariate values. Recent additive density regression
instead estimates conditional densities using
additive effects in Bayes Hilbert space
\parencite{Maier2025-AdditiveDensityRegression}. In the mixture model, covariates shift the means of the mixture components, whereas in additive density regression, covariate effects act directly on the entire conditional density. Our proposed
decomposition serves a complementary aim: it separates covariate effects on
the response mean and standard deviation, which retain exact response-scale
interpretations, from effects on standardized shape.

\paragraph{Transformation models.}
Transformation models provide another route to flexible conditional
distributions by mapping the response monotonically to a prespecified reference
distribution. Conditional transformation models let this map depend on
covariates through structured partial transformations, with monotonicity
constraints and smoothness regularization
\parencite{Hothorn2014-ConditionalTransformationModels,Hothorn2018-MostLikelyTransformations,Carlan2024-BayesianConditionalTransformation}.
The deep transformation model of
\textcite{Sick2021-DeepTransformationModels} allows each learnable parameter
of a chained transformation to depend on the inputs through neural networks,
while selected parameters may remain input-independent. The deep conditional
transformation model of
\textcite{Baumann2021-DeepConditionalTransformation} provides a closer
structured comparison. Both its response--covariate interaction, which governs
a finite vector of response-basis parameters, and its transformed-scale shift
may contain structured additive terms, neural-network predictors, or both.
The approach developed here instead targets a specific
response-scale decomposition. Affine standardization makes separate structured
predictors exactly the conditional response mean and standard deviation, while
the finite-vector transformation predictor governs standardized shape only and
is regularized toward a nested reference-family location--scale model.

Distribution-free location--scale regression learns a common monotone
transformation together with covariate-dependent location and scale predictors
for the transformed response
\parencite{Siegfried2023-DistributionfreeLocationscaleRegression}. These
predictors have useful transformation- and probability-scale interpretations
but are not generally the response mean and standard deviation.
Another transformed-response regression approach, presented by \textcite{Kowal2025-MonteCarloInference},
likewise learns a common unknown monotone
transformation jointly with linear, quantile, or Gaussian-process regression
models through Bayesian-bootstrap Monte Carlo.

The direct methodological predecessor
of our approach is the fixed-shape penalized transformation model (PTM) of
\textcite{Brachem2026-BayesianPenalizedTransformation}. It combines structured
additive location and scale predictors with one learned, covariate-invariant
latent-variable distribution and uses Markov chain Monte Carlo (MCMC)
inference. Because all observations share the same standardized latent-variable
shape, a fixed-shape PTM cannot represent, for example, two groups with equal first two
moments but different skewness or modality. Conditional-shape PTMs close this gap by allowing standardized shape to
vary with covariates through structured vector-valued additive terms. In the
remainder, PTM denotes the conditional-shape PTM when the referent is
unambiguous; we use fixed-shape PTM explicitly for its submodel.

\paragraph{Model construction.}
For covariates $\bsx$, a PTM specifies
\begin{equation}
    Y = \mu(\bsx)
    + \sigma(\bsx) h^{-1}\!\left(Z \mid \bsdelta(\bsx)\right),
    \qquad
    Z \sim F_Z,
\end{equation}
where $F_Z$ is a prespecified reference distribution and $h$ is a flexible
monotone transformation. Allowing its parameters $\bsdelta(\bsx)$ to vary with
covariates does not by itself guarantee that
$h^{-1}(Z\mid\bsdelta(\bsx))$ has mean zero and variance one; without further
constraints, nominal shape effects could therefore absorb response-level
location or scale variation. We prevent this through an affine standardization
for every covariate-specific transformation. Consequently, $\mu(\bsx)$ and
$\sigma(\bsx)$ are exactly the conditional mean and standard deviation of $Y$,
while $\bsdelta(\bsx)$ governs only standardized shape.
The particular transformation used here is flexible over a bounded
region but retains the asymptotic tail family of the reference distribution, up
to location and scale.
Because the transformation parameters jointly determine a
standardized distribution, we interpret shape effects through
standardized partial density previews, quantile summaries, and complete
conditional distributions.

We model $\mu(\bsx)$ and $\log\sigma(\bsx)$ with conventional structured
additive predictors. The shape predictor supports the same linear, smooth,
group-specific, spatial, and interaction effects, but each term returns a
finite vector over the ordered transformation parameters. Regularization acts
both over the covariate domain and across these parameters.
The penalties in the transformation predictor shrink unsupported
shape effects toward zero; the nested reference-family location--scale model is reached
in the limit when all terms are zero.
Finite-vector
structured predictors and separable penalties are established in deep
conditional transformation models and functional additive mixed models. In
the latter, the second direction indexes a functional response \parencite{Scheipl2015-FunctionalAdditiveMixed,Greven2017-GeneralFrameworkFunctional}.
In deep conditional transformation models and in our approach, it instead
indexes parameters of a monotone transformation, yielding closely related
covariate-by-transformation constructions \parencite{Baumann2021-DeepConditionalTransformation}.

\paragraph{Stochastic variational inference.}
Covariate-dependent shape adds many transformation coefficients and smoothing
parameters to those of the mean and scale predictors. Rather than extending the
full-data MCMC strategy used for PTMs to this larger model, we therefore use
automatic-differentiation stochastic variational inference
\parencite{Kucukelbir-AutomaticDifferentiationVariational} to achieve scalable
inference. Related variational
methods have been developed for additive and structured additive distributional
regression
\parencite{Kleinemeier2023-ScalableEstimationStructured,
Lichter2024-VariationalInferenceUncertainty,
Callegher2025-StochasticVariationalInferencea}. Our implementation uses a pragmatic Gaussian block structure aligned with the
model's additive building blocks. This choice supports large data sets but is
an approximation: the Gaussian form and independence between blocks cannot in
general reproduce non-Gaussian posterior shapes or arbitrary cross-block
dependence.

\paragraph{Contributions and outline.}
The paper's primary methodological contribution is a PTM in which structured
covariate effects govern standardized shape, while the mean and standard
deviation retain exact response-scale interpretations and the shape predictor
is regularized toward a nested reference-family location--scale model.
Supporting work integrates established finite-vector additive machinery with
the model's ordered standardized-shape parameters, provides model-specific
summaries for nonlinear shape effects, and implements a blocked stochastic
variational fitting strategy. The resulting predictor and inference
implementations are available in Liesel, Liesel-GAM, and Liesel-PTM
\parencite{Riebl2026-LieselPythonFramework,
Brachem2026-LieselgamGeneralizedAdditive,
Brachem2026-LieselptmPenalizedTransformation}.

The empirical studies examine both the benefits and limits of the framework.
In a controlled simulation, the PTM recovers smooth mean and scale effects together
with a covariate-dependent transition from skewness to bimodality and shrinks
unsupported shape effects toward the reference model. Under deliberate
misspecification, it remains predictively competitive with the Dirichlet-process mixture comparator DDPstar,
although the hardest setting exposes undercoverage and limits in recovering
fine features for both models. Applications to $13{,}425$ Norwegian
water-conductivity observations and $1{,}182{,}514$ German daily-temperature
observations demonstrate selective group-specific, seasonal, and spatial shape
effects, practical fitting at two data scales, and better predictive performance
than the fixed-shape PTM and Gaussian location--scale model.

The paper proceeds as follows. \autoref{sec:model} defines the model, the
standardized transformation, and the interpretation of shape effects.
\autoref{sec:mvpred-ptm}
formulates
the finite-vector shape predictor and its regularization, and
\autoref{sec:inference} presents the stochastic variational approximation and
optimization strategy. \autoref{sec:simulations} evaluates recovery,
regularization, prediction under misspecification, and uncertainty calibration.
\autoref{sec:applications} presents the two applications, and
\autoref{sec:discussion} discusses scope, practical limitations, and possible
extensions of conditional-shape PTMs.

\section{Model} \label{sec:model}

For the density construction, we assume that $F_Z$ is an absolutely
continuous reference distribution with density $p_Z$. For the mean--variance
standardization, we additionally require that $F_Z$ is nondegenerate and has
finite first two moments. We use the location--scale-standardized representation, so that
$\bbE(Z)=0$ and $\operatorname{Var}(Z)=1$; any alternative reference is first
put into this form by subtracting its mean and dividing by its standard
deviation.
For a set of observations $\{(y_i, \bsx_i)\}_{i=1}^N$, where $\bsx_i=[x_{i1}, \dots, x_{iM}]^\top$ is a covariate vector, an application of the density transformation theorem to our transformation model gives the conditional density
$$
    p_Y(y_i \mid \bsx_i)
    =
    p_Z\left( h\left( \frac{y_i - \mu(\bsx_i)}{\sigma(\bsx_i)} \mid \bsdelta(\bsx_i) \right) \right)
    \frac{1}{\sigma(\bsx_i)} \,
    h'\left( \frac{y_i - \mu(\bsx_i)}{\sigma(\bsx_i)} \mid \bsdelta(\bsx_i) \right)
    ,
$$
where $p_Z$ is the reference density specified above; all analyses reported here use the standard Gaussian reference.
We incorporate dependence on covariates through conventional univariate structured additive predictors for the location and the natural logarithm of the scale, and through a multivariate structured additive predictor for the transformation function's parameters $\bsdelta(\bsx_i) \in \bbR^D$,
$$
    \begin{alignedat}{4}
        \mu(\bsx_i)         & \quad=\quad{} \eta_{\mu,i}
                            & \quad=\quad{} f^{\mu}_{1}(\bsx_i) + \dots + f^{\mu}_{K_\mu}\!(\bsx_i)
                            & \quad=\quad{} \bfb_\mu(\bsx_i)^\top \bsgamma_\mu                                     \\
        \log \sigma(\bsx_i) & \quad=\quad{} \eta_{\sigma,i}
                            & \quad=\quad{} f^{\sigma}_{1}(\bsx_i) + \dots + f^{\sigma}_{K_\sigma}\!(\bsx_i)
                            & \quad=\quad{} \bfb_\sigma(\bsx_i)^\top \bsgamma_\sigma                               \\
        \bsdelta(\bsx_i)    & \quad=\quad{} \bseta_{\delta,i}
                            & \quad=\quad{} \bsf^{\delta}_{1}(\bsx_i) + \dots + \bsf^{\delta}_{K_\delta}\!(\bsx_i)
                            & \quad=\quad{} \bsGamma_\delta \bfb_\delta(\bsx_i).
    \end{alignedat}
$$
In the following, we omit the observation index $i$ when defining and interpreting the model at a generic covariate value $\bsx$, but retain it when referring to individual observations or aggregating likelihood contributions.
Here, $\bsgamma_\mu \in \bbR^{L_\mu}$ and $\bsgamma_\sigma \in \bbR^{L_\sigma}$ are coefficient vectors, $\bsGamma_\delta \in \bbR^{D \times L_\delta}$ is a coefficient matrix, and $\bfb_{\mu}(\bsx)$, $\bfb_{\sigma}(\bsx)$, and $\bfb_{\delta}(\bsx)$ are the corresponding vectors of basis-function evaluations. We describe the transformation function and the shape predictor in detail in
\autoref{sec:trafo} and
\autoref{sec:mvpred-ptm}, respectively. The univariate predictors are treated as special cases of the multivariate predictor. The number of elements in $\bsdelta(\bsx)$ is a manually chosen hyperparameter that determines the potential flexibility of the transformation function $h$; typical choices range from $D=5$ to $D=20$. The reported sensitivity results suggest that regularization across the elements of $\bsdelta(\bsx)$ reduces sensitivity to $D$ once the transformation is sufficiently flexible. We therefore balance flexibility against computational efficiency when choosing $D$.

\paragraph{Levels and interpretation.}

The model includes three levels, each of which can be viewed from an observational and a generative perspective. They are summarized in \autoref{tab:model-levels}.
\begin{itemize}
    \tightlist
    \item The \textit{response level} $Y$. Our primary interest is the complete
          conditional distribution of $Y$ given $\bsx$, including its expectation,
          standard deviation, quantiles, and shape. The location and scale
          predictors retain their familiar interpretations from a Gaussian
          location--scale model: terms in $\eta_\mu$ shift the conditional
          expectation, while exponentiated terms in $\eta_\sigma$ act
          multiplicatively on the conditional standard deviation. Features beyond
          location and scale can be understood more easily through the
          latent-variable level.

    \item The \textit{latent-variable level} $R$. Because $R$ has conditional mean
          zero and variance one, its density isolates the shape of the conditional
          response distribution from its location and scale. Covariate-dependent
          changes in this shape are mediated through
          $h(R\mid\bsdelta(\bsx))$.
    \item The \textit{reference level} $Z$. The reference level is the anchor for our probability model and the target of regularization. It also serves as the starting point for the generative use of our model.
\end{itemize}

\begin{table}[tb]
    \small
    \centering
    \caption{Summary of levels of analysis in our model.\label{tab:model-levels}}
    \begin{tabularx}{\linewidth}{rcXXcc}
        \toprule
        Level           & Variable & Observational view                       & Generative view                     & $\bbE(\cdot | \bsx)$ & $\operatorname{Var}(\cdot | \bsx)$ \\
        \midrule
        Response        & $Y$      & $Y \mid \bsx\sim F_{Y \mid \bsx}$        & $Y = \mu(\bsx) + \sigma(\bsx) R$    & $\mu(\bsx)$          & $\sigma(\bsx)^2$                   \\
        \addlinespace
        Latent variable & $R$      & $R = \sigma(\bsx)^{-1}(Y - \mu(\bsx))$   & $R = h^{-1}(Z \mid \bsdelta(\bsx))$ & $0$                  & $1$                                \\
        \addlinespace
        Reference       & $Z$      & $Z = h\bigl(R \mid \bsdelta(\bsx)\bigr)$ & $Z \sim F_Z$                        & $0$                  & $1$                                \\
        \bottomrule
    \end{tabularx}
\end{table}

\subsection{Transformation function} \label{sec:trafo}

Throughout this subsection, we suppress the dependence of the transformations and their derived quantities on $\bsdelta$ unless it is relevant to the argument.
We define $h: \bbR \to \bbR$ as a composition
$$
    h(r) = (\tilde h \circ g)(r) = \tilde h(g(r)),
$$
where
$\tilde h$ is the \textit{onion spline} transformation function developed
by \textcite{Brachem2026-DataEfficientGenerativeModeling} and $g$ is a standardization that enforces
the separation of location and scale information from the transformation.
We first review the definition of $\tilde h$ and then describe $g$.

\paragraph{Main transformation.}
The transformation function $\tilde h: \bbR \to \bbR$ is a strictly increasing cubic B-spline on a core subdomain $[a, b]$,
$$
    \tilde h(r) = \sum_{j=1}^J B_j(r) \left(\vartheta_1 + \sum_{j'=2}^j \exp(\vartheta_{j'}) \right), \quad \text{if} \quad r \in [a, b],
$$
and the identity function $\tilde h(r)=r$ for $r \in (-\infty, a) \cup (b, \infty)$. The functions $B_j$ are cubic B-spline basis functions defined on an equally spaced knot sequence $\xi_{-2} < \xi_{-1} < \dots < \xi_{J+1}$, with $\xi = \xi_{j+1} - \xi_j$ denoting the distance between neighboring knots. The boundaries of the core domain, $a$ and $b$, are hyperparameters chosen by the researcher and determine the knot sequence through $\xi_2 = a$ and $\xi_{J-3} = b$.
We enforce the constraints $\tilde h(a)=a$ and $\tilde h(b)=b$ by specifying the parameters $\vartheta_1, \dots, \vartheta_J$ as
\begin{align}
    \vartheta_1 & = \xi_0                                                                                                                                                                  \\
    \vartheta_j & = \log(\xi)                                                                               & \text{for}                      & \quad j \in \{2,3,4\} \cup \{J-2, J-1, J\} \\
    \vartheta_j & = \delta_{j-4} -  \log\bigl( {\textstyle \sum_{d=1}^D} \exp(\delta_d)\bigr) + \log(D \xi)
                & \text{for}                                                                                & \quad j \in \{ 5, \dots, J-3\},
\end{align}
so that we fix seven parameters and compute the remaining ones as functions of $D = J-7$ unconstrained parameters $\bsdelta = [\delta_1, \dots, \delta_D]^\top \in \bbR^D$. This definition has several useful consequences.
The fixed boundary increments make the transition to the identity twice
continuously differentiable. The normalization of the flexible increments
fixes their total rise at $D\xi$; if all elements of $\bsdelta$ are equal, all
increments equal $\xi$ and $\tilde h$ reduces to the identity.

\paragraph{Standardization.}
Applying $\tilde h$ directly induces the unstandardized latent random variable
$\tilde R=\tilde h^{-1}(Z)$ with cumulative distribution function (CDF) $\tilde F_R(r) = F_Z\bigl(\tilde h(r)\bigr)$
and corresponding density $\tilde p_R(r) = p_Z\bigl(\tilde h(r)\bigr)\, \tilde h'(r)$.
The stated reference assumptions and the identity tails of $\tilde h$
ensure that $\tilde R$ is nondegenerate with finite mean and variance; see
\autoref{cor:identity-tail-moments}.
The expectation and variance implied by this density are
\begin{equation} \label{eq:R-expectation-and-var}
    \tilde\mu_R = \int r\times\tilde p_R(r)\,\mathrm dr
    \quad
    \text{and}
    \quad
    \tilde\sigma_R^2 = \int (r-\tilde\mu_R)^2 \times \tilde p_R(r)\,\mathrm dr.
\end{equation}
In general, these quantities do not satisfy $\tilde\mu_R=0$ and $\tilde\sigma_R^2=1$, and thus the direct use of $\tilde h$ does not preserve the central model assumption that $R$ has zero mean and unit variance.
We therefore introduce the standardization map
$$
    g(r)=\tilde\sigma_R\,r+\tilde\mu_R
$$
and apply $h(r) = (\tilde h \circ g)(r)$ instead. The distribution of $R$ then has CDF $F_R(r) = F_Z(h(r))$, and $R$ satisfies $\bbE(R)=0$ and $\operatorname{Var}(R)=1$ by construction.
Although the integrals in \eqref{eq:R-expectation-and-var} are not readily available analytically, we can approximate them efficiently by piecewise Gauss--Legendre quadrature between the outer spline knots. By \autoref{prop:identity-tails}, the omitted contributions come only from
the reference tails. For the symmetric bounds and standard Gaussian reference
used here, these contributions cancel for the mean and are negligible for the
variance.

\autoref{fig:onion-spline-illustration} illustrates the unstandardized and standardized transformations and the corresponding latent-variable densities for two draws from the onion-spline parameter distribution.
The transformation panels can be read like quantile--quantile (Q--Q) plots with reversed axes: at probability level $u$, the plot shows $(Q_R(u),Q_Z(u))$, whereas a conventional Q--Q plot shows $(Q_Z(u),Q_R(u))$. Since $F_R(r)=F_Z(h(r))$, a transformation function evaluation that lies above the identity indicates that more cumulative probability is placed below $r$ than under the reference distribution, and a transformation function evaluation below the identity indicates the opposite. Thus, for negative $r$, a curve above the identity indicates more left-tail mass than the reference, whereas for positive $r$, a curve below the identity indicates more right-tail mass; the opposite deviations indicate less mass in the respective tail.

\begin{figure}[tb]
    \centering
    \includegraphics[width=\linewidth]{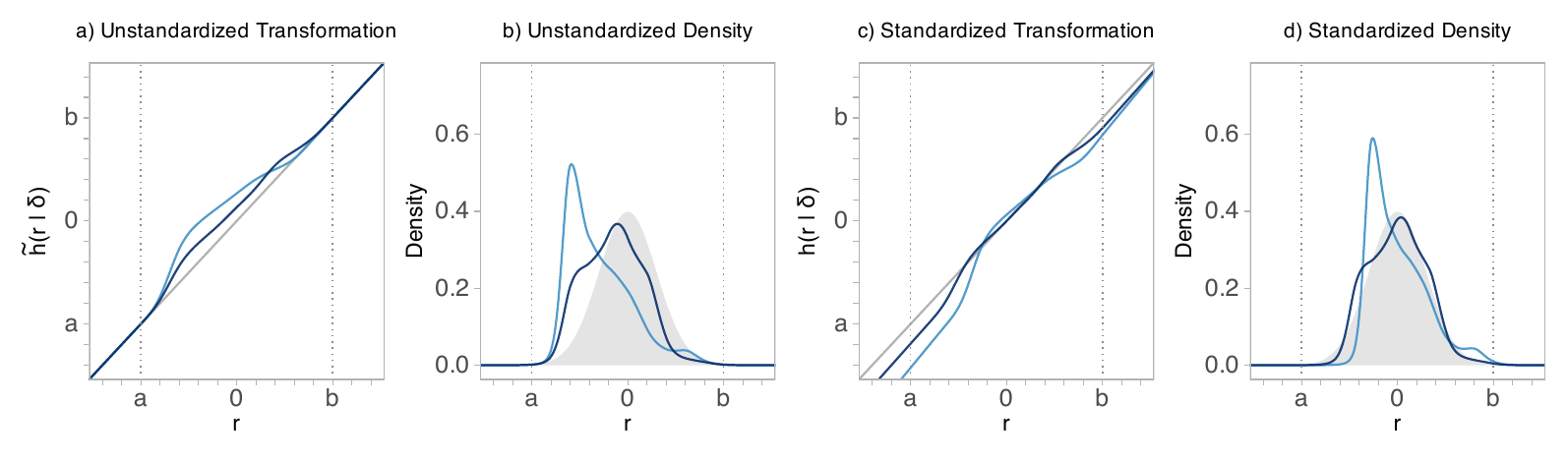}
    \caption{
        Illustration of unstandardized and standardized onion-spline transformations
        and their induced latent-variable densities. The two colored curves in each
        panel correspond to two different sets of $\bsdelta$ parameters. The left pair shows the unstandardized transformation and
        density, while the right pair shows the corresponding transformation and
        density after the mean-scale standardization. The gray line denotes the
        identity transformation or the standard Gaussian reference density, and
        the dotted lines mark the boundaries $a$ and $b$ of the spline core.
    }
    \label{fig:onion-spline-illustration}
\end{figure}

\paragraph{The reference distribution as a base model.}
If all elements of $\bsdelta$ are equal, that is, $\bsdelta=c\bfone$, then
$\tilde h$ is the identity transformation and $\tilde R=Z$.
Because the reference distribution has mean zero and unit variance by assumption, the
subsequent standardization has no effect, and hence $R=Z$ as well.
The reference distribution can thus be viewed as a base model for $R$ that is reached when $\bsdelta=c\bfone$ for any $c \in \bbR$.
More generally,
the normalization of the flexible increments is invariant to a shared additive shift,
so $\bsdelta$ and $\bsdelta+c\bfone$ define the same transformation for every
$\bsdelta\in\bbR^D$ and $c\in\bbR$. A common shift is therefore
unidentified, which motivates us to impose a sum-to-zero constraint on
the shape predictor in \autoref{sec:mvpred-ptm}. We use the base-model property to set
up penalties that shrink differences among the free transformation parameters
$\bsdelta$, so that our model includes regularization toward the base model.
This property is shared with location--scale PTMs and scalable composite transformations
\parencite{Brachem2026-BayesianPenalizedTransformation,Brachem2026-DataEfficientGenerativeModeling}, where it was presented more formally.

The transformation and standardization machinery described in
this section is implemented in the Python library Liesel-PTM
\parencite{Brachem2026-LieselptmPenalizedTransformation}.

\subsection{Tail behavior}
The following proposition and corollary describe the relevant tail properties
and their implications for the existence of moments.

\begin{proposition}[Tail behavior]
    \label{prop:identity-tails}
    Let $\tilde R$ have CDF
    $\tilde F_R(r)=F_Z\bigl(\tilde h(r)\bigr)$ and density $\tilde p_R$. Then, for
    $r\in(-\infty,a]\cup[b,\infty)$, we have $\tilde F_R(r)=F_Z(r)$ and $\tilde p_R(r)=p_Z(r)$.
\end{proposition}

\begin{proof}
    Outside $[a,b]$, $\tilde h(r)=r$ and $\tilde h'(r)=1$. Substitution into
    the CDF and density of $\tilde R$ gives the result.
\end{proof}

The onion spline therefore changes the reference distribution only on a
bounded interval. Because $\tilde h$ is continuous, strictly increasing, and
equals the identity in both tails, it has a unique inverse for every real
value, and $\tilde p_R$ integrates to one. This directly implies that
$\tilde R$ and $R$ have finite moments whenever $Z$ does:

\begin{corollary}[Moment existence of the latent variables]
    \label{cor:identity-tail-moments}
    Under the conditions of Proposition~\ref{prop:identity-tails}, for every
    $q>0$,
    $$
        \bbE\bigl(|Z|^q\bigr)<\infty
        \quad\Longleftrightarrow\quad
        \bbE\bigl(|\tilde R|^q\bigr)<\infty.
    $$
    If, in addition, $Z$ has finite variance and
    $\tilde\sigma_R>0$, then $R=(\tilde R-\tilde\mu_R)/\tilde\sigma_R$ is
    well-defined and
    $$
        \bbE\bigl(|Z|^q\bigr)<\infty
        \quad\Longleftrightarrow\quad
        \bbE\bigl(|\tilde R|^q\bigr)<\infty
        \quad\Longleftrightarrow\quad
        \bbE\bigl(|R|^q\bigr)<\infty.
    $$
\end{corollary}

\begin{proof}
    By Proposition~\ref{prop:identity-tails}, $\tilde R$ and $Z$ have identical
    tail contributions to their absolute moments, while their contributions on
    the bounded interval $[a,b]$ are finite. This proves the first equivalence.
    In particular, finite variance of $Z$ implies that the moments required to
    standardize $\tilde R$ exist for every $\bsdelta$. The second equivalence
    then follows because a nondegenerate affine transformation preserves the
    existence of absolute moments.
\end{proof}

The construction does not change which absolute moments exist. Thus, the
moments needed for standardization exist whenever $Z$ has finite variance.
Identity tails are sufficient but not necessary for this guarantee; their
stronger benefit is well-controlled extrapolation on the unstandardized scale.

\paragraph{Tails after standardization.}
After standardization, the tails are no longer identical to the reference density
at the same value of $r$. However, the standardized distribution retains the reference tail family up to an
affine transformation. In particular,
\begin{equation}
    \bbP(R\leq a^*)=F_Z(a),
    \qquad
    \bbP(R\geq b^*)=1-F_Z(b),
    \label{eq:standardized-reference-tail-probabilities}
\end{equation}
where $a^* = \tilde\sigma_R^{-1}(a-\tilde\mu_R)$ and $b^* = \tilde\sigma_R^{-1}(b-\tilde\mu_R)$.
This is useful when $[a,b]$ covers the part of the
distribution whose shape is most strongly informed by the data and the reference
distribution provides a sensible model for more extreme outcomes. It is
restrictive when the asymptotic tails are themselves of interest. With a
Gaussian reference, for example, the model cannot learn polynomial tails beyond
$a^*$ and $b^*$. Because $a^*$ and $b^*$ depend on
$\bsdelta$ through $\tilde\mu_R$ and $\tilde\sigma_R$, the reference-tail
regime begins at shape-dependent rather than fixed thresholds. The spline can
therefore remain flexible over a finite, data-relevant tail range, although
its asymptotic tail class remains fixed.

These observations also guide the practical choice of $a$, $b$, and $D$.
The interval $[a^*,b^*]$ always contains the probability mass
$F_Z(b)-F_Z(a)$, so choosing $a$ and $b$ as reference quantiles directly
determines how much of the distribution is modeled flexibly. With a standard
Gaussian reference, $a=-4$ and $b=4$ place approximately $99.994\%$ of the
probability mass in the flexible region, making this choice a reasonable default.
The bounds may be chosen asymmetrically if additional flexibility is needed in
one tail, and Q--Q plots of the transformed residuals can indicate whether
either bound should be extended because the reference-tail regime begins too
early.

The width of the flexible region should be considered jointly with $D$: at
fixed $D$, widening $[a,b]$ reduces the local resolution of the transformation
function. Standardization can amplify this effect because
$b^*-a^*=(b-a)/\tilde\sigma_R$, thereby spreading the same $D$ parameters over
a wider interval on the scale of $R$. A larger $D$ may then be useful. On the
other hand, the number of coefficients in a covariate-dependent shape predictor
grows linearly with $D$; for example, $\bsGamma_\delta$ contains
$D L_\delta$ coefficients before constraints. We therefore use $D=10$ as a
modest default that offers useful local resolution while limiting the
computational and estimation burden. Regularization guards against unnecessary
complexity, and larger values can be used when the chosen bounds or diagnostic
checks indicate that more flexibility is needed.

\subsection{Interpreting additive effects}
\label{sec:interpretation}

The model uses three additive predictors, but addition acts on different scales.
The standardization in \autoref{sec:trafo} gives the location and scale
parameters direct interpretations because the latent variable $R$ has zero mean
and unit variance. In particular, we have
\begin{equation}
    \bbE(Y\mid\bsx)=\mu(\bsx),
    \qquad
    \operatorname{SD}(Y\mid\bsx)=\sigma(\bsx).
\end{equation}
An additive change of $c$ in $\mu(\bsx)$ shifts the conditional distribution by
$c$, while the same change in $\log\sigma(\bsx)$ multiplies the conditional
standard deviation by $\exp(c)$. Addition in the shape predictor
$\bsdelta(\bsx)$ has no analogous immediate interpretation on the response
scale. While examining its effect in this section, we suppress the dependence on $\bsx$
and first consider the raw transformation, since the standardization $g$ leaves the
shape of the density unchanged.

Because $\tilde h(a)=a$ and $\tilde h(b)=b$, the transformation's total rise over
the core interval is fixed at $b-a$; $\bsdelta$ only redistributes this increment
budget. We can therefore direct our attention to a relative increment profile
\begin{equation} \label{eq:interpretation-profile}
    \pi_d
    =
    \frac{\exp(\delta_d)}
    {\sum_{k=1}^{D}\exp(\delta_k)},
    \qquad
    d=1,\dots,D,
\end{equation}
where $\pi_d$ is the share of the total rise over $[a,b]$ assigned to flexible
increment $d$ and $\sum_{d=1}^{D}\pi_d=1$. The raw transformation slope is a
smoothed version of this increment profile, normalized so that the uniform
profile yields unit slope. Although the full
profile determines the transformation and hence the induced distribution, an
individual profile entry still has no isolated density-level interpretation. By the
change-of-variables identity in \autoref{sec:trafo}, the full profile affects the
density through both the transformation value $\tilde h(r)$, which determines where the reference density is evaluated, and the local slope $\tilde h'(r)$, which scales the density evaluation.
These effects can reinforce or offset one another.

Now suppose that a baseline shape parameter vector $\bsdelta_a$ and a perturbation
$\bsdelta_b$ combine as $\bsdelta_c=\bsdelta_a+\bsdelta_b$. For
$v\in\{a,b,c\}$, let $\bspi_v$ be the profile obtained from $\bsdelta_v$
through \eqref{eq:interpretation-profile}. Addition of the parameter vectors becomes
pointwise multiplication of their profiles followed by normalization:
\begin{equation}
    \pi_{c,d}
    =
    \frac{\pi_{a,d}\pi_{b,d}}
    {\sum_{k=1}^{D}\pi_{a,k}\pi_{b,k}},
    \qquad d=1,\dots,D.
    \label{eq:interpretation-product-pooling}
\end{equation}
The pooling identity in \eqref{eq:interpretation-product-pooling} reinforces
shared emphasis, whereas a small share in either profile suppresses the
corresponding increment. A uniform perturbation profile is
neutral, while $\bsdelta_b=-\bsdelta_a+c\bfone$, where $c$ can be zero or any other real-valued scalar, makes
$\bsdelta_c$ constant and recovers the uniform reference profile. The pooling
identity and exact expressions for the induced changes in the raw transformation
and log density are derived in \autoref{app:interpretation-details}.
\autoref{fig:interpretation-additive-perturbation} traces one reinforcing
perturbation from the shape parameter vectors through the increment profiles to
the raw transformations and densities.

\begin{figure}[tb]
    \centering
    \includegraphics[width=\linewidth]{
        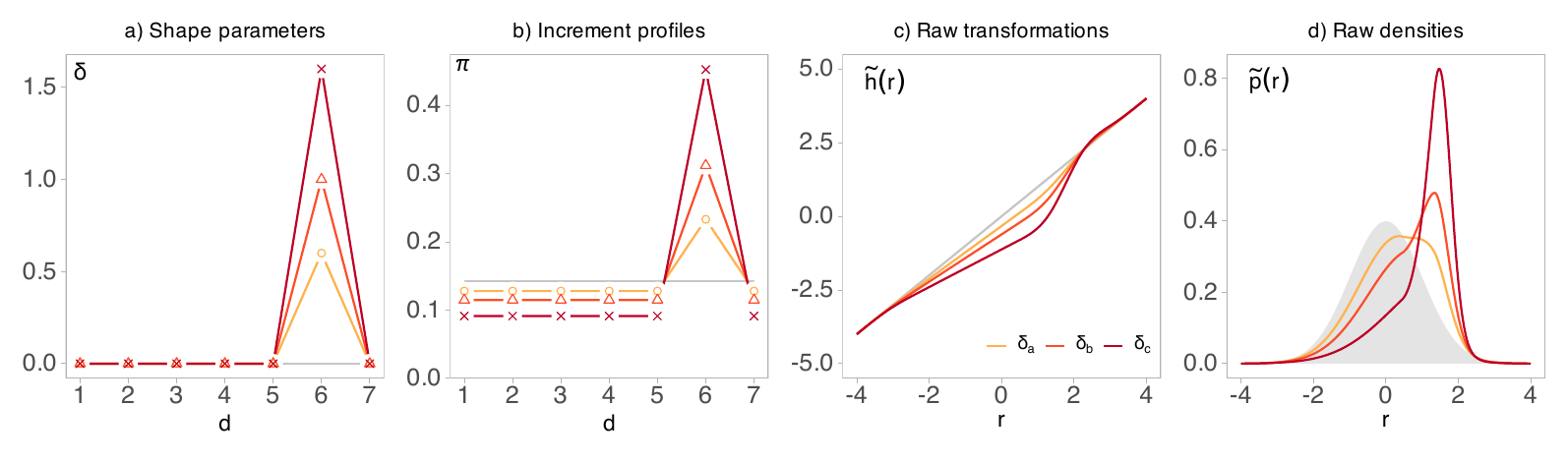
    }
    \caption{
        End-to-end illustration of an additive shape perturbation. The combined
        shape parameter vector satisfies
        $\bsdelta_c=\bsdelta_a+\bsdelta_b$. Panels~a) and b) show the parameter
        vectors and their relative increment profiles, respectively; in
        panel~b), $\bspi_c$ is the normalized pointwise product of $\bspi_a$ and
        $\bspi_b$. Panels~c) and d) show the resulting raw transformations and
        densities. The semitransparent gray elements mark the neutral parameter
        vector, uniform profile, identity transformation, and standard Gaussian
        density, as appropriate.
    }
    \label{fig:interpretation-additive-perturbation}
\end{figure}

The figure also shows the distribution induced by the $\bsdelta_b$ example in isolation.
This density should be understood only as an interpretive aid: the perturbation enters the model through the shape parameters, not as a separate density that is combined with the baseline density.
In particular, $p_c$ is not a mixture, convolution, or normalized pointwise product
of $p_a$ and $p_b$. The exact combination occurs on the increment-profile
scale; the resulting density is obtained only after constructing and
standardizing the combined transformation.

\autoref{fig:interpretation-qualitative-examples} shows four examples illustrating how the isolated perturbation density qualitatively previews a perturbation's effect. Shared
asymmetry is reinforced in panel~a), whereas opposing asymmetries balance in
panel~b). Distinct unimodal and bimodal contributions combine in panel~c), and
the same isolated bimodal perturbation exactly cancels its inverse in
panel~d). Thus, the perturbation density can indicate the type of feature being
favored, while the combined density remains conditional on the baseline.

\begin{figure}[tb]
    \centering
    \includegraphics[width=\linewidth]{
        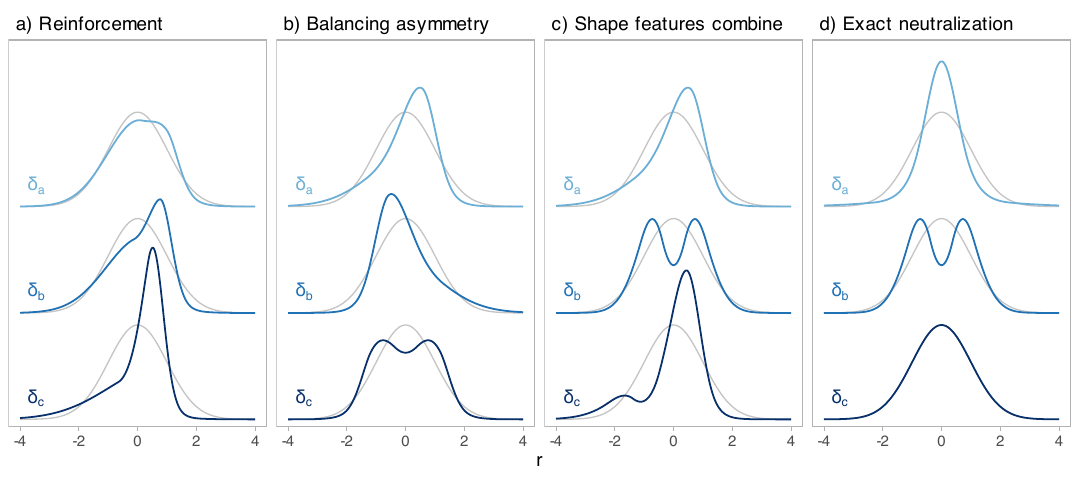
    }
    \caption{
        Illustrative density-level consequences of profile pooling. Within
        each panel, the densities induced by $\bsdelta_a$, $\bsdelta_b$, and
        $\bsdelta_c$ are shown on vertically shifted rows, with the standard
        Gaussian reference repeated in gray at each level. Panel~a)
        shows reinforcement of a shared asymmetric feature. In panel~b),
        mirror-image profiles balance overall skewness but retain two regions
        of concentration. Panel~c) combines an asymmetric unimodal shape
        contribution with a symmetric bimodal contribution. Panel~d) uses
        the bimodal $\bsdelta_b$ from panel~c) together with
        $\bsdelta_a=-\bsdelta_b$, so $\bsdelta_c=\bfzero$ and $p_c$ is exactly
        the reference density. In every panel, $p_c$ is induced by
        $\bsdelta_c=\bsdelta_a+\bsdelta_b$, and all densities have mean zero and
        unit variance. The vertical shifts are only for display.
    }
    \label{fig:interpretation-qualitative-examples}
\end{figure}

Together, the exact profile identity and the examples delimit the
interpretation of additive shape effects. Additivity is exact at the
increment-profile level, whereas its standardized density consequences are
nonlinear and baseline dependent.

\FloatBarrier

\section{The multivariate shape predictor in penalized transformation models} \label{sec:mvpred-ptm}

We define a $D$-dimensional multivariate structured additive predictor
$\bseta \in \bbR^{D}$ as a sum of covariate functions with a multivariate
codomain:
\begin{align}
    \dimann{(D \times 1)}{\bseta} \quad = \quad
    \dimann{(D \times 1)}{\bsf_1(\bsx)} + \cdots +
    \dimann{(D \times 1)}{\bsf_K(\bsx)} \quad=\quad
    \dimann{(D \times L)}{{\bsGamma \vphantom{(_i)}}}
    \dimann{(L \times 1)}{\bfb(\bsx)}.
\end{align}
For term $k \in \{1, \dots, K\}$, the $(L_k \times 1)$ vector of
basis function evaluations is given by $\bfb_k(\bsx)$; $L = L_1 + \cdots +
L_K$ denotes the sum of basis dimensions over all terms.

\paragraph{Prior specification.}
Let $\bsGamma_k$ denote the slice of $\bsGamma$ that contains the
$L_k$ columns corresponding to term $k$, such that
$\bsf_k(\bsx) = \bsGamma_k\bfb_k(\bsx)$. We can then collect all coefficients
of term $k$ across the predictor's $D$ dimensions in a vector:
$$
    \dimann{(DL_k \times 1)}{\bsgamma_k} =
    \operatorname{vec}(
    \dimann{(D \times L_k)}{\bsGamma_k\vphantom{_k}}) = [
    \gamma^{(k)}_{11}, \dots, \gamma_{D1}^{(k)}, \dots, \gamma^{(k)}_{1L_k}, \dots, \gamma_{DL_k}^{(k)}
    ]^\top.
$$
We define a potentially rank-deficient multivariate Gaussian prior
for $\bsgamma_k$ with unnormalized density
\begin{equation}
    \label{eq:mvstar-prior}
    p(\bsgamma_k | \bstau^2_k, \psi^2_k) \propto
    \operatorname{pdet}\bigl(\bfP_k\bigr)^{\nicefrac{1}{2}}
    \exp \left(
    - \frac{1}{2}
    \bsgamma_k^\top
    \bfP_k
    \bsgamma_k
    \right),
\end{equation}
where we write $\operatorname{pdet}(\bfP_k)$ for the product of the
nonzero eigenvalues of $\bfP_k$.
The dependence on $\bstau^2_k \in \bbR^{M_k}_{>0}$ and $\psi^2_k \in \bbR_{>0}$ is notationally suppressed on the right-hand side and arises through the potentially rank-deficient positive-semidefinite precision matrix $\bfP_k$, given by
$$
    \dimann{(DL_k \times DL_k)}{\bfP_k} =
    \bfP_k(\bstau^2_k, \psi^2_k) =
    \biggl(
    \dimann{(L_k \times L_k)}{\bfP^{(x)}_{k}(\bstau^{2}_{k}) }
    \otimes \bfI_D\biggr) +
    \biggl(
    \bfI_{L_k} \otimes
    \dimann{(D \times D)}{\bfP^{(d)}(\psi^{2}_{k})}
    \biggr),
$$
where $\bfI_{D}$ denotes the $(D\times D)$ identity matrix and
$\otimes$ denotes the Kronecker product.
The first component of $\bfP_k$ governs within-dimension
regularization in the covariate directions, with the variances in
$\bstau^2_k$ controlling its strength. The second component governs
between-dimension regularization through
$\bfP^{(d)}(\psi^{2}_{k}) = \psi^{-2}_k \bfK_k^{(d)}$; smaller values of $\psi^2_k$
imply stronger regularization across the $D$ predictor dimensions.

Finite-vector structured additive predictors and two-direction
penalties are established in deep conditional transformation models and
functional additive mixed models
\parencite{Baumann2021-DeepConditionalTransformation,Scheipl2015-FunctionalAdditiveMixed,Greven2017-GeneralFrameworkFunctional}.
We adapt this machinery to ordered standardized-shape parameters and actively
use between-dimension regularization to shrink unsupported effects toward a
nested reference-family location--scale model.
The general predictor construction described here is implemented in
Liesel-GAM \parencite{Brachem2026-LieselgamGeneralizedAdditive}.

Linear constraints can be applied along both axes of $\bsGamma_k$. A
within-dimension constraint has the form
$\bsGamma_k\bfA_k^{(x)\top}=\bfzero$, whereas a
cross-dimension constraint has the form
$\bfA^{(d)}\bsGamma_k=\bfzero$. We impose both through
null-space reparameterizations; full details are given in
\autoref{app:mv-star-predictors}.

\paragraph{Between-dimension first-order random-walk penalty.}

For between-dimension regularization of the terms in $\bseta_{\delta}$, we choose a first-order random-walk penalty as described in \autoref{app:mvterm-bases-and-penalties}; that is, $\bfK_k^{(d)} = \bfD^\top_1 \bfD_1$ for all $k=1, \dots, K_{\delta}$, where $\bfD_1$ denotes the $(D-1) \times D$ first-difference matrix. This penalty discourages large differences between neighboring shape coefficients and thus abrupt changes in the corresponding increment profile, transformation slope, and induced density. In the absence of information supporting a more specific structure, regularization toward such smoothly varying distributional shapes provides a sensible default. As $\psi_k^2 \to 0$, the coefficients of term $\bsf_k(\bsx)$ approach a common value; the sum-to-zero constraint introduced below forces this value to zero, so the term is removed from the shape predictor. If all shape terms reach this limit, the transformation becomes the identity and the latent-variable density reduces to the reference density. The reported Scenario B comparison shows little change between $D=10$ and $D=20$, but the application comparisons show that predictive criteria can vary across $D$ and core choices. The penalty may therefore reduce, but does not eliminate, sensitivity to $D$, especially because large $D$ must be balanced against the added computational demand. Because the number of coefficients in every shape term grows linearly with $D$, we use $D=10$ as a modest default. The center and right columns in the second row of \autoref{fig:mvterm-regularization-comparison} illustrate how decreasing $\psi_k^2$ suppresses variation among the densities induced by a covariate-dependent multivariate penalized B-spline (P-spline). Using an intercept-only shape predictor, \autoref{fig:onion-spline-prior} shows that $\psi^2=0.25$ permits substantially greater variation in the transformation and density, particularly away from the center, whereas draws for $\psi^2=0.01$ concentrate more closely around the identity transformation and reference density.

\begin{figure}[tb]
    \centering
    \includegraphics[width=\linewidth]{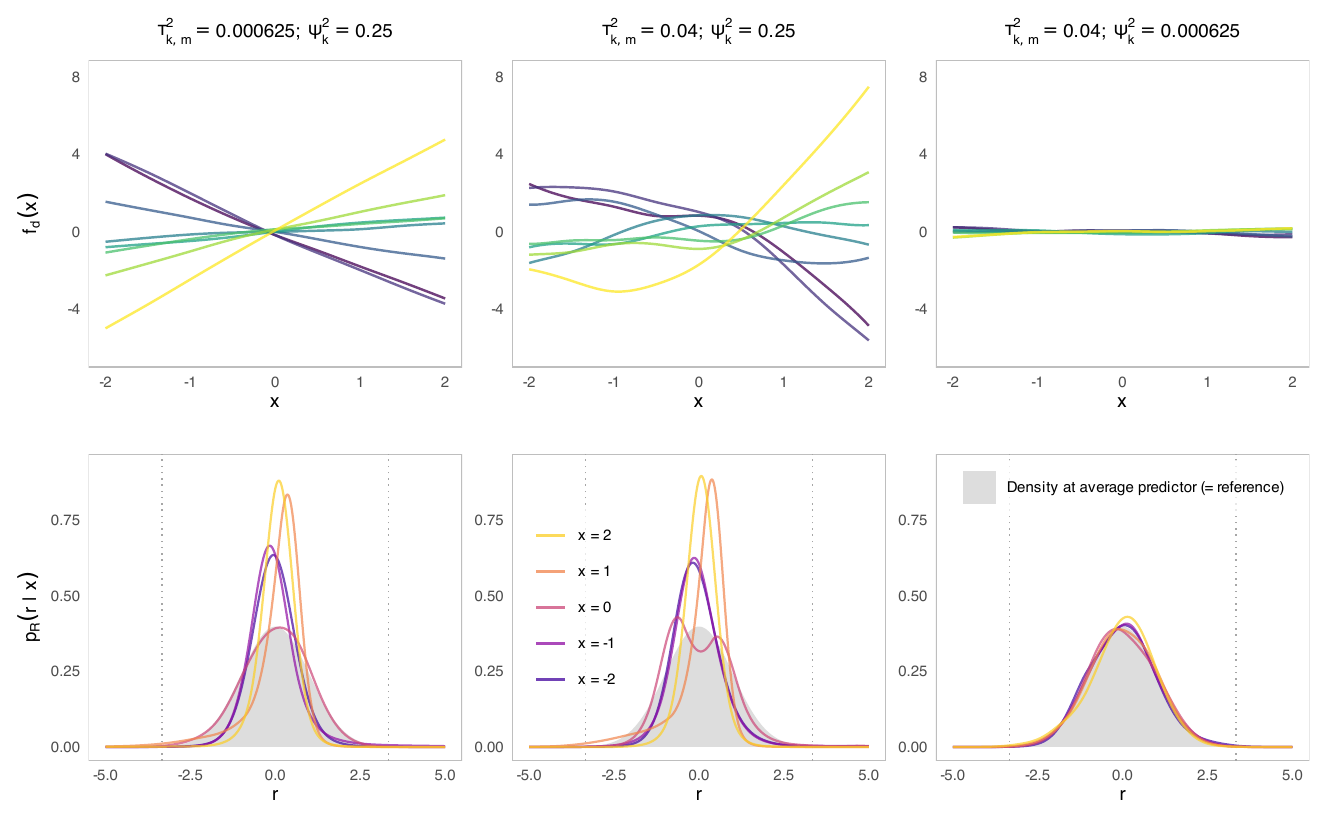}
    \caption{Complementary effects of the within-dimension variance $\tau_{k,m}^2$ and the first-order random-walk between-dimension variance $\psi_k^2$ for a multivariate P-spline with $D=8$. The center column is the common baseline with $\tau_{k,m}^2=0.04$ and $\psi_k^2=0.25$; the left column decreases $\tau_{k,m}^2$ to $0.000625$, and the right column decreases $\psi_k^2$ to $0.000625$. The prior is parameterized in the $(D-1)$-dimensional constrained space, so the functions in the first row sum to zero across dimensions at every $x$. The second row shows the corresponding standardized latent-variable densities at $x\in\{-2,-1,0,1,2\}$. The gray area shows the density evaluated at the empirical average of the term predictor over $x$. Because the term is centered, this average predictor is zero and therefore yields the chosen reference density, which is standard Gaussian in this illustration; it is not the average of the colored conditional densities. Dotted lines mark the onion-spline boundaries $a=-10/3$ and $b=10/3$.}
    \label{fig:mvterm-regularization-comparison}
\end{figure}

\paragraph{Between-dimension sum-to-zero constraint.}
As noted in \autoref{sec:trafo}, the onion-spline transformation is invariant to a shared shift of its shape parameters: $\bsdelta$ and $\bsdelta+c\bfone$ yield the same normalized increments and therefore the same transformation and density. We remove this unidentified constant using the constraint framework above, choosing $\bfA^{(d)}=\bfone^\top$ and imposing $\bfA^{(d)}\bsGamma_k=\bfzero_{1\times L_k}$ for every shape term $k$. Consequently, $\bfone^\top\bseta_\delta(\bsx)=0$ for the complete shape predictor. Because every equivalence class generated by a shared shift contains exactly one zero-sum representative, this constraint ensures identifiability without restricting the induced distributions.

\begin{figure}[tb]
    \centering
    \includegraphics[width=\linewidth]{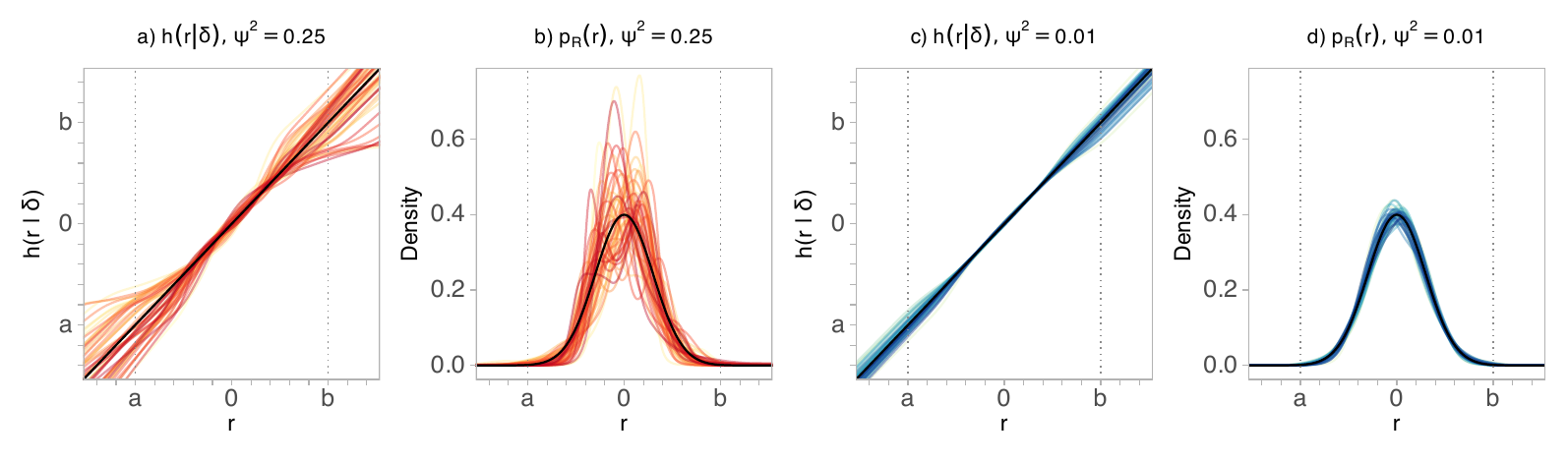}
    \caption{
        Prior samples for the standardized onion-spline transformation with
        $D=8$ and the induced latent-variable density. Panels (a) and (b) show 50 standardized
        transformations and the corresponding latent-variable densities for
        $\psi^2=0.25$; panels (c) and (d) show the same quantities for
        $\psi^2=0.01$.
        The black lines denote the identity transformation or the standard
        Gaussian reference density, and the dotted lines mark the boundaries
        $a=-10/3$ and $b=10/3$ of the spline core.
    }
    \label{fig:onion-spline-prior}
\end{figure}

\paragraph{Hyperprior for the between-dimension smoothing variance $\psi^2$.}
Following the scale-dependent prior construction of \textcite{Klein2016-ScaledependentPriorsVariance}, we assign each between-dimension smoothing variance a $\operatorname{Weibull}(1/2,\theta)$ prior, which favors shrinkage toward the reference model while retaining a right tail that permits larger deviations. We do not currently calibrate $\theta$ to a prespecified effect-scale probability as in the full scale-dependent construction. For the shape effects considered here, we instead use the pragmatic fixed choice $\theta=0.05$, guided by previous empirical work. To assess its implications, we drew $5{,}000$ samples from the prior for an intercept-only shape predictor with $D=10$, using the scaled first-order random-walk penalty and the
sum-to-zero constraint. For each draw, we standardized the implied density and
computed its total variation distance from the Gaussian base density,
$\operatorname{TV}(f,f_0)=\frac{1}{2}\int |f(r)-f_0(r)|\,\mathrm{d}r$. The
median distance was $0.069$, corresponding to 93.1\% density overlap; the 90th,
95th, and 99th percentiles were $0.203$, $0.258$, and $0.370$, respectively.
Only two of the $5{,}000$ draws exceeded a distance of $0.62$. Thus, for the shape
intercept, the prior is skeptical about large deviations from the base
distribution but retains a tail that permits appreciable deviations. This
prior-predictive check characterizes the implemented choice rather than
providing a formal calibration, and its numerical conclusions need not transfer
unchanged to covariate-dependent shape effects.

A more directly interpretable specification would select $\theta$ under the
actual basis and penalty scaling so that the prior probability of exceeding a
prespecified distributional deviation is controlled, following the
prior-predictive strategy of
\textcite{Brachem2026-BayesianPenalizedTransformation}. For a
covariate-dependent shape effect, this calibration would evaluate the induced
density changes over the relevant covariate domain and account jointly for its
within-covariate and between-dimension variance parameters.

\paragraph{Term-wise density previews and average densities.}
The additive structure of the transformation predictor suggests a useful way to visualize the contribution of an individual term. For a focal term $k$, write the partial predictor
$$
\bsdelta_k(\bsx) = \bseta_{\delta,k} = \bsgamma_0^\delta + \bsf_k^\delta(\bsx),
$$
where $\bsgamma_0^\delta$ is the predictor's multivariate intercept. A term-specific
density preview can then be obtained by evaluating the standardized latent-variable
density using $\bsdelta_k(\bsx)$; that is,
$$
p_{R,k}^*(r\mid \bsx) = p_{R,k}^*(r\mid \bsdelta_k(\bsx)) = p_Z\{h(r \mid \bsdelta_k(\bsx))\}\, h'(r \mid \bsdelta_k(\bsx)).
$$
If the individual terms are centered over their observed covariate values, which is the 
case for many effect classes, this is exactly the density evaluated at the averages of all nuisance terms.

The preview density must nevertheless be distinguished from the density
averaged over the nuisance covariates,
$$
\overline{p}_{R,k}(r\mid\bsx)
=
\int
p_R\{r\mid\bsdelta_k(\bsx)+\bsdelta_{-k}(\bsu)\}
p_{\bsX}(\bsu)\,\mathrm{d}\bsu,
$$
where $p_{\bsX}$ is the joint density of the nuisance covariates.
In general, $p_{R,k}^*(r\mid\bsx)\neq\overline{p}_{R,k}(r\mid\bsx)$ because
the density depends nonlinearly on the shape predictor. In practice, the
expectation is approximated by the sample average.

Two complementary displays make the previews useful for
model inspection. Density previews at selected covariate values show the type
of shape feature favored by an individual term, while standardized quantile
curves can reveal tail behavior over a continuous covariate domain and isolate it
from location and scale effects. Neither display is itself a response-scale prediction:
term-wise previews omit the remaining shape terms, and partial standardized
quantiles must be combined with the complete shape, location, and scale
predictors before they can be interpreted predictively.

\section{Approximating the posterior through stochastic variational inference}\label{sec:inference}

For estimation, all coefficient constraints are first resolved so that inference
is carried out on identifiable real-valued parameters. Positive smoothing variances
are represented on the log scale. In particular, for the within- and
between-dimension variances we write $\ell_{\tau,k}=\log(\tau_k^2)$ and $\ell_{\psi,k}=\log(\psi_k^2)$.
Collecting the free coefficients and log
variances gives the parameter tuple
$
\bstheta
=
\bigl(
\bsgamma_\mu,
\bsgamma_\sigma,
\operatorname{vec}(\bsGamma_\delta),
\bsell_\mu,
\bsell_\sigma,
\bsell_{\delta,\tau},
\bsell_{\delta,\psi}
\bigr)
$.
We transform each smoothing-variance prior to the corresponding log scale,
including its Jacobian factor, and let $p(\bstheta)$ denote the resulting joint
prior density on the free, unconstrained scale. The factorization of
the joint prior before resolving constraints and transforming the smoothing
variances is given in \autoref{app:mv-star-predictors}. Writing $\bsy = [y_1, \dots, y_N]^\top$
and assuming conditional independence gives the
unnormalized joint posterior
$$
    p(\bstheta\mid\bsy)
    \propto
    \left[
        \prod_{i=1}^N
    p_Y(y_i\mid\bsx_i;\bsgamma_\mu, \bsgamma_\sigma, \bsGamma_\delta)
        \right]
    p(\bstheta),
$$
where each likelihood contribution depends on the parameters through $\mu(\bsx_i) = \bfb_\mu(\bsx_i)^\top \bsgamma_\mu$,
$\log\sigma(\bsx_i) = \bfb_\sigma(\bsx_i)^\top \bsgamma_\sigma$, and
$\bsdelta(\bsx_i) = \bsGamma_\delta \bfb_\delta(\bsx_i)$ as described in
\autoref{sec:model}.

The covariate-dependent transformation makes posterior computation substantially
more demanding than for a conventional location--scale model. Even intercept-only
PTMs become difficult to scale beyond moderate sample sizes
\parencite{Brachem2026-BayesianPenalizedTransformation}. Conditional-shape PTMs increase both the
data requirements and the per-observation cost: covariate-dependent shape effects
typically require more observations for reliable estimation, while every likelihood contribution
entails two numerical integrations. This makes
full-data MCMC computationally unattractive. In pilot experiments, MCMC schemes
based on iteratively weighted least-squares proposals or Hamiltonian Monte Carlo updates
mixed slowly for the shape-predictor
coefficients, while the No-U-Turn Sampler explored the posterior particularly slowly. These results
suggest that expensive likelihood evaluations are compounded by posterior geometry
that is difficult for generic samplers to navigate.

We therefore approximate the posterior by variational inference (VI) using
stochastic optimization with mini-batches. VI replaces the intractable posterior with a tractable family
of candidate distributions, and optimization selects the member that
best approximates the posterior.
For this family, we use Gaussian blocks with
different covariance structures for the model's location, scale, and shape components;
the resulting approximation is fitted with a staged optimization scheme. Related work
on additive and distributional regression spans several variational constructions.
\textcite{Kleinemeier2023-ScalableEstimationStructured}
combine a factor-covariance Gaussian approximation for the
regression coefficients with exact conditional posteriors for smoothing
variances under inverse-gamma priors;
\textcite{Lichter2024-VariationalInferenceUncertainty}
compare block-diagonal and full mean-field approximations with
semi-implicit VI and a hybrid optimization scheme; and
\textcite{Callegher2025-StochasticVariationalInferencea} derive Gaussian
approximations for the regression coefficients from local Gaussian
approximations to the additive predictors, treating the smoothing parameters
either by point estimation or variationally.

\subsection{Stochastic variational objective}

Let $q_{\bsphi}(\bstheta)$ be a tractable density on the free, unconstrained
parameters defined above. Whereas $\bstheta$ denotes the model parameters,
$\bsphi$ indexes a distribution used to approximate the posterior over
$\bstheta$. In this work, we use the common evidence lower bound (ELBO)
objective, where $\operatorname{KL}$ denotes the Kullback--Leibler divergence,
and determine $\bsphi$ by maximizing the following criterion:
\begin{align}
    \operatorname{ELBO}(\bsphi)
     & = \bbE_{q_{\bsphi}}\!\left[
        \log p(\bsy\mid\bstheta)
        \right]
    - \operatorname{KL}\!\left\{
    q_{\bsphi}(\bstheta)\,\|\,p(\bstheta)
    \right\} \notag \\
     & = \log p(\bsy)
    - \operatorname{KL}\!\left\{
    q_{\bsphi}(\bstheta)\,\|\,p(\bstheta\mid\bsy)
    \right\}.
    \label{eq:vi-elbo}
\end{align}
The first line shows that the ELBO balances expected log-likelihood against divergence from the variational approximation to the prior.
The second line shows that maximizing Equation~\eqref{eq:vi-elbo} minimizes
the reverse Kullback--Leibler divergence to the posterior. Here, $\log p(\bsy)$
is the log marginal likelihood, also called the log model evidence;
it does not affect the optimization because it is constant in $\bsphi$.

For the PTM, the ELBO is not available in closed form, and evaluating all $N$
costly likelihood contributions at every update would defeat scalability. We
therefore use two independent approximations at each update. First, we draw $S$
parameter vectors from $q_{\bsphi}$ to approximate the expectation over
$\bstheta$. Second, we sample a mini-batch
$\mathcal{I}\subset\{1,\ldots,N\}$ of size $m$ to approximate the full-data
log-likelihood for each parameter draw:
\begin{align}
    \widehat{\operatorname{ELBO}}_{\mathcal{I},S}(\bsphi)
    = \frac{1}{S}\sum_{s=1}^{S}\left[
        \frac{N}{m}\sum_{i\in\mathcal{I}}
        \log p\bigl(y_i\mid\widetilde{\bstheta}^{\bsphi}_s\bigr)
        + \log p\bigl(\widetilde{\bstheta}^{\bsphi}_s\bigr)
        - \log q_{\bsphi}\bigl(\widetilde{\bstheta}^{\bsphi}_s\bigr)
        \right],
    \qquad
    \widetilde{\bstheta}^{\bsphi}_s=\mathcal{T}_{\bsphi}(\bszeta_s).
    \label{eq:vi-stochastic-elbo}
\end{align}
Here, $\bszeta_s$ is a parameter-free standard Gaussian random
vector, and $\mathcal{T}_{\bsphi}$ is a differentiable sampling map chosen so
that $\widetilde{\bstheta}^{\bsphi}_s\sim q_{\bsphi}$. Under uniform
subsampling, the factor $N/m$ makes the mini-batch sum unbiased for the full-data
log-likelihood. The prior and variational-density terms are not subsampled and
therefore remain unscaled. Together with the independent parameter draws, this
makes Equation~\eqref{eq:vi-stochastic-elbo} an unbiased estimator of the ELBO.
Unbiasedness does not ensure low variance, however, so $m$ must be large enough
for mini-batches to represent the distributional features that identify shape
effects.

Variational sampling and data subsampling provide two distinct sources of
randomness; hence, our approach can be described as doubly stochastic VI
\parencite{Hoffman2013-StochasticVariationalInference,Titsias2014-DoublyStochasticVariational}.
Pathwise derivatives of Equation~\eqref{eq:vi-stochastic-elbo} are evaluated by
automatic differentiation, following the general automatic differentiation
variational inference (ADVI) construction
\parencite{Kucukelbir-AutomaticDifferentiationVariational}. Reparameterization
separates the parameter-free random draws $\bszeta_s$ from $\bsphi$, allowing
$\nabla_{\bsphi}$ to pass through the sampled log target. The resulting Monte
Carlo gradient can therefore be used with standard stochastic-gradient
optimizers without deriving model-specific score estimators.

\subsection{Blocked Gaussian variational family}

The stochastic estimator above introduced the sampling map; we now specify the
blocked variational distribution that it induces.
We partition $\bstheta$ into $B$ non-overlapping vectors $\bstheta_b$ of
dimension $P_b$ and use independent Gaussian blocks,
\begin{equation}
    q_{\bsphi}(\bstheta)
     = \prod_{b=1}^{B}
    \calN\bigl(\bstheta_b\mid\bsnu_b,\bsOmega_b\bigr).
    \label{eq:vi-blocked-family}
\end{equation}
Here, we let $\calN(\,\cdot\mid\bsnu,\bsOmega)$ denote the
multivariate Gaussian density with mean $\bsnu$ and covariance matrix
$\bsOmega$.
Depending on the block, the covariance is parameterized as
\begin{equation}
    \bsOmega_b =
    \begin{cases}
        \operatorname{diag}\!\left\{
        \exp(2\lambda_{b,1}),\ldots,\exp(2\lambda_{b,P_b})
        \right\},
         & \text{diagonal Gaussian block},              \\[2mm]
        \bsLambda_b\bsLambda_b^\top,
         & \text{dense/full-covariance Gaussian block}.
    \end{cases}
    \label{eq:vi-covariance-options}
\end{equation}
For a diagonal block, $\lambda_{b,j}$ is the log standard deviation, so
$\bsLambda_b$ is diagonal with entries $\exp(\lambda_{b,j})$,
$j=1,\ldots,P_b$. Consequently, the variances are
$\exp(2\lambda_{b,j})$. In the dense case,
$\bsLambda_b$ is lower triangular with a positive diagonal obtained from
unconstrained variational parameters. For draw $s$, we independently sample
$\bszeta_{s,b}\sim\calN(\bfzero,\bfI_{P_b})$ for $b=1,\ldots,B$ and stack
these block vectors as
$\bszeta_s=(\bszeta_{s,1}^\top,\ldots,\bszeta_{s,B}^\top)^\top$. The corresponding sampling map
acts blockwise as
$\mathcal{T}_{\bsphi}(\bszeta_s)_b
=\bsnu_b+\bsLambda_b\bszeta_{s,b}$. The diagonal option in
Equation~\eqref{eq:vi-covariance-options} uses $P_b$ covariance parameters and a total of
$2P_b$ variational parameters including the mean. A dense block uses
$P_b(P_b+1)/2$ Cholesky parameters plus $P_b$ means, and hence
$P_b(P_b+3)/2$ variational parameters in total. The dense covariance also
requires quadratic storage and a triangular matrix-vector product for each
draw, which becomes costly for large $P_b$ and incurs additional cost when
$S$ is increased to reduce the variance of the ELBO estimator.

\begin{figure}[tbp]
    \centering
    \includegraphics[width=\linewidth]{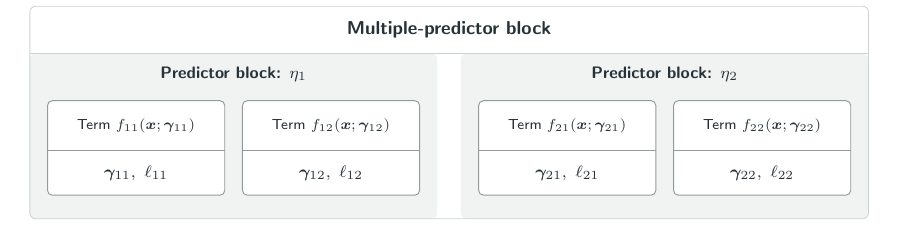}
    \par\medskip
    \includegraphics[width=\linewidth]{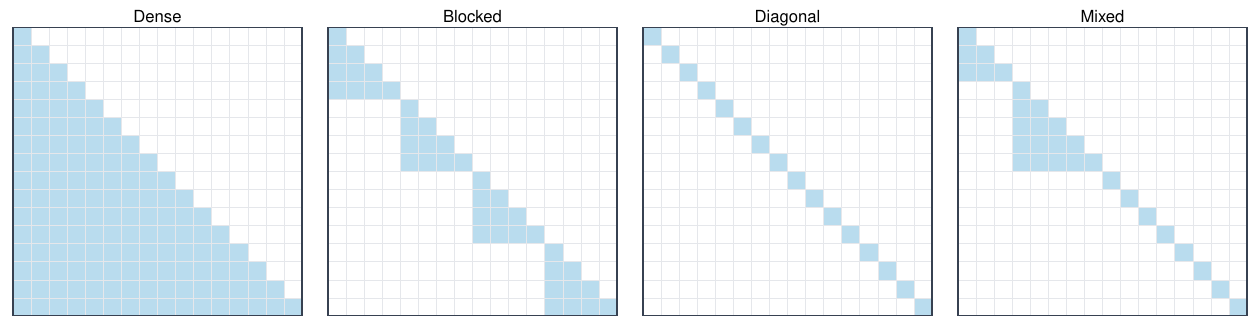}
    \caption{Illustration of blocking and sparsity patterns. The upper panel
        shows model-aligned groups: term blocks couple the coefficients
        $\bsgamma$ of one model term with its log smoothing variance $\ell$;
        predictor blocks combine terms within one predictor; and a
        multiple-predictor block additionally permits dependence across
        predictors. The lower panel shows sparsity patterns for the
        lower-triangular Cholesky factor. Blue cells are free parameters and
        white cells are fixed at zero.}
    \label{fig:vi-blocking-structures}
\end{figure}

The blocking structure mediates a trade-off between the expressiveness of the
approximation and
its computational cost. We use dense blocks for moderate-dimensional groups
whose parameters are expected to be strongly related, such as the parameters within
a predictor or a single term. Large random-intercept and tensor-product terms
may instead receive diagonal blocks to reduce the storage and computational burden.
The exact blocking structure is
part of the model specification for each application. This rule is a
model-specific compromise between
a global dense approximation and complete mean-field factorization, consistent
with the broader use of structured Gaussian approximations in additive models
\parencite{Kleinemeier2023-ScalableEstimationStructured,Lichter2024-VariationalInferenceUncertainty,Callegher2025-StochasticVariationalInferencea}.
Common options for blocking and sparsity patterns are illustrated in
\autoref{fig:vi-blocking-structures}.

The factorization excludes dependence between blocks, even when both blocks are
dense. It can therefore understate uncertainty when cross-block association is
important. Gaussian blocks also cannot directly represent nonlinear dependence or skewed, heavy-tailed,
or multimodal posteriors, and they may be unreliable under weak identification.
These are approximation restrictions rather than properties of the PTM
posterior itself.

\subsection{Staged initialization and optimization}

The model's structured separation into location--scale and shape components
allows us to construct the variational approximation in progressively richer
stages. Because the transformation standardization separates shape from
location and scale, the less expensive location--scale parameters can be fitted
before activating the higher-dimensional shape predictor, whose updates are
computationally more expensive. The procedure illustrated in
\autoref{fig:vi-staged-optimization} therefore concentrates computation on the
parameters currently being introduced. In the reported fits, the shape-only VI
stage often already produced an approximation close to the final result, so the
joint stage required only a short refinement.

We initialize $\bsGamma_\delta = \bfzero$, which
yields the identity transformation and hence the reference density. Conditional
on this reference shape, a maximum a posteriori (MAP) fit establishes a warm
start for the location--scale parameters. A second MAP fit then optimizes the
shape parameters conditional on that solution. We use the resulting point estimates as initial values
for the variational block means.
\begin{figure}[tbp]
    \centering
    \includegraphics[width=\linewidth]{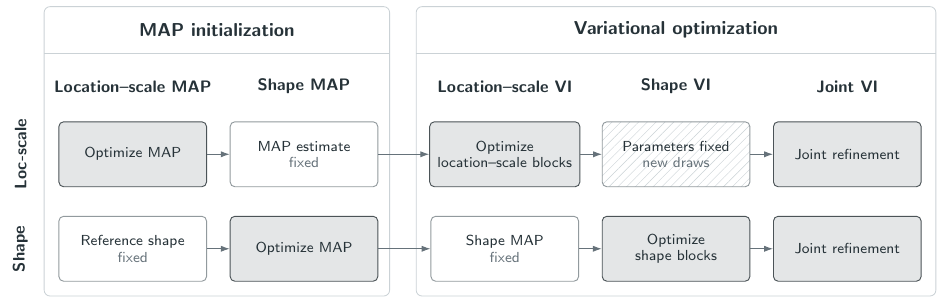}
    \caption{Staged initialization and optimization of the blocked variational
        approximation. Shaded boxes denote the parameter group currently
        optimized, unshaded boxes denote fixed point estimates, and the hatched box
        denotes a fixed variational block from which new draws are sampled.
        Between MAP initialization and variational optimization, temporary
        smoothing-variance bounds are removed and the MAP estimates initialize
        the unrestricted variational family. The final stage updates all blocks
        simultaneously while retaining the prescribed block factorization.}
    \label{fig:vi-staged-optimization}
\end{figure}

During the location--scale-only variational stage, we keep the shape parameters
fixed at their MAP estimates. In the subsequent shape-only stage,
the location--scale variational parameters are held fixed, but draws from these blocks do enter
every stochastic ELBO evaluation. This practice already propagates marginal
location--scale uncertainty into the shape objective before joint refinement.
The final joint stage updates all variational parameters simultaneously.
Location--scale and shape
covariance can be represented only when their parameters share a block; the
separate blocks used in our fits do not estimate such cross-block covariance.

In early pilot fits, this staged sequence produced not only faster but also
more stable optimization
than updating all variational parameters jointly from the outset. When a
variational block is first activated, we initialize its covariance as a diagonal
matrix with small marginal variances. For dense blocks, the off-diagonal
Cholesky parameters therefore start at zero and are learned during optimization.
This initialization also produced stable optimization paths in our pilot fits.

We optimize the negative stochastic ELBO estimate with Adam and a learning-rate
schedule. Every optimizer update uses a new mini-batch and independently samples
new standard-normal draws for the variational approximation. We group successive
mini-batch updates into epochs, with one epoch comprising one pass through the
training data. The epoch-level training ELBO averages the stochastic estimates
over the mini-batches in that epoch. Early stopping terminates a stage when its
relative improvement remains below the user-specified tolerance for the
user-specified number of consecutive epochs. This patience requirement prevents
an isolated noisy epoch from terminating a stage. Batch size, number of Monte
Carlo draws, learning-rate schedule, patience, and block membership are fit-specific tuning
choices whose appropriate values depend on sample size and model dimension.
In practice, we inspect whether the ELBO trace plateaus over the final patience
window and whether relevant posterior summaries remain stable across repeated
fits or initializations. Persistent drift or disagreement indicates that the
optimization schedule or run length should be revisited.

After fitting, draws from $q_{\bsphi}$ are transformed back to the model scale
and used to compute posterior means, standard deviations, intervals, and predictive
distribution summaries.

\subsection{Software implementation.}
The blocked stochastic variational inference described in this section is
implemented in the Python-based probabilistic programming framework Liesel \parencite{Riebl2026-LieselPythonFramework}. The computations in this paper use development
versions of Liesel, Liesel-GAM, and Liesel-PTM
\parencite{Wiemann2026-LieselProbabilisticProgramming,
Brachem2026-LieselgamGeneralizedAdditive,
Brachem2026-LieselptmPenalizedTransformation}.

\section{Simulation studies}\label{sec:simulations}

\subsection{Study design}

The simulation studies address two questions. Scenario A is a
controlled recovery experiment in which the response is generated from the
location--scale decomposition assumed by the PTM. It tests recovery of smooth
location and scale effects and of a standardized latent-variable distribution
whose shape may depend on a covariate. It also tests regularization when a
covariate effect on shape is absent. Scenario B generates the
response directly from a complex mixture of Gaussians. It is intended as
a predictive stress test under misspecification because the fitted PTM imposes
restrictions detailed below.
The data-generating processes (DGPs) are summarized below, with further details given in \autoref{app:simulation-dgps}. We use training sample sizes $N\in\{500,2{,}000,10{,}000\}$ and an independent test
sample of size $2{,}000$ in each of $100$ replications.

\paragraph{Scenario A: Effect recovery and latent-variable-shape calibration.}
We generate data directly from the model $Y_i = \mu(\bsx_i) + \sigma(\bsx_i)R_i$, $i=1, \dots, N$,
where $\bsx_i$ consists of four independent covariates from $\operatorname{U}(-2,2)$.
The location and scale predictors combine oscillating, U-shaped, linear, and bell-shaped effects. We consider two standardized-residual subscenarios. In the first, $R_i$ is drawn from a standard Gaussian distribution.
In the second, the covariate $x_1$ controls a smooth transition from a
skew-normal distribution to a bimodal Gaussian mixture. Both
component distributions are standardized to mean zero and variance one, so that the generated
$\mu(\bsx)$ and $\sigma(\bsx)$ remain the conditional mean and standard
deviation even though the shape changes with $x_1$.

\paragraph{Scenario B: Complex-mixture stress test.}
We generate six independent covariates on $[-1,1]$ and draw $Y_i$ from a
four-component Gaussian mixture. The first mixing function depends nonlinearly on $x_1$,
$x_2$, and their interaction, while a second mixing function depends on $x_4$. Component
means depend on $x_1$, $x_2$, $x_3$, and $x_4$, and their standard deviations
depend on $x_1$, $x_2$, and $x_4$; $x_5$ and $x_6$ are irrelevant distractors.
This deliberately complex setup produces densities with skewness, shoulders,
multiple modes, and comparatively broad tails over the examined finite range. Because $Y_i$ is generated directly from this mixture, we evaluate
Scenario B directly on the response scale. The fitted PTM is deliberately
misspecified in this scenario: the omitted $x_1x_2$ interaction in a mixture
weight can induce nonadditivity in the conditional mean, standard deviation,
and standardized shape, whereas the reported predictors for the mean, log
standard deviation, and standardized shape are additive; its finite
transformation basis and fixed Gaussian asymptotic tail class impose further
restrictions.

\paragraph{Fitted models.}

The conditional-shape PTM contains single-covariate P-splines for every covariate in the
location, log-scale, and transformation predictors. We use 20 basis functions
for each location and log-scale effect and nine for each transformation effect;
all P-splines use second-order difference penalties. The onion-spline core is
$[a,b]=[-4,4]$, and we use $D=10$ transformation parameters with first-order
random-walk penalties across the parameters. The smoothing variances for the
location and log-scale effects have $\operatorname{InverseGamma}(1,0.005)$
hyperpriors, while the within-covariate and between-parameter smoothing
variances in the transformation predictor have
$\operatorname{Weibull}(1/2,0.05)$ hyperpriors.

We compare the conditional-shape PTM with a fixed-shape PTM that retains a flexible but
covariate-independent transformation, a Gaussian location--scale model, and
DDPstar, a 20-component truncation of a Dirichlet-process mixture of normal
structured additive regression models
\parencite{Rodriguez-Alvarez2024-DensityRegressionDirichleta}. DDPstar was also
one of the strongest density-regression competitors in the original PTM study
\parencite{Brachem2026-BayesianPenalizedTransformation}. The fixed-shape PTM is the exact submodel obtained by retaining the flexible shape intercept but omitting every covariate-dependent shape term, and it represents the structure of the earlier PTM. The conditional-shape PTM, fixed-shape PTM,
and Gaussian model are fitted using the variational procedure described in
\autoref{sec:inference}; DDPstar is fitted by MCMC. In Scenario A, the
predictive-accuracy and calibration figures show the variational conditional-shape PTM fit
together with MCMC fits of the three competitors; the corresponding
variational fits of the fixed-shape PTM and Gaussian model were similar and are
omitted. The effect-recovery comparison retains both variational and MCMC fits
of the fixed-shape PTM and Gaussian model. In Scenario B, the conditional-shape PTM, fixed-shape PTM, and
Gaussian model use VI, whereas DDPstar uses MCMC. Pilot MCMC fits of the conditional-shape PTM
mixed slowly and were not pursued further.

\paragraph{Evaluation.}

For Scenario A, effect recovery is measured by the posterior expected mean
squared error, averaged over an evaluation grid, and by coverage and width of
90\% pointwise credible intervals. The quality of the full-distribution fit
is evaluated on an
independent test set using the Kullback--Leibler divergence from the true
conditional density and the posterior-expected absolute deviation between fitted and true
conditional CDFs (labelled MAD). As a training-data model-selection criterion, we additionally
report Pareto-smoothed importance-sampling leave-one-out cross-validation (PSIS-LOO)
\parencite{Vehtari2017-PracticalBayesianModel}, computed exclusively
from the training-data likelihood contributions and reported on the deviance
scale divided by the training-sample size. We also computed the
Watanabe-Akaike information criterion (WAIC) as a sensitivity check but do not report it
because it yields the same conclusions as PSIS-LOO. Lower values are
preferable for both test-set measures and PSIS-LOO. Pointwise credible intervals
for the CDF quantify posterior uncertainty about the fitted conditional CDF.

\subsection{Results}

\begin{figure}[tb]
    \includegraphics[width=\linewidth]{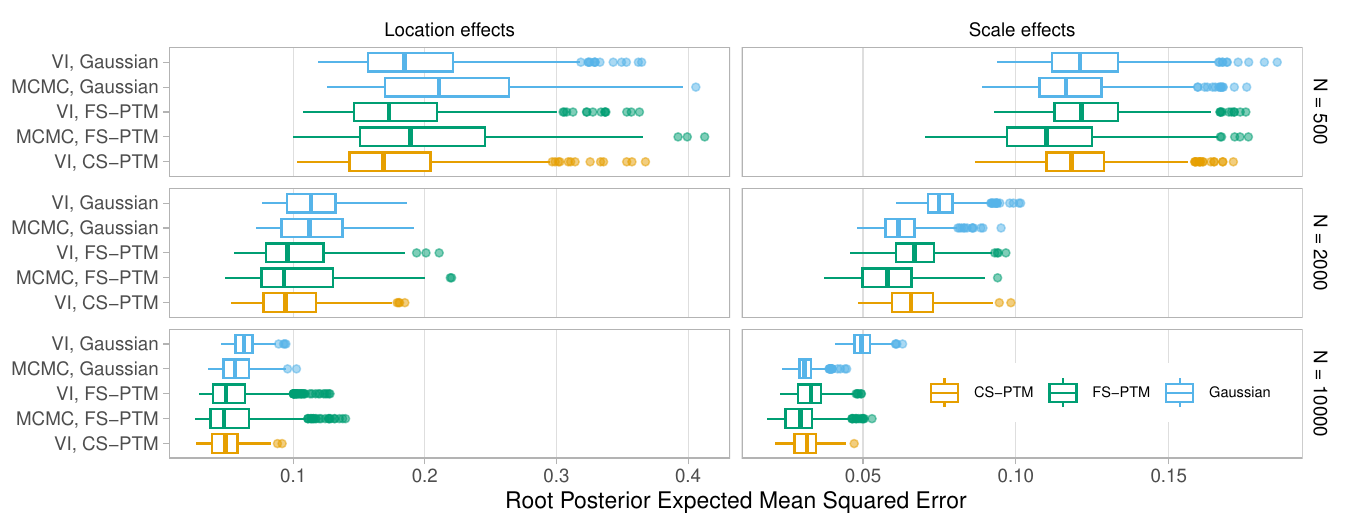}
    \vspace{0.5em}
    \includegraphics[width=\linewidth]{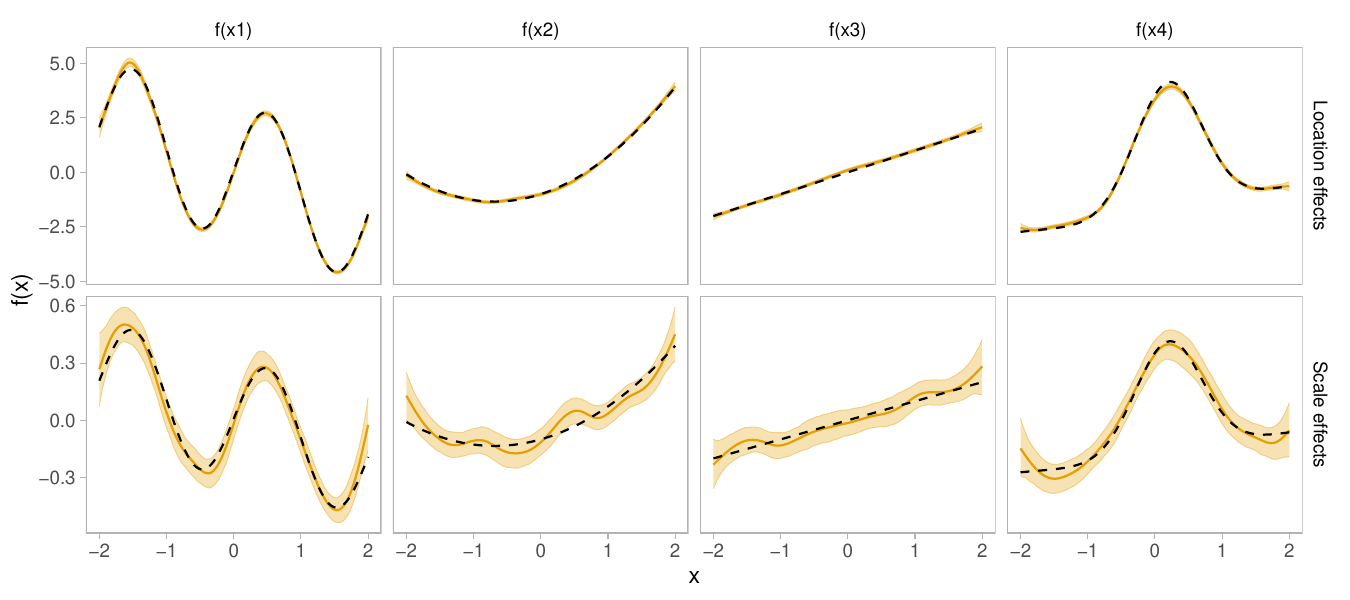}
    \caption{Recovery of the smooth location and log-scale effects in Scenario
        A. Top: root posterior expected mean squared errors by sample size. The
        distributions pool the four covariate effects, the Gaussian and
        skew--mixture subscenarios, and 100 replications. The conditional-shape PTM (CS-PTM) is
        fitted by VI; the fixed-shape PTM (FS-PTM) and Gaussian model are fitted by both VI
        and MCMC. Bottom: representative recovery of the four location effects
        and log-scale effects at $N=2{,}000$. Solid colored lines are posterior
        means, shaded bands are 90\% pointwise credible intervals, and dashed
        black lines are the true functions.}
    \label{fig:sim-effect-recovery}
\end{figure}

\paragraph{Recovery of location and scale effects.}
The smooth response-level effects are recovered well. The top panel of
\autoref{fig:sim-effect-recovery} shows that the root posterior expected mean
squared error decreases steadily with sample size for every model. The
variational and MCMC fits of the fixed-shape PTM and Gaussian model show broadly
similar recovery. In the representative fit shown in the
bottom panel, the posterior means from the PTM closely follow all four location and all four scale functions at $N=2{,}000$. Thus, the PTM successfully recovers
the response-level location and scale effects while also estimating
covariate-dependent shape.

\paragraph{Predictive distributional accuracy.}
\autoref{fig:sim-performance} summarizes the predictive-distributional-accuracy results. 
In the Gaussian case of Scenario A, the conditional-shape PTM, fixed-shape PTM, and Gaussian
model perform almost identically. This indicates
that the additional transformation terms are successfully regularized when
they are unnecessary. In the skew-mixture case, the conditional-shape PTM improves most
rapidly with sample size and has the smallest Kullback--Leibler divergence and
CDF deviation at every sample size, as well as the smallest normalized
PSIS-LOO value from $N=2{,}000$ onward. The fixed-shape PTM improves on the
Gaussian model but cannot reproduce the transition in shape across $x_1$.
DDPstar struggles most at $N=500$ on the test-set criteria. It improves with
sample size and approaches the fixed-shape PTM by $N=10{,}000$, but remains behind the
conditional-shape PTM.

\begin{figure}[tb]
    \includegraphics[width=\linewidth]{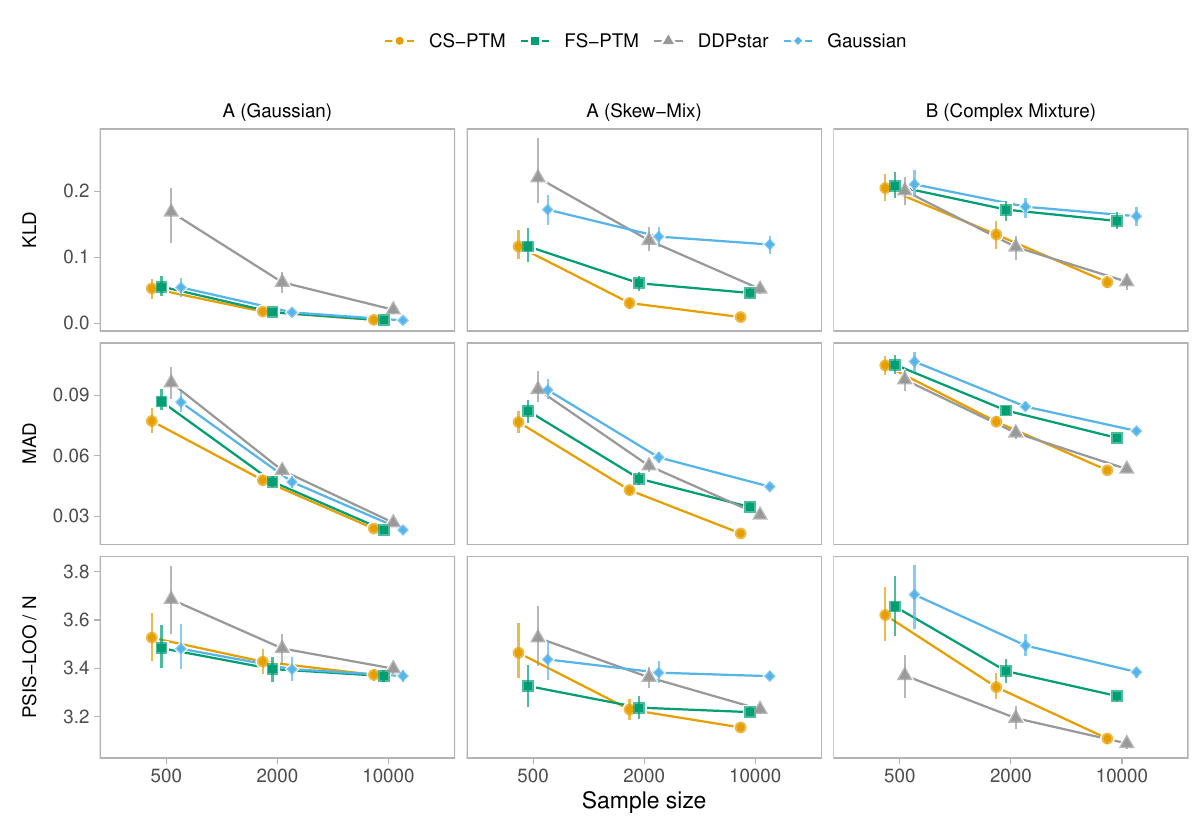}
    \caption{Response-scale distributional accuracy. Points are means over
        100 replications and vertical lines span the 0.1 and 0.9 quantiles.
        PSIS-LOO is reported on the deviance scale and divided by $N$; lower
        values are preferable in every row. In Scenario A, CS-PTM uses VI and
        the three competitors use MCMC; the similar VI fits of the FS-PTM
        and Gaussian model are omitted. In Scenario B, CS-PTM, FS-PTM, and
        Gaussian use VI, whereas DDPstar uses MCMC. CS-PTM and FS-PTM
        denote the conditional-shape and fixed-shape PTMs, respectively.}
    \label{fig:sim-performance}
\end{figure}

Scenario B is harder for every model. At $N=500$, their Kullback--Leibler
divergences and CDF deviations are similar. With increasing sample size, the
PTM and DDPstar improve substantially, whereas the fixed-shape PTM and
Gaussian model level off at noticeably larger values. DDPstar has substantially
smaller normalized PSIS-LOO values at $N=500$ and $N=2{,}000$. This advantage is
consistent with its mixture structure matching the data-generating process
more directly and, at $N=500$, with its better-behaved PSIS importance
weights. At $N=10{,}000$,
DDPstar and the PTM are approximately tied on Kullback--Leibler divergence and
CDF deviation. Thus, despite its additive transformation predictor, the PTM
yields predictive gains similar to those of flexible mixture-based density
regression. As a sensitivity check, \autoref{fig:sim-performance-nparam}
shows that increasing the transformation dimension from $D=10$ to $D=20$
does not meaningfully change the PTM's performance in Scenario B, supporting
the more economical default.

\paragraph{Diagnostics for PSIS-LOO.}
Pareto-$k$ warnings for the PSIS-LOO metric occurred in several simulation settings. These warnings indicate that Pareto-smoothed importance sampling had difficulty reliably approximating the corresponding leave-one-out posteriors, which may result from influential observations, imperfections in the variational posterior approximation, or both. Consequently, the affected PSIS-LOO estimates should be interpreted cautiously. The fraction of problematic observations and the frequency of affected fits generally decreased with increasing sample size, and warnings were more prevalent under the model misspecification of Scenario B. In sum, greater weight should be placed on the Kullback--Leibler divergence (KLD) and mean absolute deviation (MAD), which assess predictive performance directly using independently generated test data.

\begin{figure}[tb]
    \includegraphics[width=\linewidth]{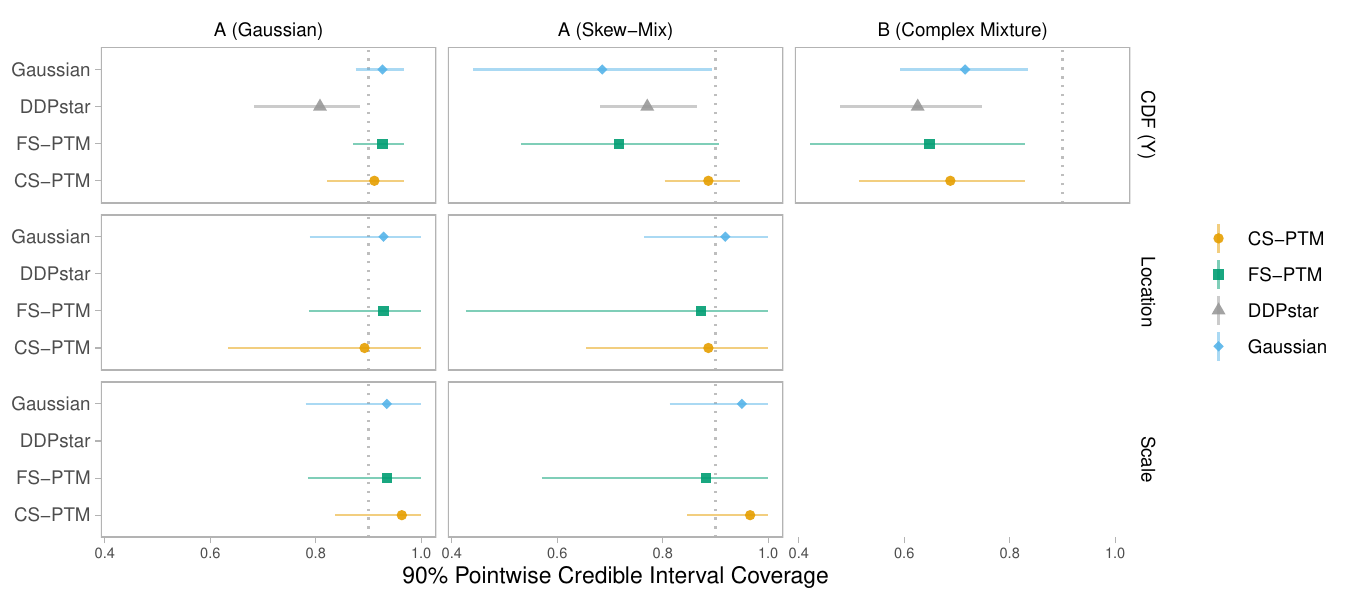}
    \caption{Empirical coverage of 90\% pointwise credible intervals for the
        response CDF and the location and scale effects. Points summarize
        coverage over replications and sample sizes;
        horizontal lines show the corresponding central 90\% range. Empty
        panels indicate quantities without a true or model-specific counterpart.
        Dotted vertical lines mark the nominal coverage level. In Scenario A,
        CS-PTM uses VI and the three competitors use MCMC; the similar VI fits of
        the FS-PTM and Gaussian model are omitted. In Scenario B, CS-PTM, FS-PTM, and Gaussian use VI, whereas DDPstar uses MCMC.
        CS-PTM and FS-PTM denote the conditional-shape and fixed-shape
        PTMs, respectively.}
    \label{fig:sim-calibration}
\end{figure}

\paragraph{Uncertainty calibration.}
The results in \autoref{fig:sim-calibration} indicate that the fitted models
provide reasonably calibrated uncertainty in Scenario A, rather than only
accurate point estimates. In the Gaussian subscenario, the 90\%
pointwise intervals for the response CDF and the location and scale effects
generally have coverage close to the nominal level. Under the skew-mixture
DGP, the PTM has the best CDF coverage of the four models.
Undercoverage is pronounced for all models in Scenario B. Thus, uncertainty
quantification is informative but not uniformly calibrated under more complex
or misspecified data-generating processes. Interval-width summaries, included in the appendix
(\autoref{fig:sim-ci-width}), generally decrease with $N$, as expected. In
Scenario B, the conditional-shape PTM, fixed-shape PTM, and DDPstar have narrower response-CDF
intervals than the Gaussian model despite the undercoverage of all four
models. In Scenario A's skew-mixture subscenario, the conditional-shape and fixed-shape PTMs have
narrower response-CDF and location-effect intervals than the Gaussian model;
the fixed-shape PTM also has narrower scale-effect intervals, while those of the
conditional-shape PTM are similar. Interval width should therefore be interpreted jointly
with coverage rather than as a performance criterion on its own.

\begin{figure}[tb]
    \includegraphics[width=\linewidth]{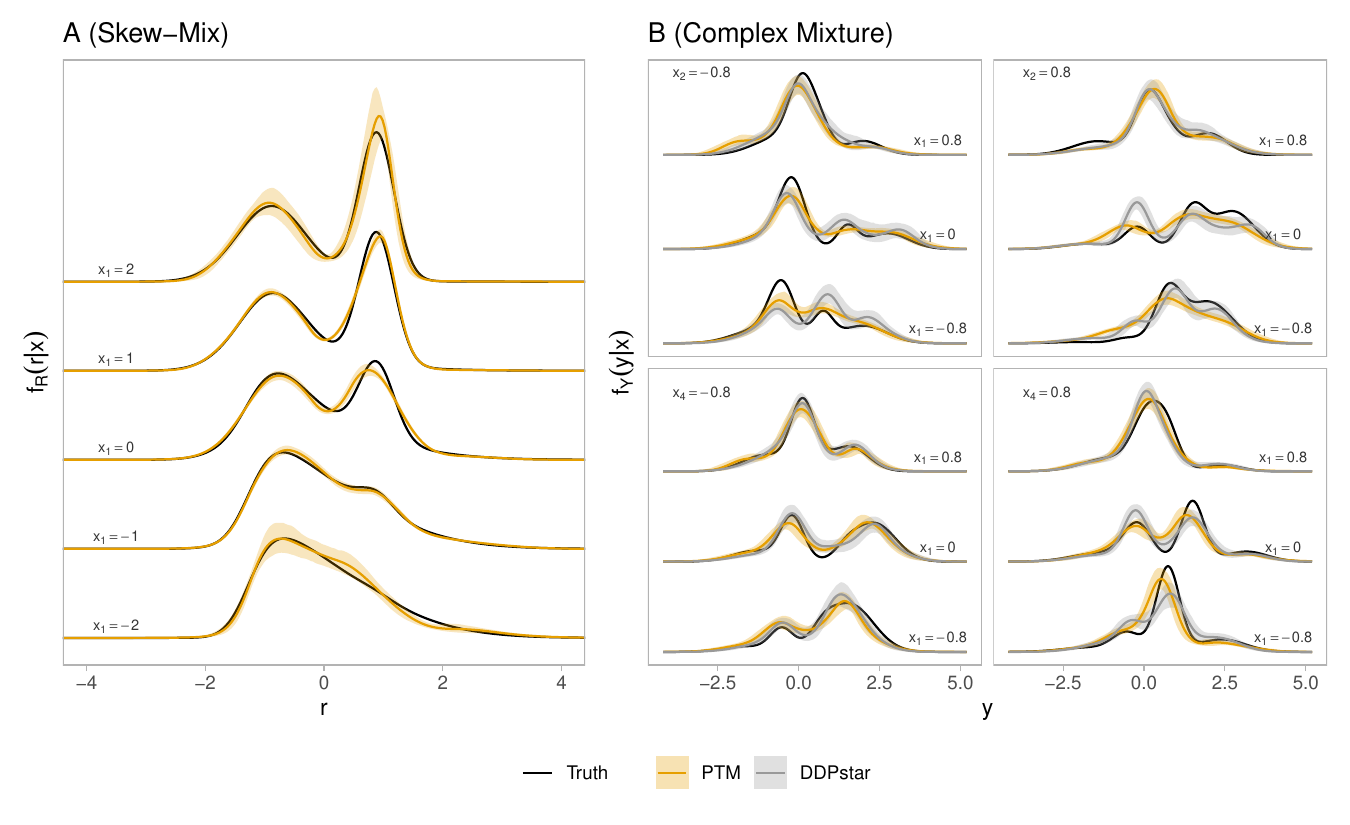}
    \caption{Selected conditional density estimates at $N=10{,}000$. Left:
        standardized-residual densities in the skew-mixture subscenario of
        Scenario A as $x_1$ varies, with the remaining covariates fixed at zero.
        Right: response densities in Scenario B for combinations of $x_1$ with
        $x_2\in\{-0.8,0.8\}$ (top) or $x_4\in\{-0.8,0.8\}$ (bottom), with
        the remaining covariates fixed at zero. Solid colored lines are
        posterior means, bands are 90\% pointwise credible intervals, and
        black lines are the true densities. Curves are vertically offset.}
    \label{fig:sim-densities-demo}
\end{figure}

\paragraph{Graphical diagnostics.}
The conditional density profiles in \autoref{fig:sim-densities-demo} show
strong recovery in Scenario A and a reasonable approximation of the principal
features in the harder Scenario B. In Scenario A, the PTM recovers the
transition from a skewed density to a bimodal density as $x_1$ increases. It
captures the locations and relative prominence of the modes well across the
covariate range. The additional previews in \autoref{fig:sim-densities-controls}
show that the estimated residual-density profiles for $x_2$, $x_3$, and $x_4$ remain highly similar across focal values at each sample size, as expected.

In Scenario B,
both the PTM and DDPstar reproduce the principal movement of modes and
shoulders across the selected $x_1$:$x_2$ and $x_1$:$x_4$ contrasts. Neither
method reproduces every sharp feature of the true mixture, but the PTM tracks
the main qualitative changes with comparatively smooth density estimates. The
lower-sample-size counterparts in \autoref{fig:sim-densities-demo-n500} and
\autoref{fig:sim-densities-demo-n2000} clearly show the sample-size progression.
In Scenario A, the PTM broadly captures the
transition from skewness to bimodality already at $N=500$, but smooths the narrow second peak
for nonnegative $x_1$ into a shoulder. At $N=2{,}000$, it resolves the two modes
much more clearly and follows their locations and relative heights more closely.
In Scenario B, both the PTM and DDPstar
improve from $N=500$ to $N=2{,}000$: their intervals narrow and the principal mode
locations become clearer. Nevertheless, both models struggle with the fine structure
of Scenario B: sharp shoulders and secondary modes are sometimes smoothed or
missed, so the recovery remains incomplete even at larger sample sizes.

Overall, the simulations show that the PTM recovers location, scale, and
covariate-dependent shape under the structured data-generating process and
remains competitive under misspecification, while uncertainty calibration
and fine multimodal-feature recovery remain challenging in the hardest setting.

\section{Applications}\label{sec:applications}

\subsection{Norwegian water conductivity}
\label{sec:application-norway-water}

\subsubsection{Data and scientific context}
We illustrate our model first using electrical-conductivity measurements for Norwegian
rivers and lakes from the European Environment Agency (EEA)
\emph{Waterbase -- Water Quality ICM, 2024}\footnote{ICM denotes inland,
    coastal, and marine waters.} dataset
\parencite{EuropeanEnvironmentAgency2025-Waterbase}. The dataset is licensed
under Creative Commons Attribution 4.0 International (CC BY 4.0). The source extract contains
$13{,}533$ measurements from $370$ monitoring sites collected between 2000 and
2019. Before fitting, we excluded 47 values reported below the analytical limit
of quantification and 62 records carrying the EEA lower-reliability flag
\texttt{U}; one record met both criteria. This filter removed 108
records and left $13{,}425$ observations, of which $10{,}083$ were taken from
rivers and $3{,}342$ from lakes. The \texttt{U} flag indicates lower reliability,
not that a record is necessarily invalid. The monitoring sites represented here belong to 14 river-basin
districts. Each district is a regional water-management area covering land whose
surface water drains into one or more neighboring river systems, together with
the associated groundwater and coastal waters. We use the natural logarithm of conductivity, measured in
$\mu\mathrm{S}/\mathrm{cm}$, as the response.

Electrical conductivity reflects dissolved ionic content and is widely used as
a nonspecific indicator of salinization and water quality. It varies with
geology, hydrological conditions, and anthropogenic inputs
\parencite[see, for example, ][]{Kaushal2021-FreshwaterSalinizationSyndrome}.
Consequently, while unusually high or low measurements cannot by themselves be
attributed to a particular source, they may motivate further
investigation.
Modeling its complete conditional distribution is useful in this setting
because changes in skewness, tails, or modality can alter which observations
are considered unusual even if the conditional mean and variance are similar.
Note, however, that the present analysis is primarily intended to demonstrate our
model rather than to provide a comprehensive environmental assessment.

\subsubsection{Model specification}
Let $\text{wb}_i$ indicate the water-body type, that is, whether observation $i$ was obtained from a river ($\text{wb}_i = 1$)
or a lake ($\text{wb}_i = 0$), let $\text{day}_i$ denote the day of year\footnote{Leap-day observations were retained when defining calendar day of year (1--366).}, let $\text{year}_i$ denote the calendar year,
and let $\text{rbd}_i$ denote the river-basin district. We standardize the
log-conductivity response to mean zero and unit sample standard deviation for
fitting and transform reported effects and predictive quantities back to the
log-conductivity scale. We use the
same covariates in the location, log-scale, and shape predictors,
\begin{alignat*}{6}
    \eta_{t,i}
     & = \gamma^t_0
     & \quad           & + \gamma^t_1 \,\text{wb}_i
     & \quad           & + f^t_2(\text{day}_i)
     & \quad           & + f^t_3(\text{year}_i)
     & \quad           & + f^t_4(\text{rbd}_i),
     & \qquad          & t \in \{\mu,\sigma\},       \\
    \bsdelta(\bsx_i)
     & = \bsdelta_0
     & \quad           & + \bsgamma^\delta_1 \,\text{wb}_i
     & \quad           & + \bsf^\delta_2(\text{day}_i)
     & \quad           & + \bsf^\delta_3(\text{year}_i)
     & \quad           & + \bsf^\delta_4(\text{rbd}_i).
\end{alignat*}
The seasonal effects $f_2^t(\text{day}_i)$ are cyclic P-splines, the long-term trend
$f^t_3(\text{year}_i)$ is a P-spline with a second-order difference penalty, and
$f^t_4(\text{rbd}_i)$ is a discrete spatial effect with a Markov random field penalty
based on neighboring river-basin districts. Their multivariate analogues enter
the transformation predictor. The focal fit uses $D=20$ transformation
parameters and onion-spline core $[a,b]=[-5,5]$; the sensitivity fit with
monitoring-site effects uses $D=10$ and $[a,b]=[-4,4]$. Both use the first-order
random-walk penalty across transformation dimensions described in
\autoref{sec:mvpred-ptm}. Water-body type enters all three predictors
as a linear effect; its shape coefficient $\bsgamma^\delta_1$ is a $D$-vector.
The day-of-year and year effects used basis dimensions 20 and 10 in the
location and log-scale predictors and 20 and 5 in the shape predictor,
respectively.

Following \textcite{Klein2016-ScaledependentPriorsVariance}, we calibrated
$\tau_{k,t}^2\sim\operatorname{Weibull}(1/2,\theta_{k,t})$ priors for the 
smoothing variance parameters to
assign probability $0.01$ to location and log-scale effects exceeding three
response standard deviations and $\log(2)$, respectively. The same calibration
was applied to the monitoring-site effects in the sensitivity model. The within-covariate
variances $\tau_k^2$ and between-transformation-dimension variances $\psi_k^2$
in the shape predictor independently followed
$\operatorname{Weibull}(1/2,0.05)$ priors.
We used a single dense Gaussian block for
the location--scale parameters and a diagonal Gaussian block for the
transformation parameters.

This specification describes the distribution across the sampled sites within a river-basin
district, after adjustment for the other included covariates.
This deliberately pooled specification provides a
clear illustration of discrete spatial shape effects, but repeated
measurements from the same site mean that its regional effects may also reflect
differences in the composition of monitoring sites. We examine this explicitly
in a sensitivity model that adds site-specific random intercepts to the
location and log-scale predictors. 

Additional details on the optimization settings are included in \autoref{app:application-optimization}.

\subsubsection{Results}

The staged variational fit of the PTM took approximately 46 minutes on a MacBook Pro with a 10-core
Apple M1 Pro processor and 32~GB of memory.

\paragraph{Water-body and temporal effects.}
The fitted location and log-scale effects indicate that river measurements tend to
have higher mean log conductivity but a lower conditional standard deviation
than lake measurements, with posterior mean coefficients of $0.299$ (90\%
credible interval: $[0.286, 0.311]$) for location and $-0.346$ (90\% credible
interval: $[-0.362, -0.330]$) for log-scale.
The term-wise density previews in the top-right panel of
\autoref{fig:norway-effects} also indicate marked differences beyond the first
two moments. The river preview is distinctly bimodal, with a
taller mode at negative $r$ and a smaller mode at positive $r$. The lake
preview has a sharper dominant mode just below zero and a longer right
shoulder, with only weak evidence of a secondary mode.

The seasonal location and log-scale effects in the top-left and top-center
panels show a clear annual pattern: mean log conductivity is higher early and
late in the year and lower during much of the intervening period, while the
conditional standard deviation is somewhat larger around the middle of the
year. Separate day-of-year shape previews, however, overlap almost exactly and
provide no visible evidence of a seasonal shape effect (not shown). The year effects on
location and log-scale in the middle-left and middle-center panels are small,
and the middle-right panel suggests only a small change in shape over time.

\paragraph{River-basin effects.}
The most pronounced variation occurs between river-basin districts. The
discrete spatial terms in the bottom row of
\autoref{fig:norway-effects} show drastic regional differences in
location and scale, but also in the shape of the standardized preview-density.
Some districts have comparatively smooth unimodal previews, whereas others
show marked skewness, shoulders, or multiple modes. Note that these are
term-wise density previews: for each district, the shape predictor
is evaluated with the remaining transformation terms held at
zero; this means in particular that the preview densities here implicitly condition
on lake-water.
As \autoref{tab:norway-rbd-summary} shows, the number of observations is uneven across districts and monitoring sites,
so previews for sparsely observed regions should be interpreted with caution.

% Requires \usepackage{booktabs}
\begin{table*}[tbp]
    \centering
    \caption{Summary of observations and monitoring sites by river basin district for the Norwegian water-conductivity data.}
    \label{tab:norway-rbd-summary}
    \small
    \begin{tabular*}{\textwidth}{@{\extracolsep{\fill}}llrrrrr@{}}
        \toprule
        \multicolumn{2}{c}{River basin district} & \multicolumn{2}{c}{Counts} & \multicolumn{3}{c}{Observations per site} \tabularnewline
        \cmidrule(lr){1-2} \cmidrule(lr){3-4} \cmidrule(lr){5-7}
        ID & Name & Observations & Sites & Median & Max. & Min. \tabularnewline
        \midrule
        NO1101 & Moere and Romsdal & 478 & 14 & 20.5 & 96 & 5 \tabularnewline
        NO1102 & Troendelag & $1{,}223$ & 39 & 19 & 234 & 1 \tabularnewline
        NO1103 & Nordland & 603 & 29 & 6 & 213 & 1 \tabularnewline
        NO1104 & Troms & 461 & 28 & 10 & 76 & 1 \tabularnewline
        NO1105 & Finnmark & 300 & 11 & 7 & 219 & 4 \tabularnewline
        NO1106 & Norwegian - Finnish & 689 & 24 & 20.5 & 157 & 7 \tabularnewline
        NO2 & Bothnian Sea & 20 & 3 & 1 & 18 & 1 \tabularnewline
        NO5 & Skagerrak and Kattegat & 106 & 8 & 9 & 27 & 3 \tabularnewline
        NO5101 & Glomma & $2{,}161$ & 43 & 20 & 801 & 1 \tabularnewline
        NO5102 & West-Bay & $2{,}398$ & 50 & 11.5 & 993 & 2 \tabularnewline
        NO5103 & Agder & $1{,}875$ & 58 & 20 & 268 & 2 \tabularnewline
        NO5104 & Rogaland & $1{,}985$ & 29 & 20 & 385 & 3 \tabularnewline
        NO5105 & Hordaland & 470 & 18 & 12 & 136 & 2 \tabularnewline
        NO5106 & Sogn og Fjordane & 656 & 16 & 24.5 & 155 & 3 \tabularnewline
        \bottomrule
    \end{tabular*}
\end{table*}

The Glomma district provides a particularly clear example. Its preview shows a
strong departure from unimodality, in addition to comparatively high location
and scale effects. This pattern is plausibly related to the composition of the
measurements in the district: about 38\% of the Glomma observations were
collected from the Alna River, mostly at a single intensively monitored site in
Oslo. Thus, a
regional density that pools several monitoring sites can become multimodal
when those sites have persistently different conductivity levels. The PTM
makes this heterogeneity visible as a structured shape effect.

\begin{figure}[tbp]
    \centering
    \begin{minipage}[t]{0.32\linewidth}
        \vspace{0pt}
        \includegraphics[width=\linewidth]{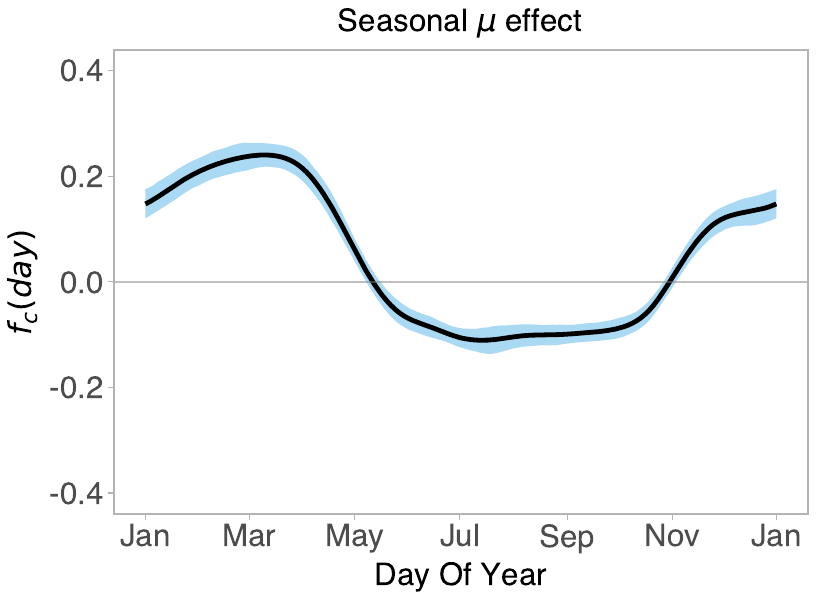}
    \end{minipage}\hfill
    \begin{minipage}[t]{0.32\linewidth}
        \vspace{0pt}
        \includegraphics[width=\linewidth]{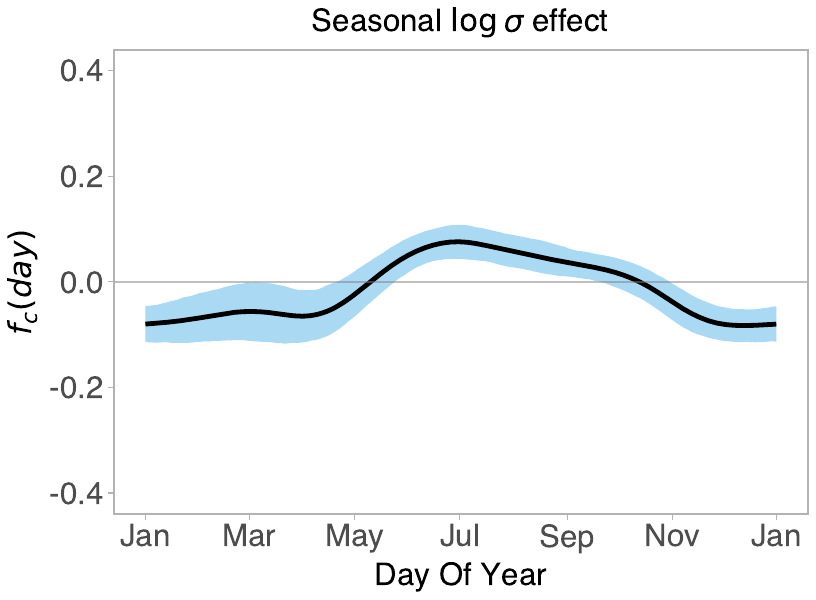}
    \end{minipage}\hfill
    \begin{minipage}[t]{0.32\linewidth}
        \vspace{0pt}
        \includegraphics[width=\linewidth]{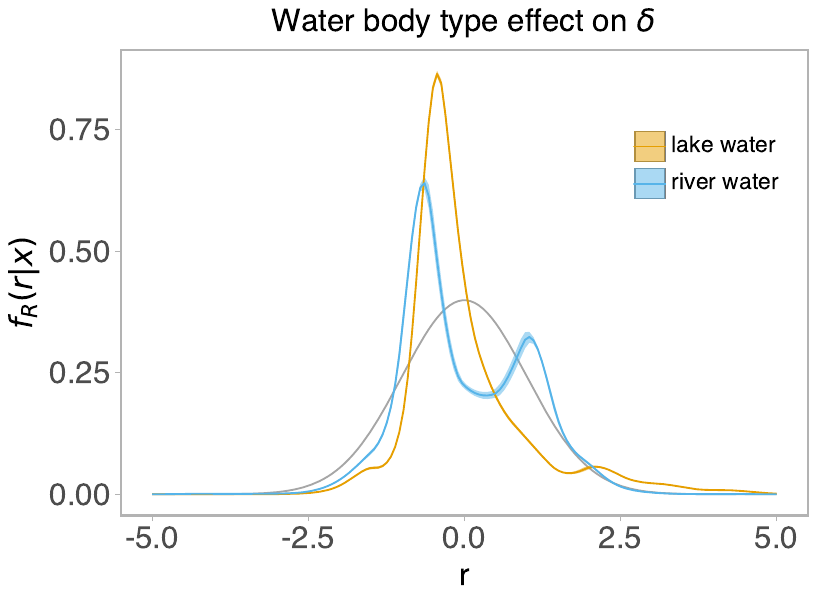}
    \end{minipage}
    \par\vspace{0.75em}
    \includegraphics[width=0.32\linewidth]{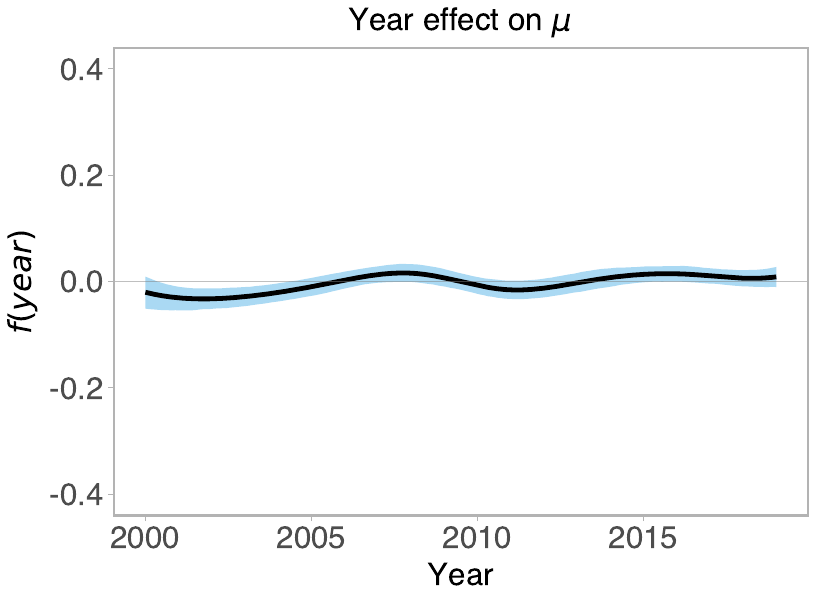}\hfill
    \includegraphics[width=0.32\linewidth]{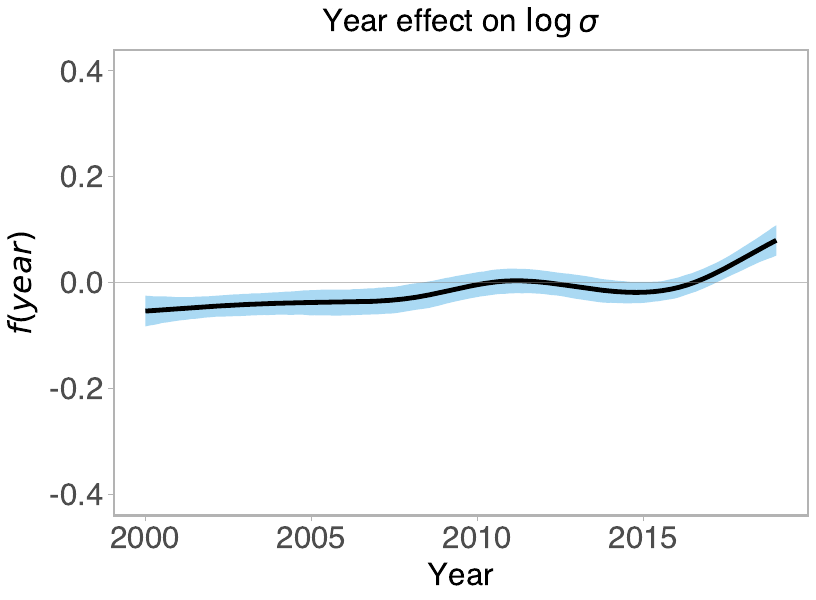}\hfill
    \includegraphics[width=0.32\linewidth]{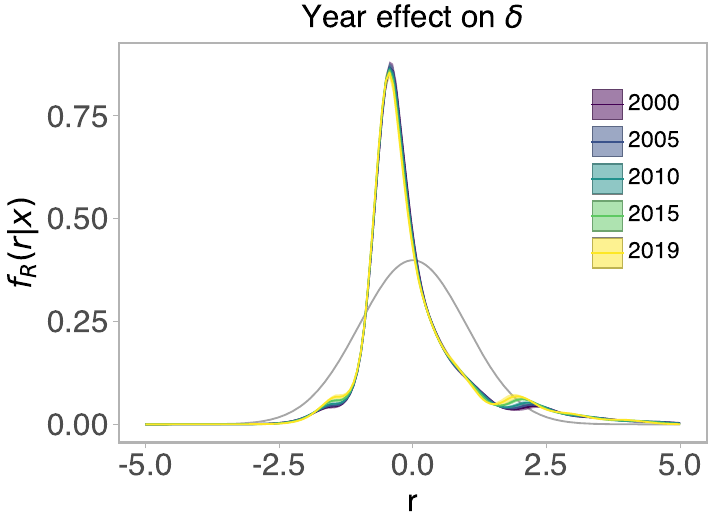}
    \par\vspace{0.75em}
    \includegraphics[width=0.32\linewidth]{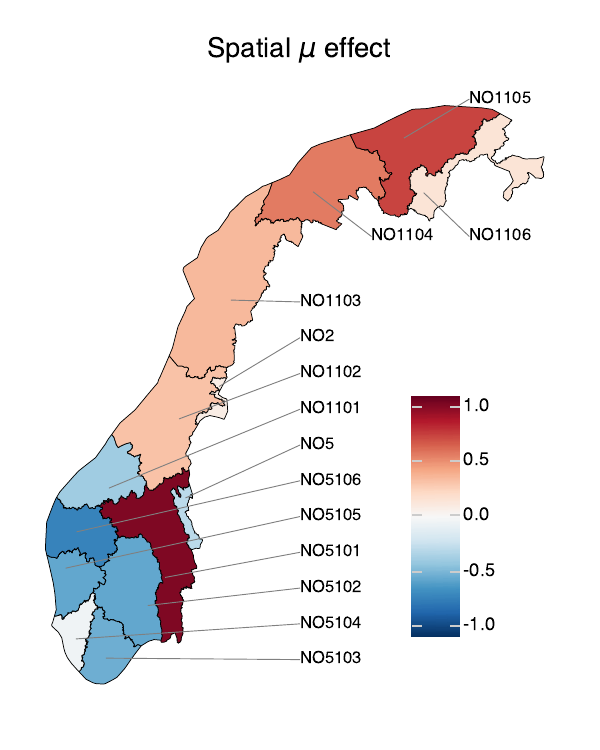}\hfill
    \includegraphics[width=0.32\linewidth]{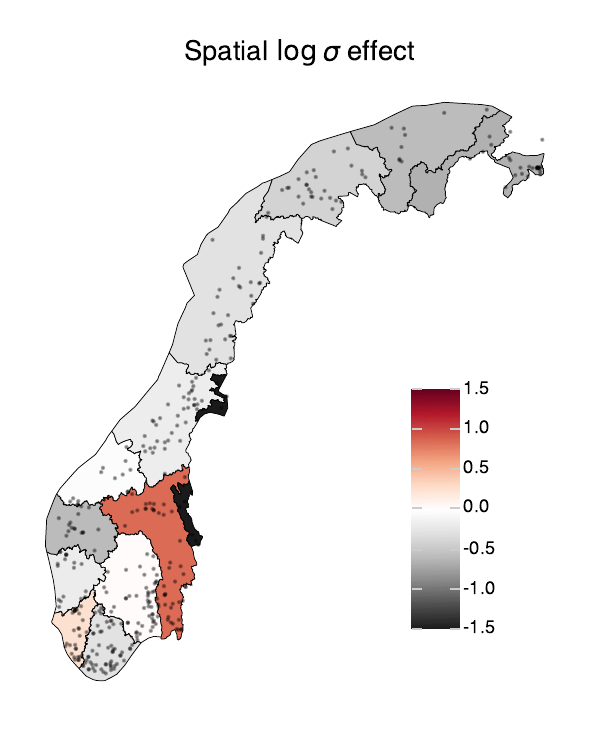}\hfill
    \includegraphics[width=0.32\linewidth]{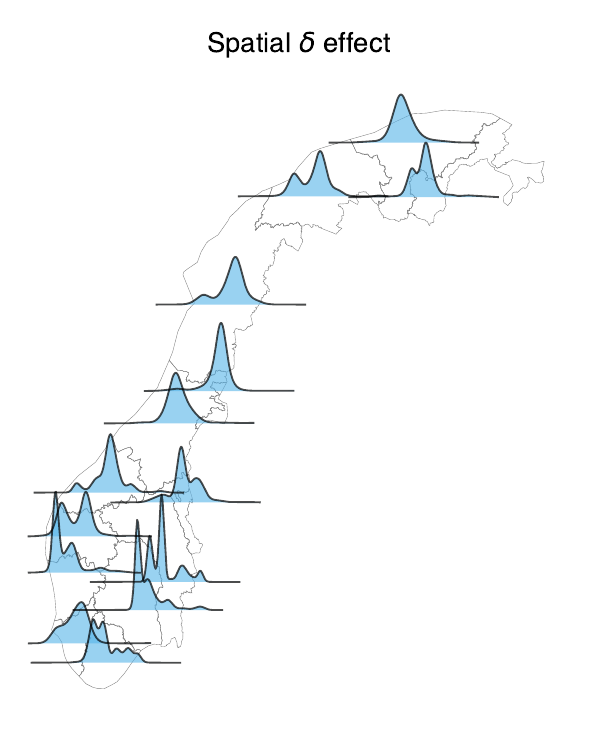}
    \caption{Estimated effects in the PTM with $D=20$ transformation
        parameters, onion-spline core $[a,b]=[-5,5]$, and no monitoring-site
        effects. The top row shows day-of-year effects on location and log-scale
        and the water-body-type effect on shape. The middle row shows year
        effects on location, log-scale, and shape; day-of-year shape previews
        are omitted because their curves overlap almost exactly. The bottom row
        shows river-basin district effects on location, log-scale, and shape.
        The log-scale plot of the bottom row shows the locations of the
        individual monitoring sites as small black dots. Blue shaded ribbons
        are 90\% pointwise credible intervals. Shape effects are
        displayed as term-wise standardized-density previews, with the
        remaining centered transformation terms held at their reference
        values.}
    \label{fig:norway-effects}
\end{figure}

\subsubsection{Sensitivity to monitoring-site effects}
Adding site-specific random intercepts to the location and log-scale predictors
absorbs much of the persistent between-site heterogeneity. In the resulting
fit, shown in
\autoref{fig:norway-effects-site-ri}, the river-basin scale contrasts are strongly attenuated and the spatial
shape previews become substantially more similar and predominantly unimodal,
although some regional structure remains in location, scale, and shape. 
As in the focal fit, the
day-of-year shape previews are essentially identical and not included in the figure. 
The water-body contrast is also strongly attenuated: both previews are broadly unimodal,
with the river preview showing a lower, broader peak and lower density over the displayed upper range
than the lake preview.

This sensitivity analysis illustrates a broader point about the interpretation
of conditional-shape PTMs. A shape effect is conditional on the location and scale structure
included elsewhere in the model. If persistent group-level shifts are omitted,
their mixture can appear as skewness or multimodality and can be represented by
the flexible transformation predictor. Adding the relevant grouping structure
may then remove much of the apparent shape effect. This is not a contradiction: 
the pooled model describes heterogeneity across the
sampled sites mainly through the broad discrete spatial effect, whereas the 
site-adjusted model can more closely represent the idiosyncrasies of particular sites.
Importantly, the focal spatial effects should not be
interpreted as intrinsic or causal properties of the river-basin districts.
At the same time, their attenuation demonstrates how structured shape effects
can serve as a diagnostic for unresolved heterogeneity and motivate a refined
location--scale specification.

Predictive comparisons of the three distributional specifications, with and
without monitoring-site effects, are summarized in
\autoref{tab:application-model-comparison}. The conditional-shape PTM has the lowest PSIS-LOO and WAIC deviances both with and without
monitoring-site effects; the fixed-shape PTM ranks second and the Gaussian
model third. 
An expanded comparison including slightly varied model setups and corresponding
transformed-residual Q--Q plots are provided in
\autoref{tab:dwd-recent-fit-comparison} and
\autoref{fig:application-qq}, respectively. When including monitoring sites, we observed elevated numbers of Pareto-$k$ warnings, so the model ranking for the site-adjusted models should not be treated as conclusive; \autoref{app:application-loo-diagnostics} provides additional details.

\begin{table}[btp]
    \centering
    \caption{Predictive comparison of the models fitted in the two
        applications. Reported quantities are the PSIS-LOO and WAIC deviances,
        their standard errors (SE), and the corresponding effective numbers of
        parameters $p_{\mathrm{eff}}$.}
    \label{tab:application-model-comparison}
    \scriptsize
    \setlength{\tabcolsep}{3pt}
    \begin{tabular*}{\linewidth}{@{\extracolsep{\fill}}l*{3}{rrr}@{}}
        \toprule
        & \multicolumn{6}{c}{Norwegian water conductivity}
        & \multicolumn{3}{c}{} \\
        \cmidrule(lr){2-7}
        & \multicolumn{3}{c}{Without site effects}
        & \multicolumn{3}{c}{With site effects}
        & \multicolumn{3}{c}{German temperature} \\
        \cmidrule(lr){2-4}\cmidrule(lr){5-7}\cmidrule(l){8-10}
        Model
        & Deviance $\downarrow$ & SE & $p_{\mathrm{eff}}$
        & Deviance $\downarrow$ & SE & $p_{\mathrm{eff}}$
        & Deviance $\downarrow$ & SE & $p_{\mathrm{eff}}$ \\
        \midrule
        \multicolumn{10}{@{}l}{\textbf{\sffamily PSIS-LOO}} \\
        CS-PTM & \textbf{$18{,}515$} & $210.6$ & $393$ & \textbf{$-10{,}885$} & $245.4$ & $941$ & \textbf{$1{,}457{,}565$} & $1{,}515.8$ & $5{,}907$ \\
        FS-PTM & $24{,}009$ & $222.9$ & $109$ & $-10{,}746$ & $243.1$ & $720$ & $1{,}504{,}869$ & $1{,}565.9$ & $3{,}163$ \\
        Gaussian  & $25{,}752$ & $230.9$ & $177$ & $-9{,}472$ & $273.1$ & $1{,}098$ & $1{,}511{,}783$ & $1{,}572.8$ & $4{,}833$ \\
        \addlinespace
        \multicolumn{10}{@{}l}{\textbf{\sffamily WAIC}} \\
        CS-PTM & \textbf{$18{,}481$} & $209.4$ & $376$ & \textbf{$-11{,}050$} & $242.8$ & $858$ & \textbf{$1{,}457{,}472$} & $1{,}515.6$ & $5{,}860$ \\
        FS-PTM & $24{,}007$ & $222.9$ & $108$ & $-10{,}818$ & $242.5$ & $684$ & $1{,}504{,}839$ & $1{,}565.9$ & $3{,}148$ \\
        Gaussian  & $25{,}749$ & $231.3$ & $176$ & $-9{,}826$ & $266.2$ & $921$ & $1{,}511{,}737$ & $1{,}572.7$ & $4{,}810$ \\
        \bottomrule
    \end{tabular*}

    \vspace{0.35em}
    \begin{minipage}{\linewidth}
        \footnotesize
        \textit{Note.} The predictive criteria are reported on the deviance
        scale, defined as $-2$ times the corresponding expected log predictive
        density; smaller values indicate better predictive fit.
        Both criteria use draws from the fitted variational approximations.
        For both applications, these observation-level criteria use
        the working assumption of conditionally independent observations.
        Within each
        application and model category, the reported fit minimizes both
        criteria among the corresponding rows in
        \autoref{tab:dwd-recent-fit-comparison}.
        CS-PTM and FS-PTM denote the conditional-shape and fixed-shape
        PTMs, respectively.

    \end{minipage}
\end{table}

\subsection{Daily air temperature in Germany}
\label{sec:application-dwd-temperature}

\subsubsection{Data and scientific context}
Our second application considers daily mean air temperature obtained from the
quality-controlled historical daily station observations of the Climate Data
Center of the German Weather Service (Deutscher Wetterdienst, DWD), version
v24.3 \parencite{DeutscherWetterdienst2024-DailyStationObservations}. The data
were downloaded in 2026 and are distributed under the Creative Commons
Attribution 4.0 International (CC BY 4.0) license. They cover the 20 complete
calendar years from 2006 through 2025 and comprise $164$ monitoring stations
distributed across Germany. After observations with missing temperature are
removed, the analysis includes $1{,}182{,}514$ station-days. The station set contains
all eligible stations above $900$ meters and a space-filling selection of stations
at lower elevations. The response is daily mean temperature in degrees Celsius.
Station elevations range from 0 to $1{,}486$ m, except for the outlying
Zugspitze station at $2{,}956$ m. Ten stations lie above
$900$ m. Their spatial distribution is shown in
\autoref{fig:dwd-temperature-stations}.

\begin{figure}[tbp]
    \centering
    \hfill
    \raisebox{-0.5\height}{%
        \includegraphics[width=0.32\linewidth]{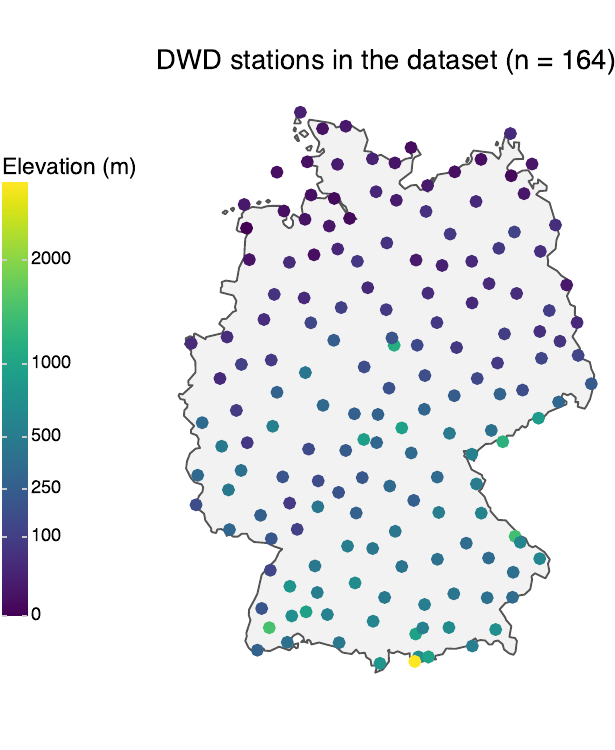}}
    \hfill
    \raisebox{-0.5\height}{%
        \includegraphics[width=0.5\linewidth]{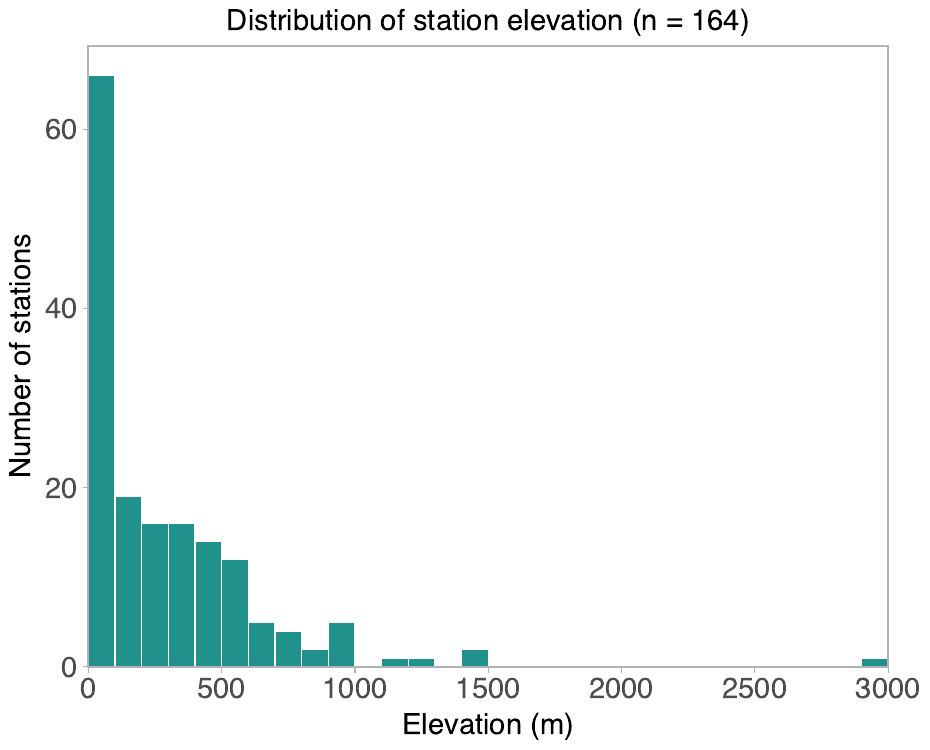}}
    \hfill
    \caption{Locations (left) and elevation distribution (right) of the 164 DWD
        monitoring stations included in the German daily-temperature application.}
    \label{fig:dwd-temperature-stations}
\end{figure}

The design is an updated temperature adaptation of earlier structured additive
distributional-regression applications to German precipitation
\parencite{Umlauf2018-BamlssBayesianAdditive,Kneib2019-ModularRegressionLego}.
Those analyses motivate the combination of seasonal, elevation, spatial, and
space-season interaction effects used here, but the response, observation
period, and interpretation differ. 
The application also provides a large-data
example with more than one million observations that demonstrates the scalability of 
our approach. The model does not use recent weather observations or output from a 
weather-forecast model. It therefore describes how temperature distributions vary across 
place and time; it does not forecast the temperature on a specific future day.

\subsubsection{Model specification}
Let $\text{day}_i$ denote the day of the year\footnote{Leap-day observations were retained when defining calendar day of year (1--366).}, $\text{elev}_i$ denote station elevation, $(\text{lon}_i,\text{lat}_i)$ denote longitude
and latitude, and $\text{year}_i$ denote the calendar year. Before fitting, we center
and scale temperature using its empirical mean and standard deviation. Apart from adding a P-spline to capture the long-term effect of $\text{year}_i$, we use the same
broad predictor structure in the location, log-scale, and shape predictors as \textcite{Umlauf2018-BamlssBayesianAdditive} and \textcite{Kneib2019-ModularRegressionLego},
\begin{alignat*}{7}
    \eta_{t,i}
     & = \gamma^t_0
     & \;              & + f^t_1(\text{day}_i)
     & \;              & + f^t_2(\text{year}_i)
     & \;              & + f^t_3(\text{elev}_i)
     & \;              & + f^t_4(\text{lon}_i)
     & \;              & + f^t_5(\text{lat}_i)
     & \;              & + f^t_6(\text{day}_i,\text{lon}_i,\text{lat}_i),      \\[0.5em]
    \bsdelta(\bsx_i)
     & = \bsdelta_0
     & \;              & + \bsf^\delta_1(\text{day}_i)
     & \;              & + \bsf^\delta_2(\text{year}_i)
     & \;              & + \bsf^\delta_3(\text{elev}_i)
     & \;              & + \bsf^\delta_4(\text{lon}_i)
     & \;              & + \bsf^\delta_5(\text{lat}_i)
     & \;              & + \bsf^\delta_6(\text{day}_i,\text{lon}_i,\text{lat}_i),
\end{alignat*}
for $t \in \{\mu, \sigma\}$.
The seasonal effects $f^t_1(\text{day}_i)$ and $\bsf^\delta_1(\text{day}_i)$ are
cyclic P-splines, the other effects are univariate and bivariate P-splines, and
$f^t_6$ and $\bsf^\delta_6$ are tensor-product terms that allow the seasonal pattern to vary over space and include the bivariate interactions. The conditional-shape PTM uses
$D=10$ transformation parameters on the core interval $[-5,5]$. We compare it with a fixed-shape PTM whose
transformation has no covariate-dependent effects and with a Gaussian
location--scale model, while retaining the same location and scale predictors.
In the location and log-scale predictors, the day, year, and elevation smooths
used basis dimension $k=20$, the longitude and latitude main effects used
$k=10$, and each margin of the tensor-product effect used $k=10$.
The corresponding dimensions in the shape predictor were $k=10$ and $k=7$,
respectively.
We calibrated the location and log-scale smoothing-variance priors as in
\autoref{sec:application-norway-water}.
The within-covariate and
between-transformation-dimension variances in
the shape predictor independently followed
$\operatorname{Weibull}(1/2,0.05)$ priors.
The variational approximation used separate dense blocks for
moderate-dimensional main-effect terms (including the main-effect terms in the shape predictor)
and diagonal blocks for the tensor-product interactions.

Additional details on the optimization settings are included in \autoref{app:application-optimization}.

\subsubsection{Results}

On a MacBook Pro with a 10-core Apple M1 Pro processor and 32~GB of memory,
the staged PTM fit took approximately 74 minutes.

\paragraph{Predictive performance and location--scale effects.}
The predictive criteria in \autoref{tab:application-model-comparison} favor
the PTM. An expanded comparison in
\autoref{tab:dwd-recent-fit-comparison} shows that other conditional-shape PTM variants with slightly
varied setups
likewise outperform the fixed-shape PTM and Gaussian models, with the focal $D=10$,
$[-5,5]$ specification attaining the best performance in both criteria.  
\autoref{fig:application-qq} provides transformed-residual Q--Q plots.
\autoref{fig:dwd-temperature-effects} summarizes the results.
The location effects show the expected seasonal, spatial, and elevation-related temperature variation, while the scale effects reveal distinct patterns in conditional dispersion. The seasonal contribution reaches
its minimum in late January and its maximum in late July to early August.
Across the nine focal cities, the seasonal posterior mean effects
excluding spatial main contributions
range from $-9.99$ to $-8.18\,{}^\circ\mathrm{C}$ at the
winter minima and from $7.98$ to $9.66\,{}^\circ\mathrm{C}$ at the summer
maxima.
The cities share a broadly similar seasonal location pattern, with modest
spatial variation mainly in the amplitude of the annual cycle.
The seasonal log-scale contribution is
largest in winter, corresponding to multiplicative scale
factors of $1.24$--$1.34$, and smallest in late September, with factors of
$0.76$--$0.84$.
These factors exponentiate the posterior mean of the seasonal main
effect plus the space-season interactions on the log-scale predictor;
they are partial scale contributions, not complete conditional standard deviations.

\begin{figure}[tbp]
    \centering
    \includegraphics[width=0.32\linewidth]{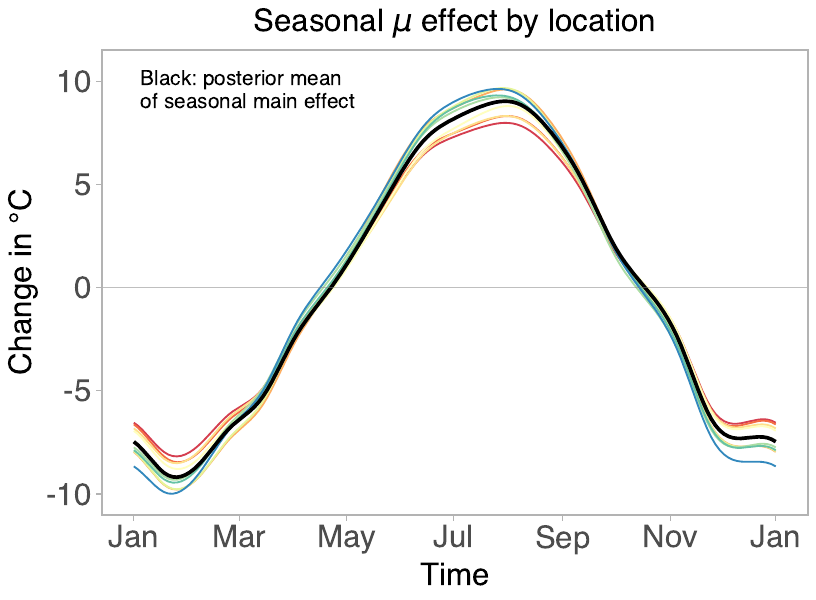}\hfill
    \includegraphics[width=0.32\linewidth]{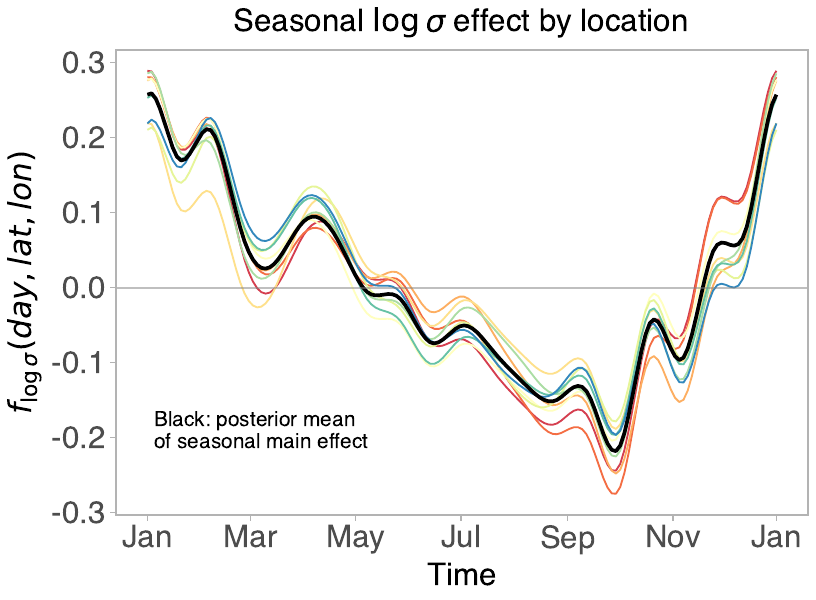}\hfill
    \includegraphics[width=0.32\linewidth]{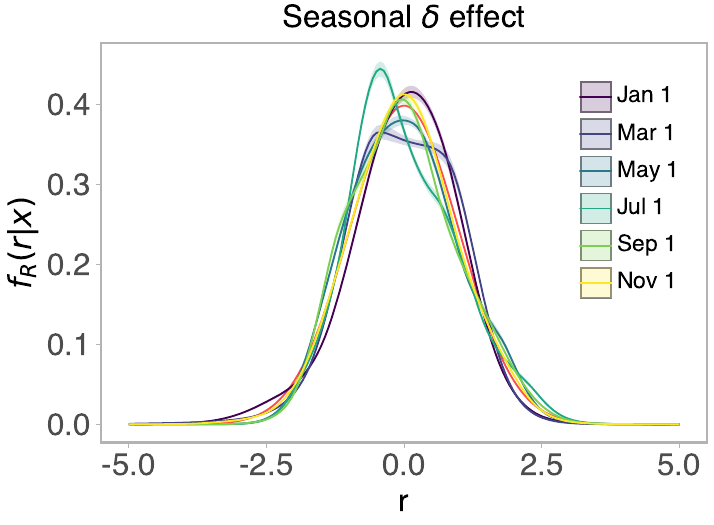}
    \par\vspace{0.75em}
    \includegraphics[width=0.32\linewidth]{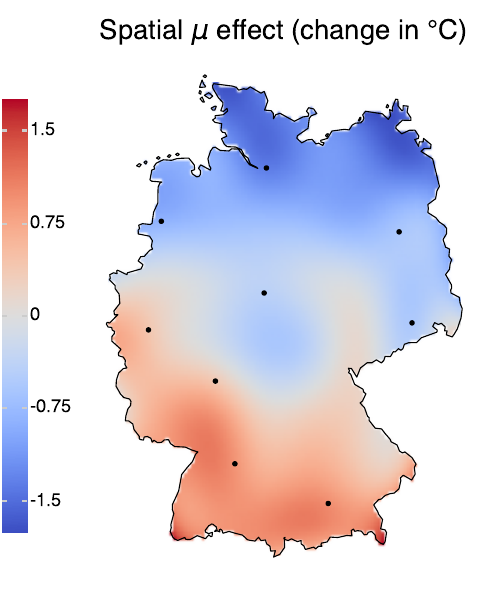}\hfill
    \includegraphics[width=0.32\linewidth]{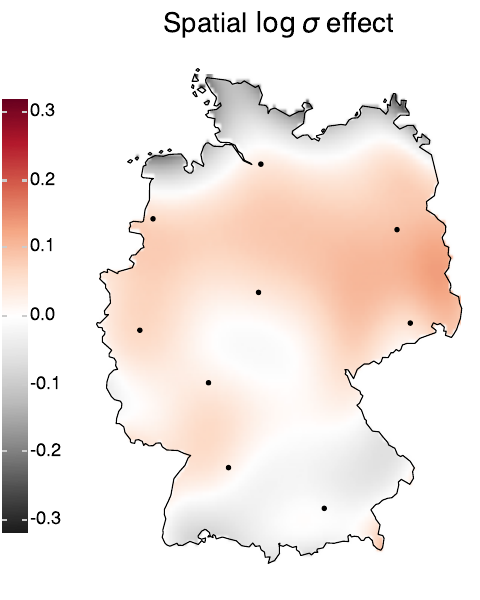}\hfill
    \includegraphics[width=0.32\linewidth]{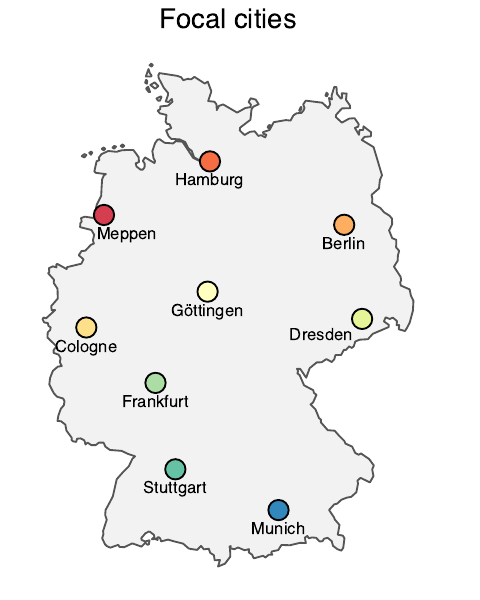}
    \par\vspace{0.75em}
    \includegraphics[width=0.32\linewidth]{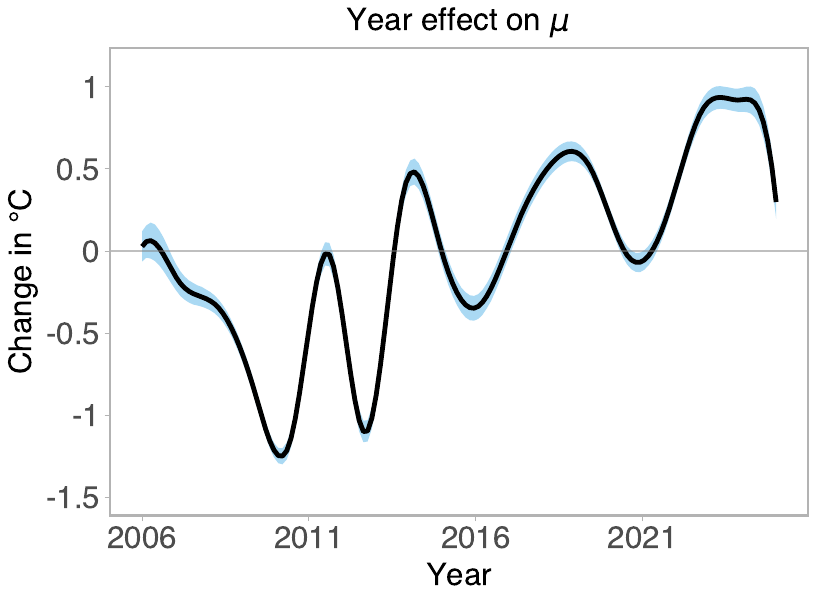}\hfill
    \includegraphics[width=0.32\linewidth]{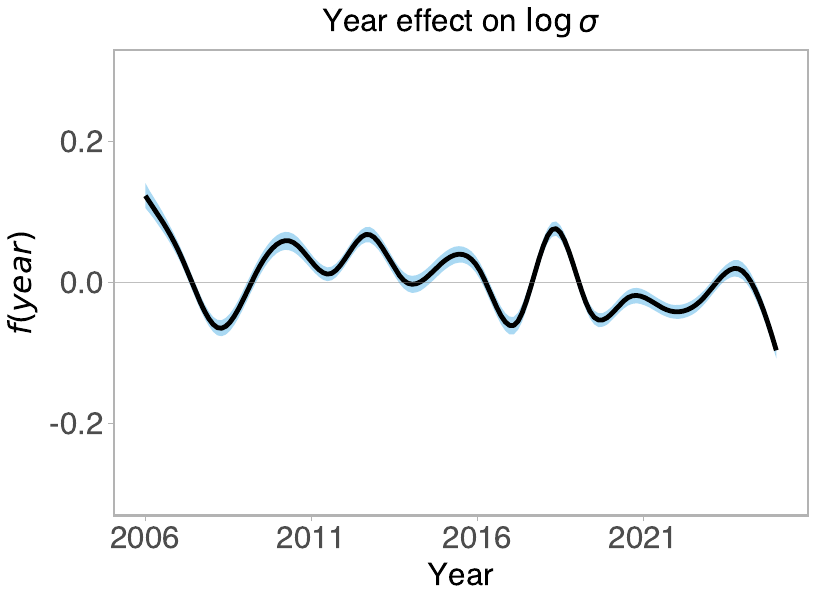}\hfill
    \includegraphics[width=0.32\linewidth]{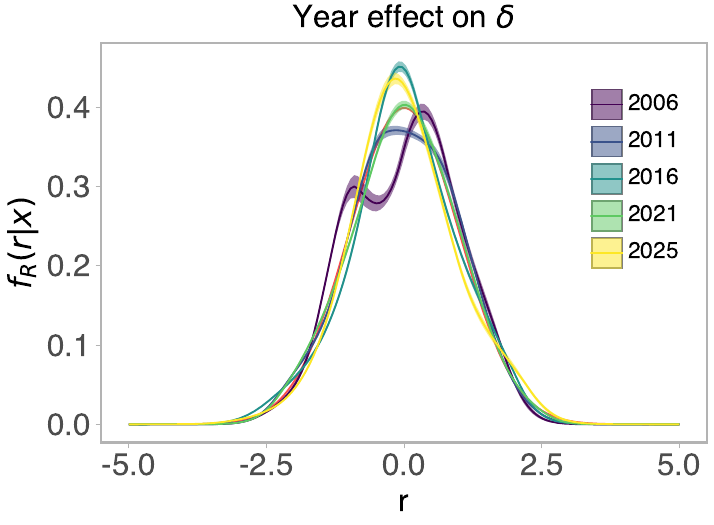}
    \par\vspace{0.75em}
    \includegraphics[width=0.32\linewidth]{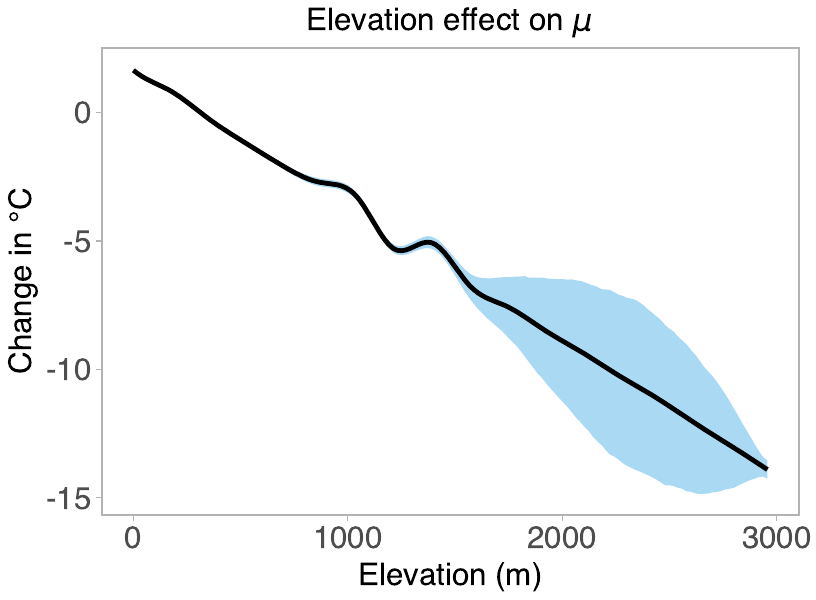}\hfill
    \includegraphics[width=0.32\linewidth]{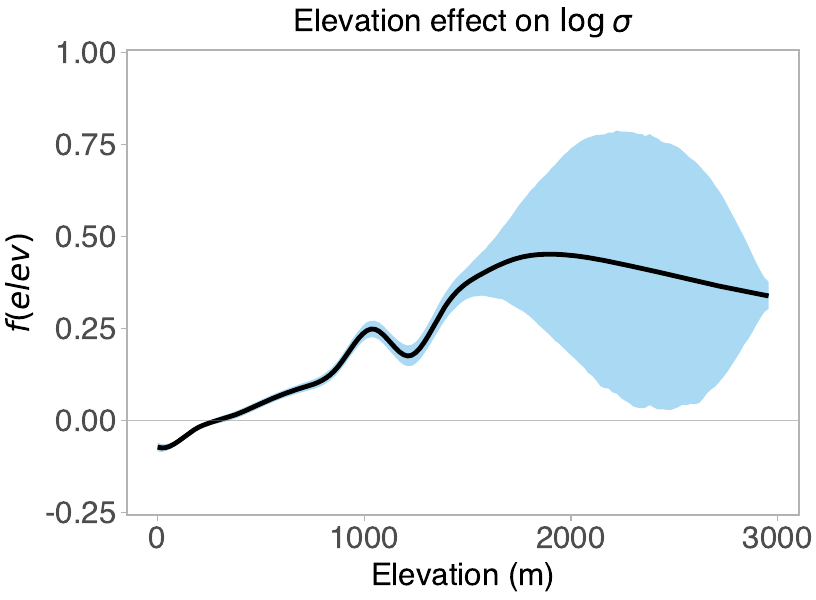}\hfill
    \includegraphics[width=0.32\linewidth]{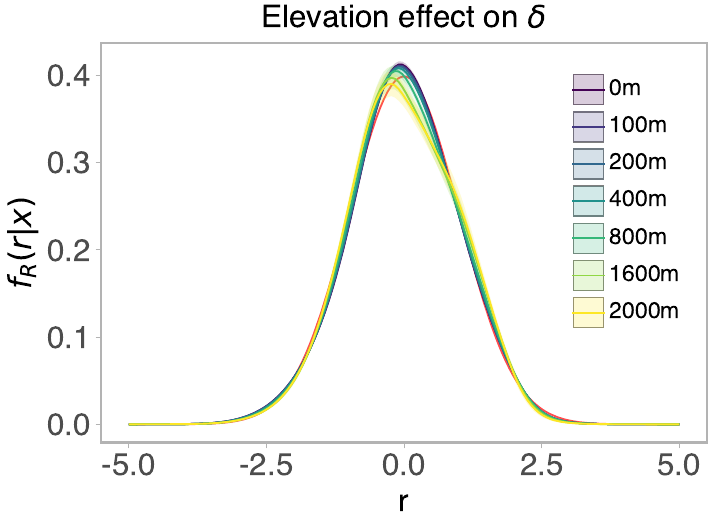}
    \caption{Effect summaries for the
        German daily-temperature application. Location effects in the left
        column are shown in degrees Celsius after multiplying effects on the
        standardized response scale by the empirical response standard
        deviation, $7.40\,^\circ\mathrm{C}$. In the seasonal location and
        log-scale panels, colored curves combine the seasonal main effect
        with the day--latitude, day--longitude, and day--latitude--longitude interactions at the nine focal cities;
        black curves show the posterior mean of the seasonal main effect alone.
        The colored curves exclude the latitude and longitude main effects
        and the latitude--longitude interaction, so colored-minus-black differences
        isolate the space-season interactions.
        Uncertainty for these effects is not shown to avoid clutter. Where
        shown in the other effect panels, blue shaded ribbons are 90\%
        pointwise credible intervals. The spatial maps
        show the spatial main contribution and exclude the day-specific
        space-season interaction; the middle-right focal-city map is the
        location and color key. Each standardized
        density preview combines the shape intercept with the indicated term
        while the other centered shape terms are held at zero. These previews
        are diagnostic partial shape summaries, not complete conditional
        densities or densities averaged over the omitted covariates.
        Thin red curves mark the standard Gaussian density for reference.}
    \label{fig:dwd-temperature-effects}
\end{figure}

The spatial location effect shows a broad north-south gradient,
with lower temperature over much of northern and northeastern Germany
and higher temperature in the south. In contrast, the spatial log-scale
effect is positive across a broad central-to-eastern region and negative along
the northern coast and in parts of southern Germany. The fitted year effect
on location oscillates but trends upward as successive low and high points generally increase.
This pattern is consistent with
general warming, but it should be noted that this is a descriptive conditional result for the selected
station network, not a causal or nationally representative trend estimate.
The year effect on log scale is also non-monotone, and it is smaller than the
seasonal scale variation.

The elevation effect on location decreases almost linearly by about
$15$--$16\,{}^\circ\mathrm{C}$ between sea level and $3{,}000$ m, equivalent
on average to roughly $5\,{}^\circ\mathrm{C}$ per $1{,}000$ m. The log-scale effect increases over the
well-supported elevation range. There are no stations between $1{,}486$ and
$2{,}956$ m, and the posterior intervals for the two elevation effects widen accordingly.
The broad interpolation across the gap is generally plausible, but its
shape is not directly supported by observations.
The near-overlap of the elevation density previews indicates that the
corresponding shape effect is strongly regularized.

\begin{figure}[tbp]
    \centering
    \includegraphics[width=\linewidth]{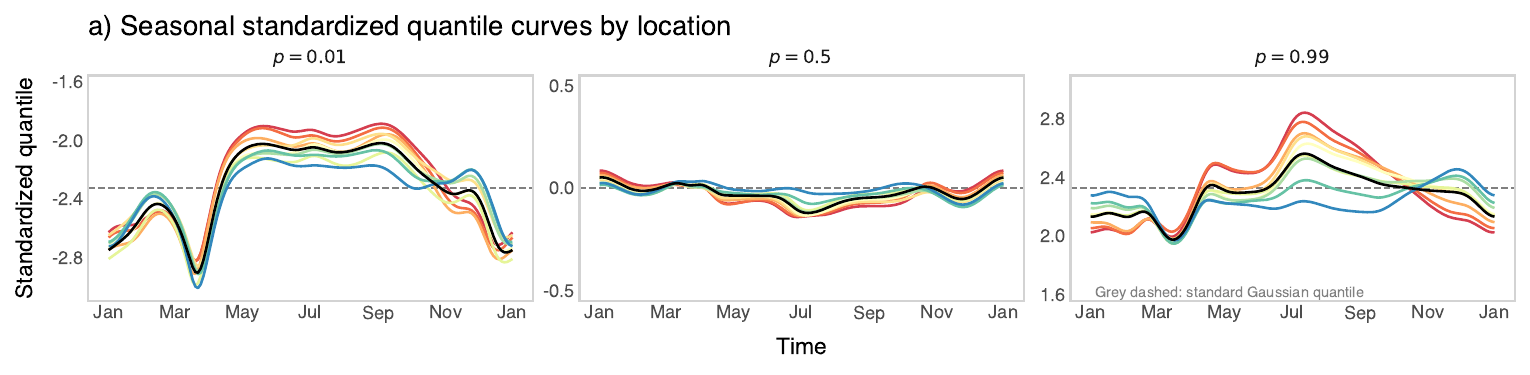}
    \par\vspace{0.35em}
    \includegraphics[width=\linewidth]{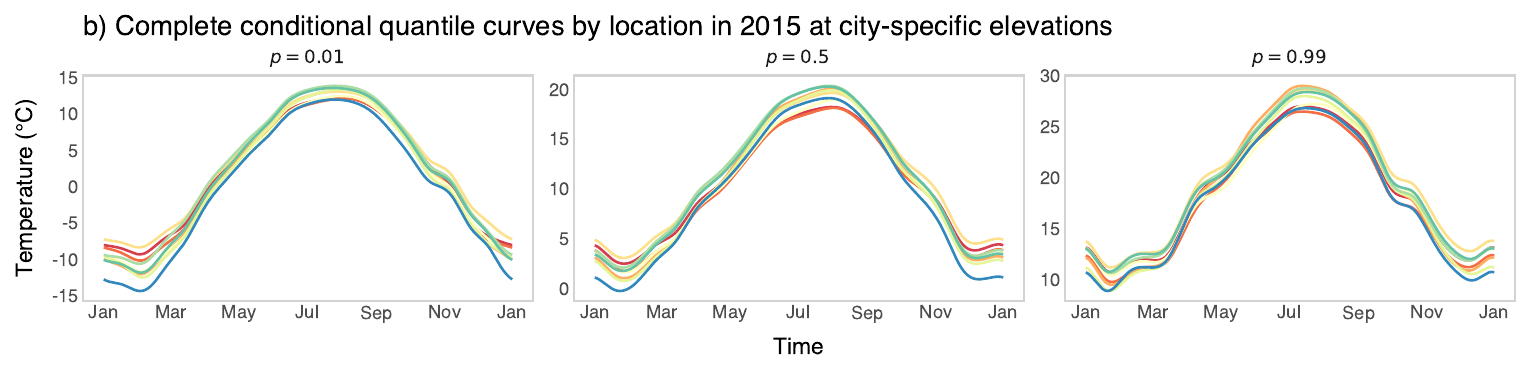}
    \par\vspace{0.35em}
    \includegraphics[width=\linewidth]{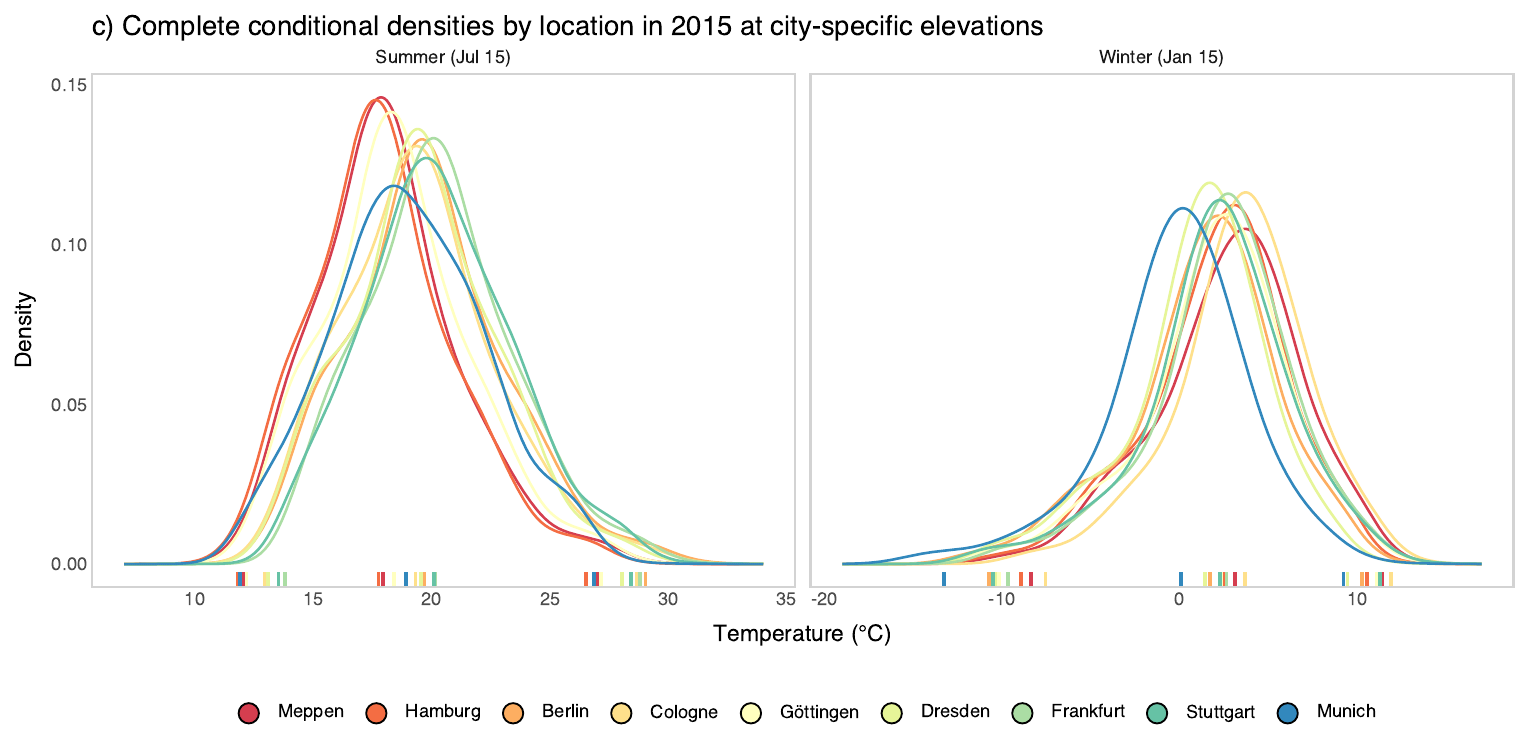}
    \caption{Seasonal distribution summaries for the PTM at the nine
        locations in the German daily-temperature application. The top row
        shows diagnostic posterior mean standardized quantiles at
        $p\in\{0.01,0.5,0.99\}$. Colored curves combine the shape intercept
        with the seasonal, spatial main, and space-season interaction terms,
        while black curves retain only the shape intercept and seasonal main
        effect; year and elevation shape effects are omitted. Gray dashed lines
        mark the corresponding standard Gaussian quantiles. The middle row
        shows complete conditional quantiles in degrees Celsius during 2015 at
        the city-specific elevations. The lower panels show the corresponding
        complete conditional densities on two representative
        days. Colored
        rug marks on the horizontal axes indicate the location-specific
        $p\in\{0.01,0.5,0.99\}$ quantiles. Credible interval ribbons are
        omitted to avoid clutter. Corresponding summaries for the fixed-shape PTM and Gaussian
location--scale model are shown in
\autoref{fig:dwd-basic-ptm-quantile-curves} and
\autoref{fig:dwd-gaussian-quantile-curves}, respectively.}
    \label{fig:dwd-temperature-quantile-curves}
\end{figure}

\paragraph{Standardized shape effects.}
The seasonal density previews in the right column of \autoref{fig:dwd-temperature-effects}
suggest a reversal in asymmetry: winter previews extend farther toward the lower end of the displayed range, whereas summer previews extend farther toward the upper end of the displayed range.

The top row of \autoref{fig:dwd-temperature-quantile-curves} probes this asymmetry at the displayed $p=0.01$ and $p=0.99$ quantiles more directly.
During much of summer, the fitted $p=0.99$ quantile exceeds the
corresponding standard Gaussian quantile at most focal locations, placing the displayed upper quantile farther from zero, while the fitted $p=0.01$ quantile is less negative at all
nine locations, placing the displayed lower quantile closer to zero. Munich is the clearest
exception for the upper quantile: its summer $p=0.99$ curve remains near or below
the Gaussian reference. In winter and early spring, the fitted $p=0.99$
quantiles are generally less extreme than the Gaussian reference, whereas the
fitted $p=0.01$ quantiles are more negative, so the displayed upper quantile lies closer to zero and the lower quantile lies farther from zero.
Because all displayed
distributions have mean zero and variance one, these movements
describe pure shape changes, isolated from location and scale.

The seasonal shape pattern also shows substantial spatial variation, most prominently
in the extreme quantiles. During the plotted summer window,
Meppen's $p=0.99$ quantile reaches $2.84$, whereas Munich's peaks at
$2.24$; differences among locations are smaller at $p=0.5$.
The changing gaps between the
colored curves and the black seasonal-main-effect curve indicate that spatial
shape effects interact with the seasonal pattern rather than acting as fixed offsets.

However, note that the density previews in \autoref{fig:dwd-temperature-effects}
and the derived quantile curves in \autoref{fig:dwd-temperature-quantile-curves} 
should not be confused with full conditional densities or quantiles of full conditional
distributions. They are instead probes that allow us to inspect how seasonal and 
spatial shape terms modify the reference distribution in isolation, giving us insight 
into how the model arrives at its predictions.

\paragraph{Complete conditional distributions.}
Panels~b and~c of
\autoref{fig:dwd-temperature-quantile-curves} show how the fitted location,
scale, and shape effects combine into complete conditional quantiles and densities
on the temperature scale, where the seasonal location effect is dominant.
\autoref{fig:dwd-basic-ptm-quantile-curves} and
\autoref{fig:dwd-gaussian-quantile-curves} provide companion displays for the
fixed-shape PTM and Gaussian location--scale model.
The colored rug marks in the lower panels show
the posterior mean complete conditional quantiles at
$p\in\{0.01,0.5,0.99\}$;
\autoref{tab:dwd-complete-conditional-quantiles} reports the exact values. 

The conditional-shape PTM's added shape flexibility
can alter quantile predictions even when the fitted centers are nearly unchanged.
For Meppen on July 15, the $(p=0.01,0.5,0.99)$ quantiles are
$(12.03,17.95,27.01)\,{}^\circ\mathrm{C}$ under the conditional-shape model,
$(11.00,18.29,25.66)\,{}^\circ\mathrm{C}$ under the fixed-shape model, and
$(10.69,18.14,25.59)\,{}^\circ\mathrm{C}$ under the Gaussian model. The conditional-shape model
and Gaussian medians differ by only $0.19\,{}^\circ\mathrm{C}$, but the
distances from the conditional-shape model median to its lower and upper quantiles are $5.92$ and
$9.06\,{}^\circ\mathrm{C}$, respectively, compared with
$7.45\,{}^\circ\mathrm{C}$ in both directions under the Gaussian model. The
farther extension into the displayed upper range of Meppen's summer density in panel~c makes this unequal
quantile spacing visible.

Munich on January 15 shows the reverse pattern: its
$p=0.01$ curve reaches particularly low winter values in panel~b, and its blue
density in panel~c extends farther into the displayed lower range. The conditional-shape model
$p=0.01$, $0.5$, and $0.99$ quantiles are $-13.30$, $0.09$, and
$9.22\,{}^\circ\mathrm{C}$, respectively. The corresponding values are
$-9.54$, $0.05$, and $9.73\,{}^\circ\mathrm{C}$ for the fixed-shape model and $-9.91$,
$-0.08$, and $9.75\,{}^\circ\mathrm{C}$ for the Gaussian model. The fitted
medians again differ little, while the conditional-shape model has a substantially lower $p=0.01$ quantile and a slightly lower $p=0.99$ quantile than both comparison models. Across all nine
locations, its winter $p=0.01$ quantiles are $1.67$--$3.75\,{}^\circ\mathrm{C}$
below those of the fixed-shape model and $1.42$--$3.39\,{}^\circ\mathrm{C}$ below those
of the Gaussian model. 

Together, the figures and predictive comparison show how the PTM preserves
a familiar location--scale interpretation while adding controlled shape flexibility where
the data support it.

\section{Discussion}\label{sec:discussion}

Conditional-shape PTMs extend the PTM family from structured
additive location--scale regression to settings in which covariates
may also affect distributional shape.
The model is constructed to include the conditional mean $\mu(\bsx)$ and
standard deviation $\sigma(\bsx)$ of the response, while the transformation predictor
$\bsdelta(\bsx)$ controls the standardized shape of the response's conditional distribution. Its vector-valued structured additive terms allow
scientifically motivated covariate effects, and their penalties regularize departures
from a reference-family (usually Gaussian) location--scale model. The construction thus
combines a coherent conditional distribution with an explicit separation of
location, scale, and shape.

The empirical results support this contribution at three distinct levels. In
the controlled setting of simulation Scenario A, the PTM recovers the response-level
effects and the covariate-dependent transition from skewness to bimodality; in
the Gaussian subscenario, its predictive performance is nearly identical to
that of the simpler nested models, indicating effective regularization of
unnecessary transformation flexibility. Under the deliberate
misspecification of Scenario B, the PTM improves with sample size and
performs approximately as well as a Dirichlet process mixture model in terms of test-set Kullback--Leibler divergence and
CDF deviation.
The two applications illustrate a complementary benefit. The PTM retains
seasonal, spatial, temporal, or group-specific shape variation while
strongly regularizing other shape effects, and it improves PSIS-LOO and WAIC
relative to both simpler comparison models.

Methodologically, conditional-shape PTMs occupy a middle ground between parsimonious
location--scale regression and less structured conditional-density methods.
They are most useful when a flexible estimate of a complete conditional distribution is required,
the first two conditional moments should retain exact response-scale
interpretations, and substantive knowledge suggests which covariates may
change shape. A Gaussian location--scale model is preferable when its residual
shape is adequate, and a fixed-shape PTM is more economical when a single
non-Gaussian standardized distribution suffices. A well-supported parametric
distributional model may be more parsimonious, while quantile regression
is natural when only selected quantiles are the estimands and mixture-based or other
flexible density regression may be preferable when unrestricted predictive
density estimation is primary. The distinctive role of the conditional-shape PTM is therefore
structured control over where distributional flexibility enters, not
universal dominance over these alternatives.

The construction also suggests a practical fitting and interpretation
strategy. Scientifically important grouping, location, and scale effects
should be specified before residual skewness or multimodality is attributed to
shape. The reference
distribution and spline-core boundaries $a$ and $b$ should reflect defensible tail
assumptions, because changing the core support can extend finite-range flexibility but
cannot change the asymptotic tail class. Shape effects can be probed through
contrasts in standardized preview densities or standardized conditional
quantiles. This decomposability is an interpretive advantage of the conditional-shape PTM, but
the resulting partial summaries should not be mistaken for response-scale
predictions or for densities marginalized over the remaining covariates.
Predictive interpretation requires the complete shape, location, and scale
predictors. Fitted models can be
checked using transformed-residual Q--Q plots.

Four principal limitations should be noted. First, the blocked
Gaussian variational approximation omits cross-block posterior dependence and
non-Gaussian posterior features, and no exact-posterior PTM benchmark is
available. This family is comparatively rudimentary relative
to recent approaches for
additive and distributional regression
\parencite{Kleinemeier2023-ScalableEstimationStructured,Lichter2024-VariationalInferenceUncertainty,Callegher2025-StochasticVariationalInferencea}.
In additive-model comparisons, retaining dependence among coefficient blocks
was crucial for calibrated credible intervals even when point estimates were
similar \parencite{Lichter2024-VariationalInferenceUncertainty}. The
undercoverage in Scenario B cannot be attributed to this
approximation alone because every model is misspecified to some degree there,
but it shows that good predictive point summaries do not guarantee calibrated
uncertainty. Second, shape effects are conditional on the full model
specification: omitted grouping structures or interactions can be absorbed as
apparent shape effects, as demonstrated by the Norwegian sensitivity analysis and the
unmodeled interaction in Scenario B. In practice, users should therefore check whether
shape effects point to influences of hidden covariates. Third, the identity-tailed
transformation inherits the asymptotic tail class of the reference
distribution. It can flexibly describe distributions over the spline core,
but it cannot learn a fundamentally different limiting tail behavior.
Finally, the hyperprior for the between-parameter smoothing variance $\psi^2$
was chosen pragmatically based on an intercept-only prior check. Its implied
regularization need not be comparable across shape terms with different
bases, penalties, or covariate domains. These limitations motivate three focused extensions: richer structured
posterior approximations,
shape-effect priors calibrated on an interpretable distributional deviation
scale, and transformations that permit more flexible tail adaptation
without sacrificing the mean-standard-deviation separation.

Conditional-shape PTMs provide a principled and practical
extension of interpretable location--scale regression: they permit
scientifically targeted changes in distributional shape within one coherent
conditional distribution and regularize toward a simpler reference model when
that flexibility is unsupported.

\section*{Data and code availability}

The data and code supporting this study are currently available upon request.
They will be made publicly available upon submission to a journal or in a
forthcoming update.

\section*{Acknowledgments}
The authors used OpenAI Codex with the GPT-5.6 Sol model for coding and writing assistance. All generated outputs were critically reviewed and revised by the authors. The authors maintain sole responsibility for the software, analyses, interpretations, and final manuscript.

\section*{Conflict of interest}

The authors declare no conflicts of interest.

\section*{Funding}
We are grateful to the German Research Foundation (DFG) for financial support through grant 443179956.

\printbibliography

\clearpage
\appendix
\footnotesize
\setkomafont{paragraph}{\footnotesize}
\numberwithin{figure}{section}
\numberwithin{table}{section}

\section{Details for multivariate structured additive predictors}
\label{app:mv-star-predictors}

This appendix collects the derivations, computational details, and
examples supporting the compact formulation in \autoref{sec:mvpred-ptm}.

Expanding the predictor into its $D$ elements and the basis vector into its $K$ terms yields
\begin{align}
    \begin{bmatrix}
        \eta_{1} \\
        \vdots   \\
        \eta_{D} \\
    \end{bmatrix}
    \quad=\quad
    \begin{bmatrix}
        f_{1,1}(\bsx) \\
        \vdots        \\
        f_{D,1}(\bsx) \\
    \end{bmatrix} +
    \cdots +
    \begin{bmatrix}
        f_{1,K}(\bsx) \\
        \vdots        \\
        f_{D,K}(\bsx) \\
    \end{bmatrix}
    \quad=\quad
    \begin{bmatrix}
        \bsgamma^\top_{1,1} & \cdots & \bsgamma^\top_{1,K} \\
        \vdots              & \ddots & \vdots              \\
        \bsgamma^\top_{D,1} & \cdots & \bsgamma^\top_{D,K} \\
    \end{bmatrix}
    \begin{bmatrix}
        \bfb_1(\bsx) \\
        \vdots       \\
        \bfb_K(\bsx) \\
    \end{bmatrix}.
\end{align}

\paragraph{Between-dimension regularization.}
Between-dimension regularization is governed by the matrix
$\bfP^{(d)}(\psi^{2}_{k}) = \psi^{-2}_k \bfK_k^{(d)}$,
where $\bfK_k^{(d)}$ is a constant, potentially rank-deficient positive-semidefinite penalty matrix that defines the correlation structure. The variance parameter $\psi^2_k$ acts as an inverse smoothing parameter, so smaller values imply stronger between-dimension regularization for term $k$. The structural penalty matrix $\bfK_k^{(d)}$ is chosen manually to define an appropriate correlation pattern.
% Note that $\bfK^{(d)}$ is not term-specific, so we assume the same between-dimension correlation structure for all $k = 1, \dots, K$, while the strength of regularization defined through $\psi^2_k$ can vary across terms.

\paragraph{Within-dimension regularization.}
In the most general case, the matrix $\bfP^{(x)}_{k}(\bstau^{2}_{k})$ is an anisotropic tensor-product precision matrix that defines the within-dimension regularization in the covariate directions. The variance parameters in $\bstau^2_k = [\tau^2_{k,1}, \dots, \tau^2_{k,M_k}]^\top$ control the strength of regularization. For an $M_k$-dimensional tensor product, it is given by the sum
$$
    \bfP^{(x)}_{k}(\bstau^{2}_{k}) =
    \tau^{-2}_{k, 1}\widetilde \bfK^{(x)}_{k,1}
    + \cdots +
    \tau^{-2}_{k, M_k}\widetilde \bfK^{(x)}_{k,M_k}
$$
with the individual $\widetilde \bfK^{(x)}_{k,m}$ for $m = 1, \dots, M_k$ given by
$$
    \widetilde \bfK^{(x)}_{k,m} = \bfI_{L_{k,1}} \otimes \cdots \otimes \bfI_{L_{k,m-1}} \otimes
    \bfK^{(x)}_{k,m}
    \otimes \bfI_{L_{k,m+1}} \otimes \cdots \otimes \bfI_{L_{k,M_k}},
$$
such that $\bfK^{(x)}_{k,m}$ is the constant, potentially rank-deficient positive-semidefinite penalty matrix defining the correlation structure in the $m$th covariate direction of term $k$.

\paragraph{Hyperpriors for variance parameters.}
Hyperprior choices for variance parameters can be based on the considerations
for hierarchical models outlined by
\textcite{Gelman2006-PriorDistributionsVariance} and the concept of penalized
model complexity developed by \textcite{Simpson2017-PenalisingModelComponent}
and adapted to structured additive models by
\textcite{Klein2016-ScaledependentPriorsVariance}. This leads to $\operatorname{Weibull}(1/2,\theta)$ priors for smoothing variances $\tau^2$ and $\psi^2$, which favor the base models reached as the variance parameters approach zero while retaining a right tail that permits larger
scales. The parameter $\theta$ denotes the scale of the hyperprior and can be chosen
using an appropriate scaling criterion; see \textcite{Klein2016-ScaledependentPriorsVariance}. Other common hyperprior choices include inverse-gamma, half-normal, and half-Cauchy distributions; see \textcite{Gelman2006-PriorDistributionsVariance} and \textcite{Klein2016-ScaledependentPriorsVariance} for discussions of their merits.

\paragraph{Joint prior.}
The joint prior for all parameters in a multivariate structured additive predictor factors into the priors for the parameters and hyperparameters of the terms:
\begin{align}
    p(\bsGamma, \bstau^2, \bspsi^2) & =
    p(\bsGamma \mid \bstau^2, \bspsi^2) \, p(\bstau^2) \, p (\bspsi^2) \\
                                    & =
    \left[
        \prod_{k=1}^K p(\bsgamma_k | \bstau^2_k, \psi^2_k)
        \right]
    \times
    \left[
        \prod_{k=1}^K
        \prod_{m=1}^{M_k}
        p(\tau^2_{k,m})
        \right]
    \times
    \left[
        \prod_{k=1}^K
        p(\psi^2_k)
        \right],
\end{align}
where $\bstau^2 = [\bstau^{2\top}_1, \dots, \bstau^{2\top}_K]^\top$ and $\bspsi^2 = [\psi^2_1, \dots, \psi^2_K]^\top$.

\paragraph{Computational efficiency.}
Direct evaluation of the prior density \eqref{eq:mvstar-prior} is computationally expensive because it requires computing $\operatorname{pdet}(\bfP_k)$. However,
the structure of $\bfP_k$ allows between-dimension regularization to be treated as an additional direction in an anisotropic tensor-product smooth. The log-pseudo-determinant can therefore be computed efficiently from eigendecompositions of the constant, relatively low-dimensional penalty matrices $\bfK^{(x)}_{k,1}, \dots, \bfK^{(x)}_{k,M_k}$ and $\bfK^{(d)}_k$, as described in Theorem~3.1 and Corollary~3.2 of \textcite{Bach2025-AnisotropicMultidimensionalSmoothing}. For a given model, these eigendecompositions need to be computed only once, greatly reducing the cost of subsequent prior-density evaluations.

\paragraph{Linear constraints.}
Linear constraints can be applied to the within-dimension and cross-dimension
directions of a term. These two cases correspond to the two axes of the
coefficient matrix $\bsGamma_k$. We apply the reparameterization approach described by \textcite{Kneib2019-ModularRegressionLego}.
We first consider a within-dimension
constraint with full-row-rank matrix
$\bfA_k^{(x)}\in\bbR^{r_k\times L_k}$. Requiring the constraint to hold in
each of the $D$ predictor dimensions gives $\bsGamma_k\bfA_k^{(x)\top}=\bfzero_{D\times r_k}$.
Now let $\bfQ_k^{(x)}\in\bbR^{L_k\times(L_k-r_k)}$ have orthonormal columns
spanning the null space of $\bfA_k^{(x)}$. The constraint can then be enforced by
writing the term as
$\bsf_k(\bsx)=\overline{\bsGamma}_k \overline{\bfb}_k(\bsx)$, where $\overline{\bfb}_k(\bsx) = \bfQ_k^{(x)\top}\bfb_k(\bsx)$ is a reparameterized basis and $\overline{\bsGamma}_k \in \bbR^{D \times (L_k-r_k)}$ is a reparameterized
coefficient matrix whose prior uses reparameterized within-dimension penalty matrices
$\bfQ_k^{(x)\top} \widetilde{\bfK}_{k,m}^{(x)}\bfQ_k^{(x)}$.
In the following, we suppress the bars: $\bfb_k(\bsx)$ and $L_k$ denote the
reparameterized covariate basis and its dimension, while $\bsGamma_k$ denotes
the corresponding coefficient matrix. The same convention applies to
$\widetilde{\bfK}_{k,m}^{(x)}$.

Now let $\bfA^{(d)}\bsGamma_k=\bfzero_{q\times L_k}$ be a constraint
in the cross-dimensional direction with constraint matrix $\bfA^{(d)}\in\bbR^{q\times D}$,
and let $\bfQ^{(d)}\in\bbR^{D\times(D-q)}$ have orthonormal columns spanning the null
space of $\bfA^{(d)}$. Then the constraint can be applied by writing
$\bsGamma_k = \bfQ^{(d)}\widetilde{\bsGamma}_k$, where $\widetilde{\bsGamma}_k \in \bbR^{(D-q) \times L_k}$
is a reparameterized coefficient matrix whose prior uses the reparameterized
cross-dimension penalty matrix $\bfQ^{(d)\top}\bfK_k^{(d)}\bfQ^{(d)}$. Since a
cross-dimensional constraint typically applies to all terms in a predictor, the
reconstruction back into $D$ dimensions can be factored out and applied once after summing all reduced terms via
$\bseta = \bsGamma \bfb(\bsx) = \bfQ^{(d)}\{[\widetilde \bsGamma_1, \dots, \widetilde \bsGamma_K] \bfb(\bsx)\}$.

The reparameterization matrices can be obtained from an eigendecomposition of
$\bfA^\top\bfA$: the eigenvectors associated with zero eigenvalues form an orthonormal basis
$\bfQ$ of the null space of $\bfA$; see Section~2.3.2 of \textcite{Kneib2019-ModularRegressionLego} for details.

\subsection{Examples of term structures} \label{app:mvterms-examples}

The general framework outlined above simplifies considerably for many concrete models. We include four common examples of specific term definitions below.

\paragraph{Intercept in multiple dimensions.} An intercept $\bsf_k(\bsx) = \bsgamma_k$ is characterized by a scalar, constant basis $\bfb_k(\bsx)~=~1$, implying $M_k=1$ and a single zero-penalty in the covariate direction, $\bfK_k^{(x)} = 0$. In this case, the prior precision matrix simplifies to
$$
    \bfP_k = \bfP_k(\psi^2_k) = \psi_k^{-2}\bfK^{(d)}_k,
$$
and the dimension of $\bsgamma_k$ reduces to $(D \times 1)$.

\paragraph{Univariate term in multiple dimensions.} A univariate term $\bsf_k(\bsx) = \bsf_k(x_{im})$ is characterized by a basis function that is constant in all except one of the covariates in $\bsx = [x_{1}, \dots, x_{M}]^\top$, so the basis becomes $\bfb_k(\bsx)~=~\bfb_k(x_{m})$ for some $m \in \{1, \dots, M\}$, implying $M_k = 1$. The covariate-related component of the prior precision matrix then simplifies to the $(L_{k,m} \times L_{k,m})$ matrix
$$
    \bfP_k^{(x)}(\bstau^2_k) = \bfP_k^{(x)}(\tau^2_{k,m}) = \tau^{-2}_{k,m}\bfK^{(x)}_{k,m}.
$$
Consequently, the full prior precision matrix is
$$
    \bfP_k = \bfP_k(\tau^2_{k,m}, \psi^2_k)
    =
    \frac{1}{\tau^2_{k,m}}
    \left(
    \bfK^{(x)}_{k,m} \otimes \bfI_D
    \right)
    +
    \frac{1}{\psi^2_{k}}
    \left(
    \bfI_{L_{k,m}} \otimes
    \bfK^{(d)}_{k}
    \right)
$$
and $\bsgamma_k$ has dimension $(DL_{k,m} \times 1)$.

\paragraph{Bivariate tensor product term in multiple dimensions.}
A bivariate tensor product term $\bsf_k(\bsx) = \bsf_k(x_{m}, x_{m'})$ for $m, m' \in \{1, \dots, M\}$ and $m \neq m'$ is characterized by a basis function that is a tensor product of two marginal univariate bases of potentially different dimensions,
$$
    \dimann{L_{k,m}L_{k,m'} \times 1}{\bfb_k(\bsx)}
    =
    \bfb_k(x_{m}, x_{m'}) =
    \dimann{L_{k,m} \times 1}{\bfb_{k,m}(x_{m})}
    \otimes
    \dimann{L_{k,m'} \times 1}{\bfb_{k,m'}(x_{m'})},
$$
implying $M_k = 2$.
The covariate-related component of the prior precision matrix becomes the Kronecker sum of the two marginal precision matrices,
$$
    \bfP_k^{(x)}(\bstau^2_k) = \bfP_k^{(x)}(\tau^2_{k,m}, \tau^2_{k,m'}) = \frac{1}{\tau^{2}_{k,m}}
    \left(\bfK^{(x)}_{k,m}
    \otimes
    \bfI_{L_{k,m'}}
    \right)
    +
    \frac{1}{\tau^{2}_{k,m'}}
    \left(
    \bfI_{L_{k,m}}
    \otimes
    \bfK^{(x)}_{k,m'}
    \right)
    .
$$
The full prior precision matrix is then
$$
    \bfP_k =
    \frac{1}{\tau^2_{k,m}}
    \left(
    \bfK^{(x)}_{k,m}
    \otimes
    \bfI_{L_{k,m'}}
    \otimes \bfI_D
    \right)
    +
    \frac{1}{\tau^2_{k,m'}}
    \left(
    \bfI_{L_{k,m}}
    \otimes
    \bfK^{(x)}_{k,m'}
    \otimes \bfI_D
    \right)
    +
    \frac{1}{\psi^2_{k}}
    \left(
    \bfI_{L_{k,m}} \otimes
    \bfI_{L_{k,m'}} \otimes
    \bfK^{(d)}_{k}
    \right)
$$
and $\bsgamma_k$ has dimension $(DL_k \times 1)$ with $L_k = L_{k,m'}L_{k,m}$.

\paragraph{Terms in a single dimension.}
A one-dimensional term is characterized by a scalar zero-penalty for the between-dimension direction, $\bfK^{(d)}_k = 0$. In this case, the prior precision matrix simplifies to
$\bfP_k = \bfP_k^{(x)}(\bstau^2_k)$
and the dimension of $\bsgamma_k$ reduces to $(L_k \times 1)$. This special case results in a conventional structured additive regression term.

\subsection{Examples of penalty and basis choices} \label{app:mvterm-bases-and-penalties}

\paragraph{Between-dimension penalty.} The choice of between-dimension penalty determines how regularization acts across dimensions. For example, choosing a zero-penalty matrix $\bfK_k^{(d)} = \bfzero_{D\times D}$ implies no correlation and therefore no regularization across dimensions. This is the implicit default in \textcite{Klein2015-BayesianStructuredAdditive} and in common structured additive distributional regression models, where several uncorrelated univariate predictors can be interpreted as the elements of one multivariate predictor.
Another natural choice is a ridge penalty $\bfK_k^{(d)} = \bfI_D$, which leads to general random-effects-like regularization. Of particular interest for our model is the choice of a first-order random-walk penalty $\bfK_k^{(d)} = \bfD^\top_1 \bfD_1$, where $\bfD_1$ denotes the $(D-1) \times D$ first-difference matrix.
This penalty encourages neighboring component functions to remain similar; in the random-walk representation, each function is centered on the preceding one, and smaller values of $\psi_k^2$ permit less variation between them.

\paragraph{Within-dimension penalties and bases.} Within-dimension penalties and bases define effect types just as they do for univariate structured additive predictors. \textcite[Ch. 9]{Fahrmeir2013-RegressionModelsMethods} provide an overview. We briefly highlight some common special cases:
\begin{itemize}
    \item \textit{Linear effects} are defined with bases $\bfb_k(\bsx) = [x_1, \dots, x_{M_k}]^\top$. They are often supplied with a zero-penalty $\bfK_k^{(x)} = \bfzero_{M_k \times M_k}$ to omit penalization.
    \item \textit{Random intercepts} for $G$ clusters use a one-hot-coded basis vector $\bfb_k(x_k)$ comprising $G$ dummy variables. Specifically, $b_{k,g}(x_k) = 1$ if $x_k = g$, and $b_{k,g}(x_k) = 0$ otherwise. The penalty matrix for a random intercept is a $G$-dimensional identity matrix.
    \item \textit{Discrete spatial effects} use one-hot-coded dummy bases analogous to those for random intercepts, indicating the discrete spatial region to which $x_k$ belongs. The penalty matrix has the form of a Gaussian Markov random field \parencite{Rue2005-GaussianMarkovRandom}, where two regions are modeled as conditionally correlated if they share a border.
    \item \textit{Penalized splines} use B-spline basis evaluations for $\bfb_{k}(x_k)$, and the penalty matrix is defined as the cross-product of the $(L_k - 2) \times L_k$ second-order difference matrix $\bfD_{2}$, such that $\bfK^{(x)}_{k} = \bfD_{2}^\top \bfD_{2}$. Combining two marginal P-splines into a bivariate tensor product as described in \autoref{app:mvterms-examples} can construct smooth surfaces such as spatial effects or generalized interactions.
\end{itemize}

The complementary effects of within- and between-dimension regularization are demonstrated in the top row of \autoref{fig:mvterm-regularization-comparison} using penalized splines for the within-dimension and a first-order random-walk between-dimension penalty. Both directions receive a sum-to-zero constraint. Decreasing $\tau_{k,m}^2$ at fixed $\psi_k^2$ increasingly suppresses nonlinear within-dimension variation and drives the P-spline toward its centered linear null space. As $\psi_k^2$ decreases at fixed $\tau_{k,m}^2$, between-dimension variation is increasingly suppressed, ultimately collapsing all functions toward zero.

\section{Details for interpreting additive shape perturbations}
\label{app:interpretation-details}

This appendix derives the profile relationships used in
\autoref{sec:interpretation}. We retain the notation
$\bsdelta_c=\bsdelta_a+\bsdelta_b$ and
\begin{equation}
    \pi_{v,d}
    =
    \frac{\exp(\delta_{v,d})}
    {\sum_{k=1}^{D}\exp(\delta_{v,k})},
    \qquad
    v\in\{a,b,c\}.
\end{equation}

\subsection{Increment profiles and slopes}

On the spline core, define the cumulative bases
\begin{equation}
    A_j(r)
    =\sum_{\ell=j}^{J}B_\ell(r),
    \qquad
    j=2,\dots,J.
\end{equation}
The parameterization in \autoref{sec:trafo} implies that flexible increment
$d$ can be written as
\begin{equation}
    \exp(\vartheta_{v,d+4})=D\xi\pi_{v,d},
\end{equation}
whereas each fixed increment equals $\xi$. If $H_0$ collects the intercept
and the fixed-increment terms, the unity decomposition of the cubic B-spline
basis gives
\begin{equation}
    \tilde h_v(r)
    =
    H_0(r)
    +
    D\xi
    \sum_{d=1}^{D}\pi_{v,d}A_{d+4}(r).
    \label{eq:app-interpretation-profile-transformation}
\end{equation}
Let $\mathcal J_{\mathrm{fix}}=\{2,3,4,J-2,J-1,J\}$. Here $B_j^{(2)}$ denotes the $j$th quadratic B-spline basis function. Differentiating yields
\begin{equation}
    \tilde h_v'(r)
    =
    \sum_{j\in\mathcal J_{\mathrm{fix}}}B_j^{(2)}(r)
    +
    D\sum_{d=1}^{D}B_{d+4}^{(2)}(r)\pi_{v,d}.
    \label{eq:app-interpretation-profile-slope}
\end{equation}
The quadratic B-spline bases are nonnegative and form a partition of unity.
Consequently, the local slope is a smoothed average of the relative
control-point increments, which equal one for fixed increments and
$D\pi_{v,d}$ for flexible increments. The contribution associated with each profile entry affects the transformation slope only over the limited region where its B-spline basis function is nonzero. However, changing one shape parameter changes all profile entries through the softmax normalization. Locality therefore
applies to the increments, not to individual parameters in $\bsdelta_v$.

\subsection{Additive perturbations}

Let $\kappa_{ab}=\sum_{k=1}^{D}\pi_{a,k}\pi_{b,k}$. Addition of the shape
parameter vectors gives
\begin{align}
    \pi_{c,d}
    & =
    \frac{\exp(\delta_{a,d})\exp(\delta_{b,d})}
    {\sum_{k=1}^{D}\exp(\delta_{a,k})\exp(\delta_{b,k})}
    \\
    & =
    \frac{\pi_{a,d}\pi_{b,d}}{\kappa_{ab}}.
    \label{eq:app-interpretation-pooling}
\end{align}
Subtracting \eqref{eq:app-interpretation-profile-transformation} for the
baseline from the same identity for the combined profile gives
\begin{align}
    \tilde h_c(r)
    -\tilde h_a(r)
    & =
    D\xi\sum_{d=1}^{D}
    \bigl(\pi_{c,d}-\pi_{a,d}\bigr)A_{d+4}(r),
    \\
    \tilde h_c'(r)
    -\tilde h_a'(r)
    & =
    D\sum_{d=1}^{D}
    \bigl(\pi_{c,d}-\pi_{a,d}\bigr)B_{d+4}^{(2)}(r).
\end{align}
Both differences are zero outside the spline core because every
unstandardized onion transformation is the identity there.

Before standardization, the induced density is
$\tilde p_v(r)=p_Z\{\tilde h_v(r)\}\tilde h_v'(r)$. Hence, wherever the
reference density is positive, we have
\begin{equation}
    \log\frac{\tilde p_c(r)}{\tilde p_a(r)}
    =
    \log\frac{\tilde h_c'(r)}{\tilde h_a'(r)}
    +
    \log\frac{p_Z\{\tilde h_c(r)\}}{p_Z\{\tilde h_a(r)\}}.
    \label{eq:app-interpretation-raw-density-ratio}
\end{equation}
The first term is the change in the Jacobian contribution and the second is
the change in the reference-density contribution. The mean-scale
standardization in \autoref{sec:trafo} is then applied separately to the
distributions induced by $\bsdelta_a$ and $\bsdelta_c$. Thus, the exact pooling
of increment profiles does not translate into an additive density response;
the standardized density consequences remain nonlinear and baseline
dependent.

\FloatBarrier
\subsection{Density at the average predictor and average density}
\label{app:average-density-comparison}

\begin{figure}[H]
    \centering
    \includegraphics[width=\linewidth]{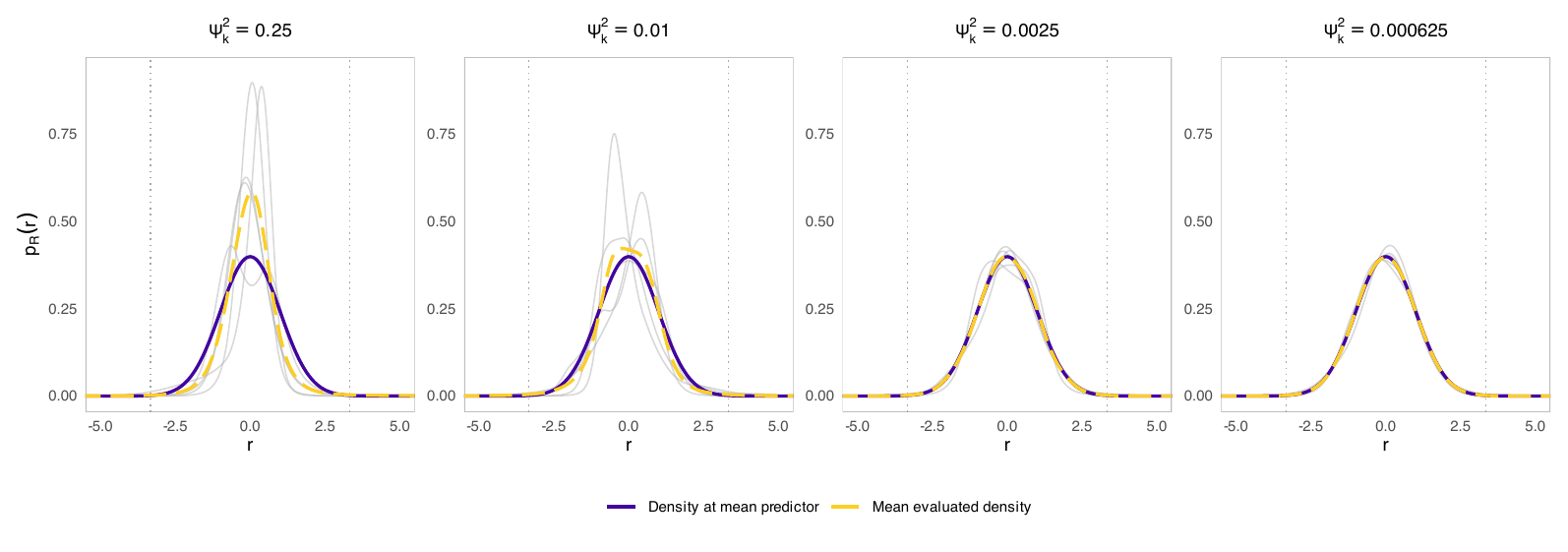}
    \caption{Density at the average transformation predictor versus the average of the evaluated densities for the seed-1 prior sample used in \autoref{fig:mvterm-regularization-comparison}. The columns correspond to the four values of $\psi_k^2$, with $\tau_{k,m}^2=0.04$ held fixed. The solid purple curve is the standardized latent-variable density evaluated at the empirical average of the transformation predictor over the 251 covariate values used for centering. This average is zero, so the resulting density equals the standard Gaussian reference. The dashed yellow curve is the equally weighted average of the 251 standardized conditional densities. Thin gray curves show the conditional densities at $x\in\{-2,-1,0,1,2\}$, and dotted vertical lines mark the onion-spline boundaries $a=-10/3$ and $b=10/3$.}
    \label{fig:mvterm-average-effect-density}
\end{figure}
\FloatBarrier

\section{Simulation details} \label{app:simulation-dgps}

\subsection{Scenario A}

For each replication, let $x_1,\dots,x_4$ be independent
$\operatorname{U}(-2,2)$ variables. The four response-level effects are
\begin{align*}
    f_1(x) & = -x + \pi\sin(\pi x),                   \\
    f_2(x) & = x + \frac{(2x)^2}{5.5},                \\
    f_3(x) & = x,                                     \\
    f_4(x) & = 0.5x + 15\phi\{2(x-0.2)\}-\phi(x+0.4),
\end{align*}
where $\phi$ is the standard Gaussian density. These functions produce,
respectively, oscillating, U-shaped, linear, and bell-shaped effects. We set
\begin{equation*}
    \mu(\bsx)=\sum_{j=1}^{4}f_j(x_j),
    \qquad
    \log\sigma(\bsx)=0.1\sum_{j=1}^{4}f_j(x_j),
    \qquad
    Y=\mu(\bsx)+\sigma(\bsx)R.
\end{equation*}

We consider two conditional distributions for $R$. In the Gaussian
subscenario, $R\sim\mathcal{N}(0,1)$ independently of $\bsx$. For the
skew--mixture subscenario, first let $S$ have a skew-normal distribution with
location zero, scale one, and shape parameter $\alpha=5$. With
\begin{equation*}
    d=\frac{\alpha}{\sqrt{1+\alpha^2}},
    \qquad
    m_S=d\sqrt{\frac{2}{\pi}},
    \qquad
    s_S^2=1-\frac{2d^2}{\pi},
\end{equation*}
define the standardized variable $\widetilde S=(S-m_S)/s_S$. Next, let
\begin{equation*}
    M\sim 0.5\,\mathcal{N}(-2,1^2)
    +0.5\,\mathcal{N}(1,0.5^2),
    \qquad
    m_M=-0.5,
    \qquad
    s_M^2=2.875,
\end{equation*}
and define $\widetilde M=(M-m_M)/s_M$. The conditional residual density is
\begin{equation*}
    p_R(r\mid x_1)
    =\{1-\omega(x_1)\}p_{\widetilde S}(r)
    +\omega(x_1)p_{\widetilde M}(r),
    \qquad
    \omega(x_1)=\operatorname{logit}^{-1}(2x_1).
\end{equation*}
Since both $\widetilde S$ and $\widetilde M$ have mean zero and variance one,
$R$ does as well for every value of $x_1$. Hence, the functions above are the
true conditional mean and log standard deviation in both subscenarios. Only
$x_1$ affects the residual shape; $x_2, x_3, x_4$ enter the fitted
transformation predictor as negative controls for its smoothness penalties.

\subsection{Scenario B}

Let $x_1,\dots,x_6$ be independent $\operatorname{U}(-1,1)$ variables. Define
two mixing functions,
\begin{align*}
    w(\bsx) = \operatorname{logit}^{-1}\{-0.2 + 1.7\sin(\pi x_1) - 1.2x_2 + x_1x_2\},
    \qquad
    v(\bsx) = \operatorname{logit}^{-1}(-0.4 + 2x_4),
\end{align*}
and mixture weights
\begin{equation*}
    \omega_1=0.68w(\bsx),
    \qquad
    \omega_2=0.32w(\bsx),
    \qquad
    \omega_3=\{1-w(\bsx)\}v(\bsx),
    \qquad
    \omega_4=\{1-w(\bsx)\}\{1-v(\bsx)\}.
\end{equation*}
The four component means are
\begin{align*}
    m_1(\bsx) & = -0.2+0.8\sin(\pi x_3)+0.4x_1,    \\
    m_2(\bsx) & = m_1(\bsx)-1.15-0.5x_2,           \\
    m_3(\bsx) & = 0.4+0.6x_3+1.1\cos(\pi x_1/2),   \\
    m_4(\bsx) & = m_3(\bsx)+1.25+0.7\sin(\pi x_4),
\end{align*}
and their standard deviations are
\begin{align*}
    s_1(\bsx) & = 0.35+0.10(x_4+1)/2,                       \\
    s_2(\bsx) & = 0.65+0.25\operatorname{logit}^{-1}(2x_4), \\
    s_3(\bsx) & = 0.30+0.12(x_2+1)/2,                       \\
    s_4(\bsx) & = 0.55+0.20(1-x_1)/2.
\end{align*}
Writing $\phi(y\mid m,s^2)$ for a Gaussian density, the response density is
\begin{equation*}
    p(y\mid\bsx)
    =\sum_{k=1}^{4}\omega_k(\bsx)
    \phi\{y\mid m_k(\bsx),s_k^2(\bsx)\}.
\end{equation*}
Thus, $x_1$ changes the main mixing function as well as component locations and scales;
$x_2$ changes the main mixing function and $s_3$, and separates two component means; $x_3$ moves
component locations; and $x_4$ changes the secondary mixing function, component scales,
and the position of the fourth component. The variables $x_5$ and $x_6$ have
no effect. The omitted $x_1x_2$ interaction in the main mixing function can
induce nonadditivity in the conditional mean, standard deviation, and
standardized shape, whereas the reported predictors for the mean, log standard
deviation, and standardized shape are additive.

For both scenarios, each replication contains $12{,}000$ observations. The nested
training sets contain the first $500$, $2{,}000$, or $10{,}000$ observations, and the final
$2{,}000$ observations form a common independent test set. We use $100$ independently
generated replications for every scenario and sample size.

\subsection{Additional simulation results}

\autoref{fig:sim-performance-nparam} assesses the sensitivity of the Scenario
B results to the number of transformation parameters. \autoref{fig:sim-densities-demo-n500} and
\autoref{fig:sim-densities-demo-n2000} complement
\autoref{fig:sim-densities-demo} by showing the same selected conditional
density estimates at the two smaller training-sample sizes.
\autoref{fig:sim-densities-controls} shows the term-wise density previews for
the three negative-control covariates in Scenario A.
\autoref{fig:sim-ci-width} reports the interval-width summaries discussed in
the main text.

\begin{figure}[tbp]
    \centering
    \includegraphics[width=\linewidth]{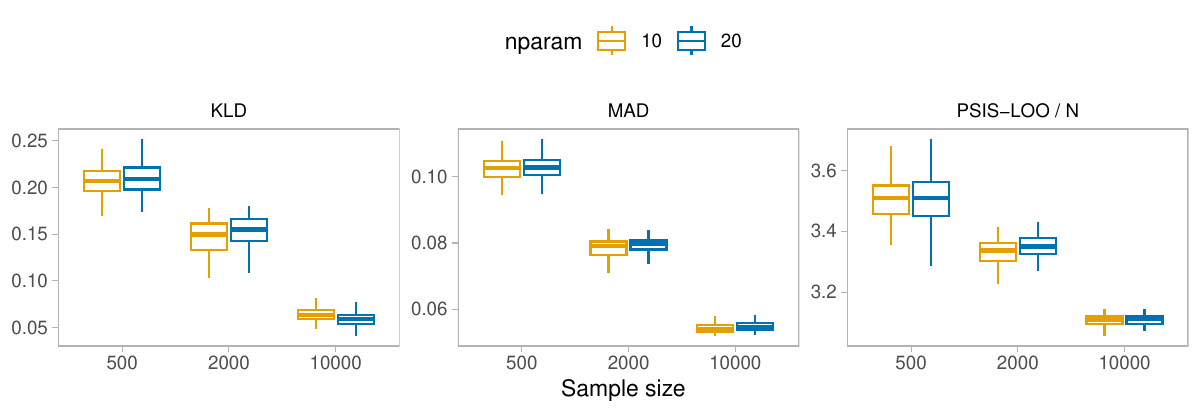}
    \caption{PTM response-scale distributional accuracy in Scenario B with
        $D=10$ and $D=20$ transformation parameters. Boxplots summarize $100$
        replications at each training-sample size. PSIS-LOO is reported on the
        deviance scale and divided by $N$; lower values are preferable in all
        three panels.}
    \label{fig:sim-performance-nparam}
\end{figure}

\begin{figure}[H]
    \centering
    \includegraphics[width=\linewidth]{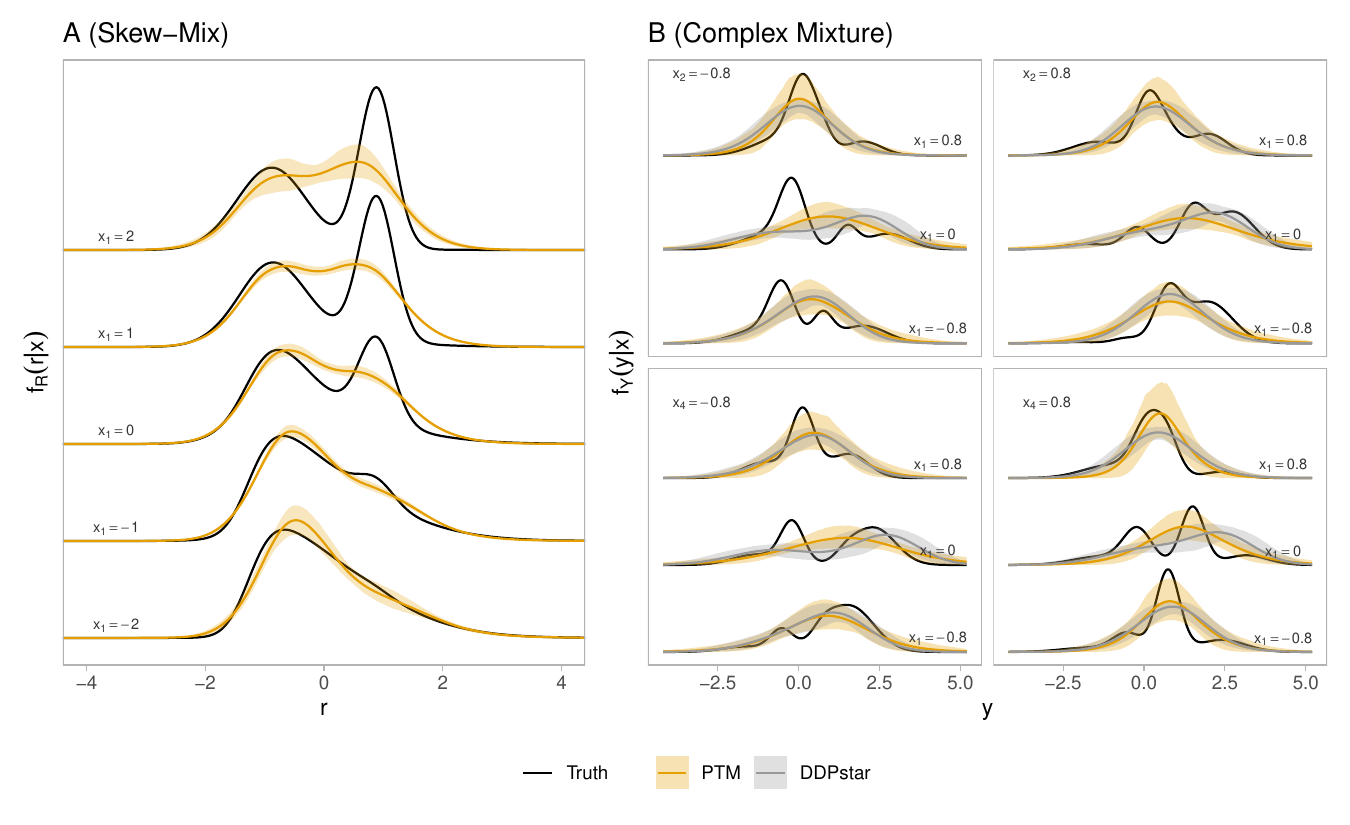}
    \caption{Selected conditional density estimates at $N=500$. The layout and
        graphical conventions match \autoref{fig:sim-densities-demo}.}
    \label{fig:sim-densities-demo-n500}
\end{figure}

\begin{figure}[H]
    \centering
    \includegraphics[width=\linewidth]{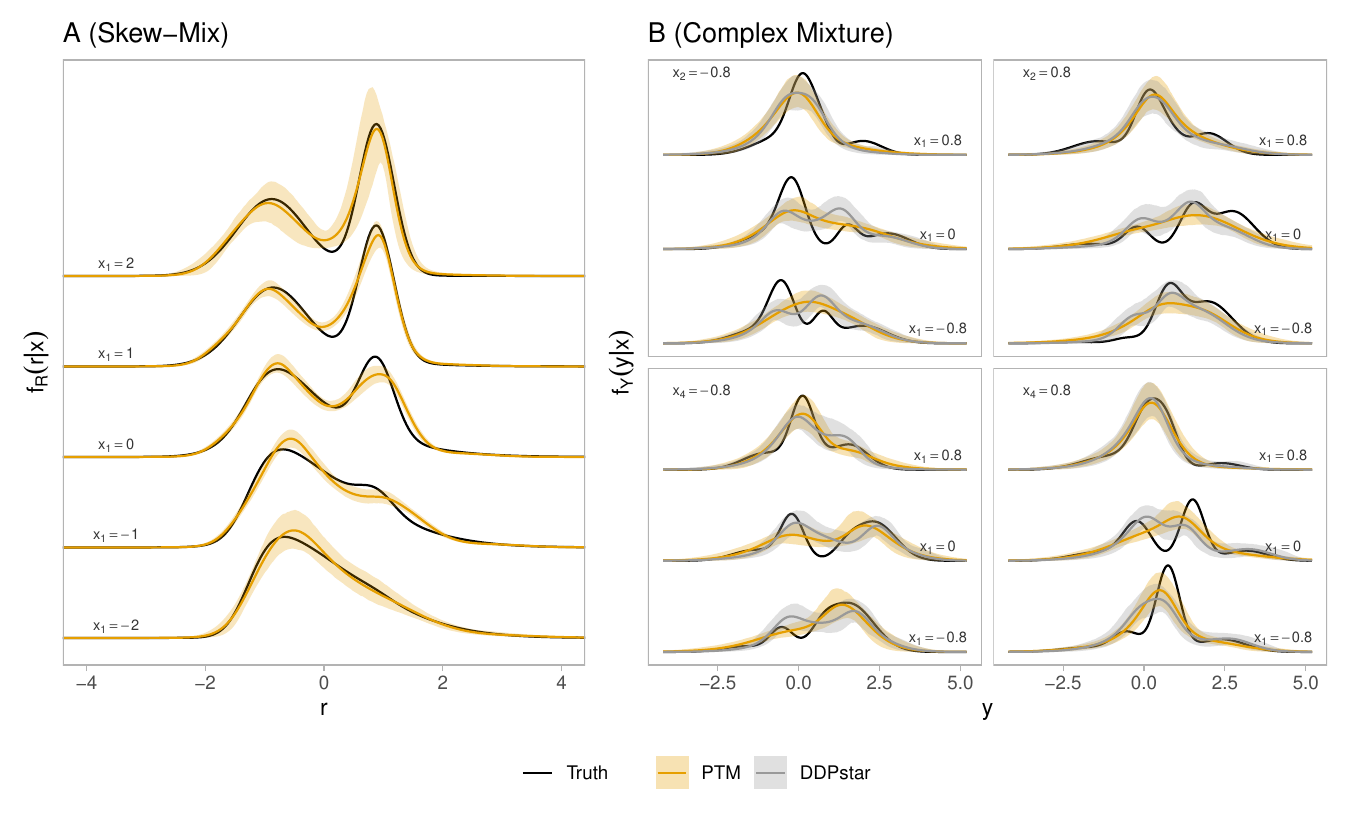}
    \caption{Selected conditional density estimates at $N=2{,}000$. The layout and
        graphical conventions match \autoref{fig:sim-densities-demo}.}
    \label{fig:sim-densities-demo-n2000}
\end{figure}

\begin{figure}[H]
    \centering
    \includegraphics[width=.8\linewidth]{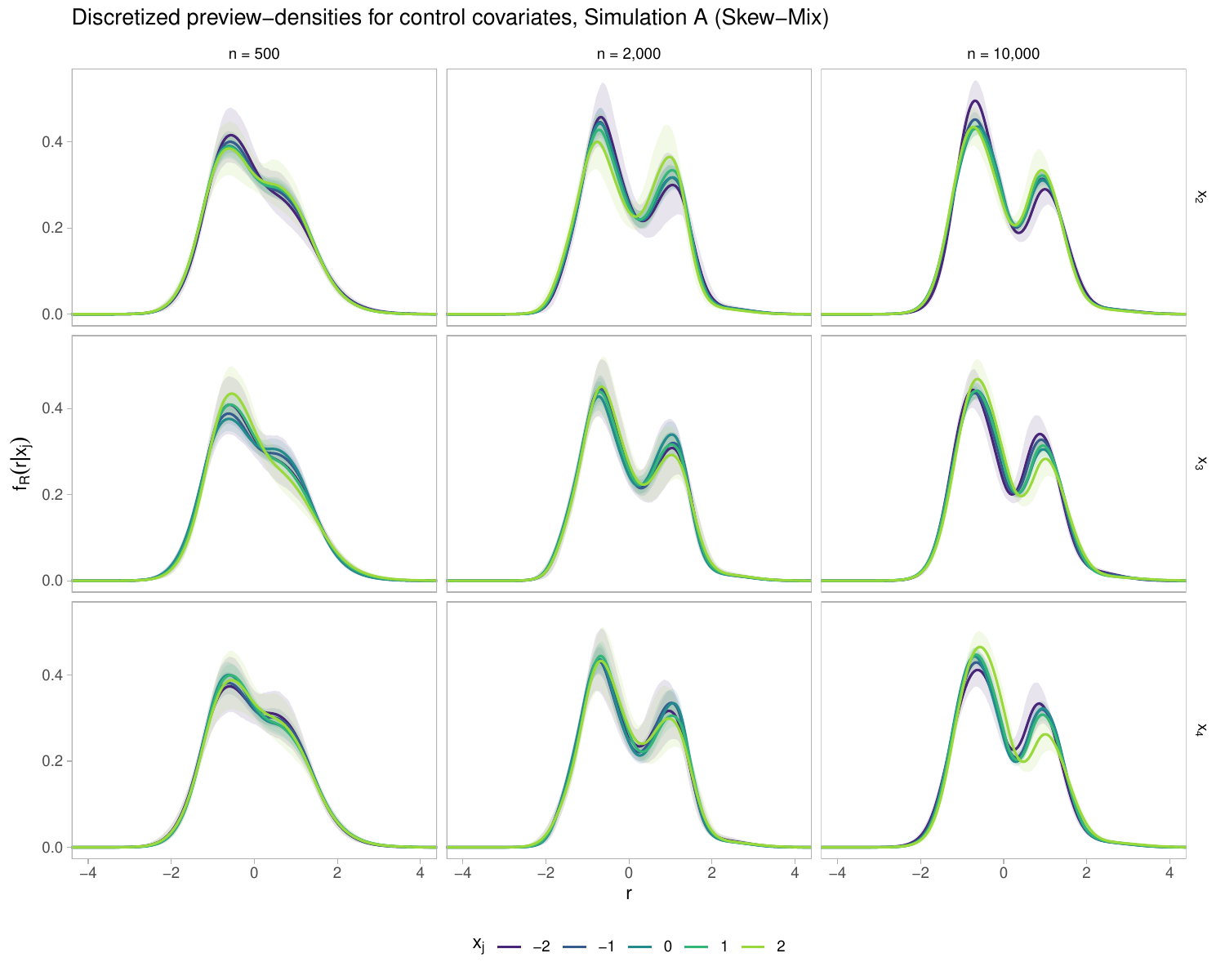}
    \caption{Discretized residual-density previews for the negative-control
        covariates in the skew--mixture subscenario of Scenario A. Rows
        correspond to $x_2$, $x_3$, and $x_4$, and columns correspond to
        $N\in\{500,2{,}000,10{,}000\}$. Within each panel, colored lines and
        bands show posterior means and 90\% pointwise credible intervals at
        focal values $-2, -1, 0, 1, 2$. Each preview is induced by the
        transformation intercept and the indicated covariate term, so it is a
        term-wise diagnostic rather than a full conditional density.}
    \label{fig:sim-densities-controls}
\end{figure}

\begin{figure}[H]
    \includegraphics[width=\linewidth]{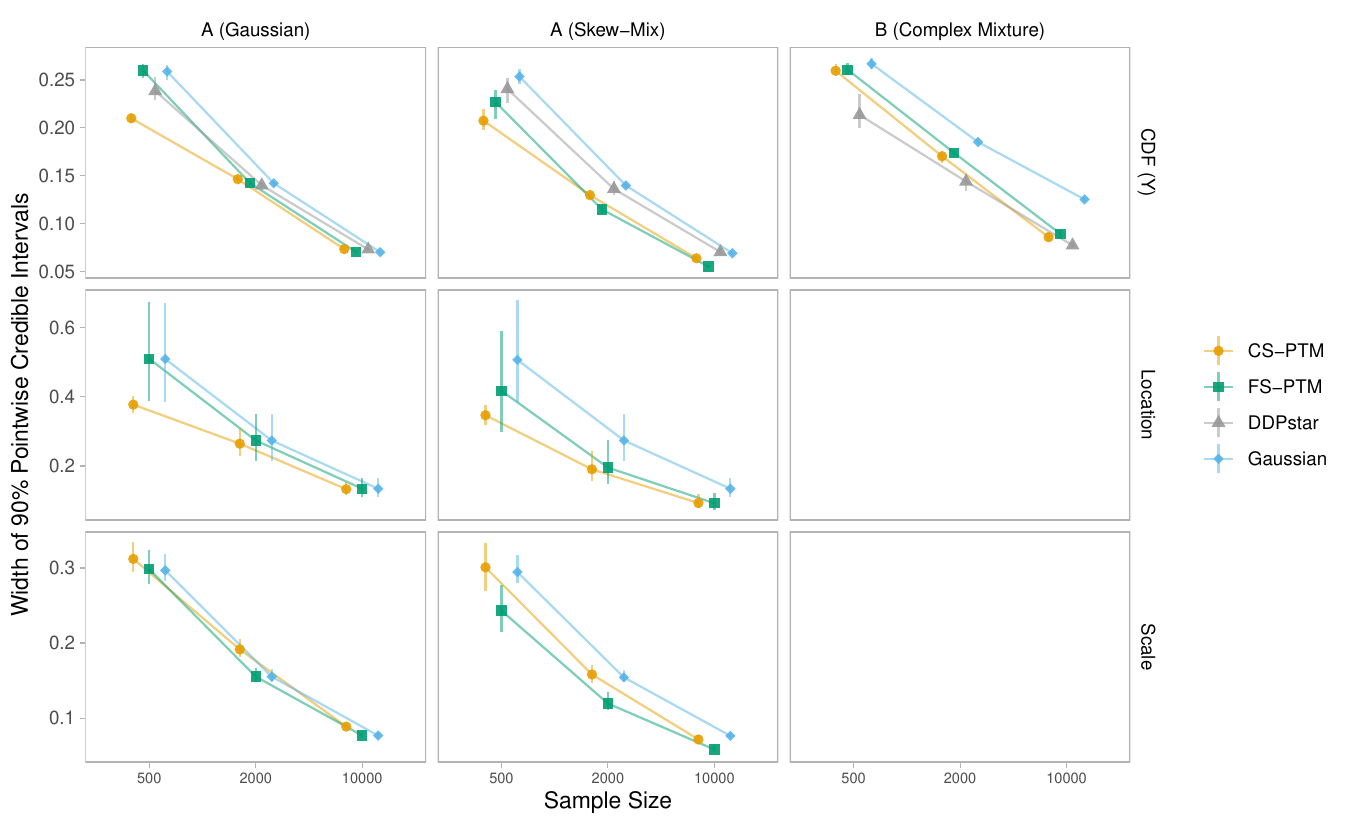}
    \caption{Mean widths of 90\% pointwise credible intervals by sample size.
        Vertical lines span the 0.05 and 0.95 quantiles over the $100$
        replications. Width is descriptive and must be considered alongside the coverage in \autoref{fig:sim-calibration}. In Scenario A,
        CS-PTM uses VI and the three competitors use MCMC; the similar VI fits of
        the FS-PTM and Gaussian model are omitted. In Scenario B, CS-PTM, FS-PTM, and Gaussian use VI, whereas DDPstar uses MCMC.
        CS-PTM and FS-PTM denote the conditional-shape and fixed-shape
        PTMs, respectively.}
    \label{fig:sim-ci-width}
\end{figure}

\FloatBarrier

\clearpage
\section{Additional application information}

\subsection{Optimization settings}
\label{app:application-optimization}

\subsubsection{Norwegian water conductivity application}
The two conditional MAP stages used the limited-memory Broyden--Fletcher--Goldfarb--Shanno (L-BFGS) algorithm with full batches for at most $20{,}000$
epochs, with relative tolerance $10^{-6}$ and early-stopping patience of five
epochs. The calibrated location and log-scale smoothing scales were initialized
at $\sqrt{\theta_{k,t}}$ and, during these initialization stages only, were
bounded below by $10^{-4}$. The three variational-inference stages used
mini-batches of $250$, $500$, and $500$ observations, respectively, with 16 Monte
Carlo draws per ELBO estimate and at most $5{,}000$ epochs per stage. Adam used a
warmup--cosine schedule from $10^{-6}$ to a peak learning rate of $10^{-4}$,
with warmup capped at $1{,}000$ optimizer updates and early-stopping patience of
100 epochs. The Gaussian comparison used patience 250 for its single
variational stage because it has no subsequent joint refinement.

\subsubsection{German temperature application}
All optimization stages used Adam with mini-batches of $3{,}778$ observations and at most
$5{,}000$ epochs per stage; variational inference used four Monte Carlo draws
per ELBO estimate. A warmup--cosine schedule increased the
learning rate from $10^{-6}$ to $3\times10^{-4}$ during preoptimization and to
$10^{-4}$ during variational inference. Warmup comprised 5\% of the planned
optimizer updates, capped at $2{,}000$ updates; the relative tolerance was
$10^{-5}$.
Early-stopping patience was five epochs during preoptimization and ten epochs
during variational inference. The location and log-scale main-effect smoothing
scales were initialized to yield approximately two effective degrees of
freedom\footnote{This conservative initialization was motivated by the absence
of data between approximately $1{,}500$ and $3{,}000$ m and was found to avoid
implausible interpolation of the elevation effect across this interval.}. 
During location--scale preoptimization,
optimization of the smoothing scales began at epoch 10 and early stopping was
disabled until epoch 60, ensuring at least 50 epochs in which coefficients and
smoothing scales were optimized jointly.

\subsection{PSIS-LOO diagnostics}
\label{app:application-loo-diagnostics}

\subsubsection{Norwegian water conductivity application}
Pareto-$k$ warnings occurred for both Norwegian conditional-shape PTM fits,
indicating that Pareto-smoothed importance sampling had difficulty reliably
approximating some leave-one-out posteriors. Such warnings may arise when
individual observations are influential, so that the full-data posterior is a
poor importance-sampling approximation to the corresponding leave-one-out
posterior, but may also reflect deficiencies in the variational posterior
approximation. The resulting PSIS-LOO estimates should therefore be interpreted
cautiously.

Using $3{,}000$ independent variational draws, $31$ of $13{,}425$ observations
($0.23\%$) exceeded the Pareto-$k$ threshold of $0.7$ for the pooled focal fit,
including $11$ observations with $k>1$. For the site-adjusted fit, $245$
observations ($1.82\%$) exceeded the threshold, including $88$ with $k>1$.
The maximum Pareto-$k$ values were $1.30$ and $2.65$, respectively, indicating
more severe instability for the site-adjusted fit.

For the pooled comparison, the conditional-shape PTM improves the reported
PSIS-LOO deviance over the next-best fixed-shape PTM by approximately
$5{,}494$, while the corresponding WAIC improvement is approximately
$5{,}525$. The large and consistent differences make it unlikely that the
qualitative conclusion is driven solely by the small fraction of problematic
observations. In the site-adjusted comparison, however, the reported PSIS-LOO
advantage is only approximately $140$ deviance units and decreases to
approximately $58$ when the conditional-shape PTM estimate is recomputed using
$3{,}000$ draws. We therefore do not interpret this comparison as strong
evidence for superior predictive performance. A $K$-fold
cross-validation analysis would be appropriate if a more definitive ranking
is desired for this analysis. Since the site-adjusted analysis is presented primarily as a
sensitivity analysis and is otherwise interpreted cautiously, we take this limitation
to not affect the main conclusions.

\subsubsection{German temperature application}
Using $3{,}000$ variational draws, Pareto-$k$ warnings occurred for 21 of the $1{,}182{,}514$ observations ($0.0018\%$), meaning that, for 21 observations, the threshold of $0.7$ was exceeded. No value was larger than $1$. Such warnings suggest that a few influential observations make it difficult for the importance-sampling approximation to estimate the corresponding leave-one-out posteriors reliably. The PSIS-LOO estimate should therefore be interpreted with some caution. Nevertheless, the advantage of the focal model under PSIS-LOO is very large and consistent with WAIC, making it unlikely that these localized warnings overturn the qualitative conclusion. A $K$-fold cross-validation analysis could provide a more definitive check.

\subsection{Additional results}
\label{app:application-results}

\autoref{tab:dwd-recent-fit-comparison} compares all completed production
fits from the recent reanalyses of the German daily-temperature and Norwegian
water-conductivity data.

\begin{table}[H]
    \centering
    \caption{Predictive comparison of the recent German daily-temperature and
        Norwegian water-conductivity fits. Reported quantities are the
        PSIS-LOO and WAIC deviances, their standard errors (SE), and the
        corresponding effective numbers of parameters $p_{\mathrm{eff}}$, as in
        \autoref{tab:application-model-comparison}.}
    \label{tab:dwd-recent-fit-comparison}
    \scriptsize
    \setlength{\tabcolsep}{3pt}
    \begin{tabular*}{\linewidth}{@{\extracolsep{\fill}}lccrrrrrr@{}}
        \toprule
        & & & \multicolumn{3}{c}{PSIS-LOO}
        & \multicolumn{3}{c}{WAIC} \\
        \cmidrule(lr){4-6}\cmidrule(l){7-9}
        Model & $D$ & $[a,b]$
        & Deviance $\downarrow$ & SE & $p_{\mathrm{eff}}$
        & Deviance $\downarrow$ & SE & $p_{\mathrm{eff}}$ \\
        \midrule
        \multicolumn{9}{@{}l}{\textbf{\sffamily Norwegian water conductivity}} \\
        \multicolumn{3}{@{}l}{\rule{0pt}{4ex}\textit{Without monitoring-site effects}}
        & & & & & & \\
        CS-PTM & 10 & $[-4,4]$
        & $19{,}125$ & $203.7$ & $260$
        & $19{,}117$ & $203.6$ & $256$ \\
        CS-PTM & 10 & $[-5,5]$
        & $19{,}571$ & $203.6$ & $237$
        & $19{,}560$ & $203.4$ & $232$ \\
        CS-PTM & 20 & $[-5,5]$
        & \textbf{$18{,}515$} & $210.6$ & $393$
        & \textbf{$18{,}481$} & $209.4$ & $376$ \\
        \addlinespace
        FS-PTM & 10 & $[-4,4]$
        & $24{,}108$ & $225.2$ & $129$
        & $24{,}105$ & $225.1$ & $128$ \\
        FS-PTM & 10 & $[-5,5]$
        & $24{,}009$ & $222.9$ & $109$
        & $24{,}007$ & $222.9$ & $108$ \\
        FS-PTM & 20 & $[-5,5]$
        & $24{,}091$ & $225.4$ & $167$
        & $24{,}087$ & $225.3$ & $165$ \\
        \addlinespace
        Gaussian & -- & --
        & $25{,}752$ & $230.9$ & $177$
        & $25{,}749$ & $231.3$ & $176$ \\
        \addlinespace
        \multicolumn{3}{@{}l}{\rule{0pt}{4ex}\textit{With monitoring-site effects}}
        & & & & & & \\
        CS-PTM & 10 & $[-4,4]$
        & \textbf{$-10{,}885$} & $245.4$ & $941$
        & \textbf{$-11{,}050$} & $242.8$ & $858$ \\
        CS-PTM & 10 & $[-5,5]$
        & $-10{,}742$ & $248.1$ & $922$
        & $-10{,}926$ & $244.9$ & $830$ \\
        CS-PTM & 20 & $[-5,5]$
        & $-10{,}493$ & $241.9$ & $980$
        & $-10{,}640$ & $241.2$ & $907$ \\
        \addlinespace
        FS-PTM & 10 & $[-4,4]$
        & $-10{,}746$ & $243.1$ & $720$
        & $-10{,}818$ & $242.5$ & $684$ \\
        FS-PTM & 10 & $[-5,5]$
        & $-10{,}713$ & $242.8$ & $740$
        & $-10{,}787$ & $242.3$ & $703$ \\
        FS-PTM & 20 & $[-5,5]$
        & $-10{,}522$ & $242.7$ & $830$
        & $-10{,}610$ & $242.2$ & $785$ \\
        \addlinespace
        Gaussian & -- & --
        & $-9{,}472$ & $273.1$ & $1{,}098$
        & $-9{,}826$ & $266.2$ & $921$ \\
        \addlinespace[0.75em]
        \multicolumn{9}{@{}l}{\textbf{\sffamily German daily temperature}} \\
        CS-PTM & 10 & $[-4,4]$
        & $1{,}467{,}667$ & $1{,}522.2$ & $7{,}361$
        & $1{,}467{,}572$ & $1{,}522.1$ & $7{,}314$ \\
        CS-PTM & 10 & $[-5,5]$
        & \textbf{$1{,}457{,}565$} & $1{,}515.8$ & $5{,}907$
        & \textbf{$1{,}457{,}472$} & $1{,}515.6$ & $5{,}860$ \\
        CS-PTM & 15 & $[-5,5]$
        & $1{,}458{,}648$ & $1{,}522.5$ & $8{,}005$
        & $1{,}458{,}400$ & $1{,}521.7$ & $7{,}881$ \\
        CS-PTM & 20 & $[-5,5]$
        & $1{,}464{,}637$ & $1{,}524.5$ & $11{,}007$
        & $1{,}464{,}304$ & $1{,}523.3$ & $10{,}840$ \\
        \addlinespace
        FS-PTM & 10 & $[-4,4]$
        & $1{,}505{,}028$ & $1{,}565.9$ & $3{,}155$
        & $1{,}504{,}998$ & $1{,}565.8$ & $3{,}140$ \\
        FS-PTM & 10 & $[-5,5]$
        & $1{,}505{,}006$ & $1{,}565.3$ & $3{,}144$
        & $1{,}504{,}976$ & $1{,}565.2$ & $3{,}130$ \\
        FS-PTM & 20 & $[-5,5]$
        & $1{,}504{,}869$ & $1{,}565.9$ & $3{,}163$
        & $1{,}504{,}839$ & $1{,}565.9$ & $3{,}148$ \\
        \addlinespace
        Gaussian & -- & --
        & $1{,}511{,}783$ & $1{,}572.8$ & $4{,}833$
        & $1{,}511{,}737$ & $1{,}572.7$ & $4{,}810$ \\
        \bottomrule
    \end{tabular*}

    \vspace{0.35em}
    \begin{minipage}{\linewidth}
        \footnotesize
        \textit{Note.} All German fits use latitude and longitude main-effect
        basis dimension $k_{\mathrm{sp}}=10$ in the location and log-scale
        predictors. $D$ is the number of free transformation parameters, and
        $[a,b]$ is the onion-spline core. Dashes indicate quantities that do
        not apply to the Gaussian fits.
        CS-PTM and FS-PTM denote the conditional-shape and fixed-shape
        PTMs, respectively.
    \end{minipage}
\end{table}

\begin{table}[H]
    \centering
    \caption{Posterior mean complete conditional quantiles in degrees Celsius
        for the nine focal locations at city-specific elevations on four
        representative days. The January 15 and July 15 entries
        correspond to the colored rug marks in the lower panels of
        \autoref{fig:dwd-temperature-quantile-curves},
        \autoref{fig:dwd-basic-ptm-quantile-curves}, and
        \autoref{fig:dwd-gaussian-quantile-curves}. CS-PTM denotes the focal
        conditional-shape PTM, FS-PTM the best fixed-shape PTM ($D=20$ and
        onion-spline core $[-5,5]$), and Gauss the Gaussian location--scale
        model.}
    \label{tab:dwd-complete-conditional-quantiles}
    \scriptsize
    \setlength{\tabcolsep}{3pt}
    \begin{tabular*}{\linewidth}{@{\extracolsep{\fill}}l*{9}{r}@{}}
        \toprule
        & \multicolumn{3}{c}{$p=0.01$}
        & \multicolumn{3}{c}{$p=0.5$}
        & \multicolumn{3}{c}{$p=0.99$} \\
        \cmidrule(lr){2-4}\cmidrule(lr){5-7}\cmidrule(l){8-10}
        Location
        & CS-PTM & FS-PTM & Gauss
        & CS-PTM & FS-PTM & Gauss
        & CS-PTM & FS-PTM & Gauss \\
        \midrule
        \multicolumn{10}{@{}l}{\textit{Spring (April 15)}} \\
        Meppen & 1.75 & 1.24 & 1.11 & 9.75 & 9.80 & 9.73 & 18.81 & 18.44 & 18.35 \\
        Hamburg & 1.32 & 0.80 & 0.54 & 9.06 & 9.08 & 8.94 & 17.77 & 17.44 & 17.34 \\
        Berlin & 1.27 & 0.84 & 0.73 & 9.61 & 9.60 & 9.70 & 18.72 & 18.44 & 18.67 \\
        Cologne & 2.20 & 1.72 & 1.52 & 10.81 & 10.61 & 10.56 & 19.90 & 19.58 & 19.61 \\
        G\"ottingen & 0.73 & 0.51 & 0.20 & 9.29 & 9.25 & 9.08 & 17.89 & 18.07 & 17.96 \\
        Dresden & 0.72 & 0.31 & 0.09 & 9.58 & 9.36 & 9.32 & 18.50 & 18.50 & 18.55 \\
        Frankfurt & 2.44 & 2.28 & 1.93 & 10.94 & 10.82 & 10.65 & 19.36 & 19.44 & 19.38 \\
        Stuttgart & 1.47 & 1.49 & 1.44 & 10.53 & 10.37 & 10.32 & 19.28 & 19.34 & 19.19 \\
        Munich & 0.05 & 0.16 & 0.28 & 9.49 & 9.38 & 9.22 & 18.33 & 18.68 & 18.16 \\
        \addlinespace[0.5em]
        \multicolumn{10}{@{}l}{\textit{Summer (July 15)}} \\
        Meppen & 12.03 & 11.00 & 10.69 & 17.95 & 18.29 & 18.14 & 27.01 & 25.66 & 25.59 \\
        Hamburg & 11.79 & 10.89 & 10.44 & 17.75 & 18.10 & 18.14 & 26.51 & 25.38 & 25.83 \\
        Berlin & 13.03 & 12.22 & 11.97 & 19.66 & 19.90 & 19.80 & 29.02 & 27.65 & 27.62 \\
        Cologne & 12.95 & 11.81 & 11.64 & 19.31 & 19.53 & 19.32 & 28.64 & 27.32 & 27.01 \\
        G\"ottingen & 12.17 & 11.40 & 11.22 & 18.40 & 18.63 & 18.53 & 27.14 & 25.94 & 25.84 \\
        Dresden & 13.06 & 12.38 & 11.96 & 19.53 & 19.61 & 19.44 & 28.02 & 26.91 & 26.92 \\
        Frankfurt & 13.78 & 12.94 & 12.67 & 20.13 & 20.38 & 20.17 & 28.80 & 27.89 & 27.67 \\
        Stuttgart & 13.52 & 12.81 & 12.64 & 20.09 & 20.28 & 20.25 & 28.41 & 27.82 & 27.87 \\
        Munich & 11.88 & 11.52 & 11.52 & 18.90 & 19.10 & 19.09 & 26.84 & 26.76 & 26.67 \\
        \addlinespace[0.5em]
        \multicolumn{10}{@{}l}{\textit{Fall (October 15)}} \\
        Meppen & 3.70 & 3.74 & 3.46 & 10.78 & 10.82 & 10.71 & 18.10 & 17.97 & 17.96 \\
        Hamburg & 3.47 & 3.69 & 3.63 & 10.42 & 10.41 & 10.47 & 17.30 & 17.20 & 17.31 \\
        Berlin & 3.50 & 3.54 & 3.39 & 10.85 & 10.72 & 10.72 & 18.20 & 17.98 & 18.05 \\
        Cologne & 4.36 & 4.59 & 4.39 & 11.89 & 12.05 & 11.96 & 19.70 & 19.60 & 19.53 \\
        G\"ottingen & 2.84 & 3.17 & 3.03 & 10.76 & 10.78 & 10.80 & 18.68 & 18.46 & 18.57 \\
        Dresden & 2.42 & 2.87 & 3.01 & 10.57 & 10.50 & 10.67 & 18.30 & 18.20 & 18.33 \\
        Frankfurt & 3.70 & 4.09 & 4.00 & 11.13 & 11.16 & 11.15 & 18.50 & 18.30 & 18.31 \\
        Stuttgart & 3.17 & 3.57 & 3.39 & 11.23 & 11.22 & 11.20 & 18.95 & 18.95 & 19.01 \\
        Munich & 1.38 & 1.91 & 1.72 & 9.62 & 9.60 & 9.51 & 17.27 & 17.37 & 17.30 \\
        \addlinespace[0.5em]
        \multicolumn{10}{@{}l}{\textit{Winter (January 15)}} \\
        Meppen & -8.36 & -6.69 & -6.94 & 3.09 & 2.82 & 2.95 & 11.41 & 12.42 & 12.83 \\
        Hamburg & -8.93 & -7.07 & -6.99 & 2.48 & 2.17 & 2.20 & 10.53 & 11.50 & 11.39 \\
        Berlin & -10.74 & -8.46 & -8.68 & 1.72 & 1.41 & 1.37 & 10.26 & 11.37 & 11.42 \\
        Cologne & -7.56 & -5.27 & -5.36 & 3.66 & 3.55 & 3.69 & 11.88 & 12.46 & 12.74 \\
        G\"ottingen & -10.17 & -7.68 & -8.09 & 2.35 & 2.10 & 1.94 & 11.10 & 11.99 & 11.98 \\
        Dresden & -10.41 & -8.09 & -8.00 & 1.43 & 1.14 & 1.22 & 9.46 & 10.46 & 10.43 \\
        Frankfurt & -9.66 & -6.85 & -7.07 & 2.63 & 2.45 & 2.45 & 11.29 & 11.85 & 11.98 \\
        Stuttgart & -10.53 & -7.31 & -7.44 & 2.27 & 2.17 & 2.17 & 11.30 & 11.75 & 11.77 \\
        Munich & -13.30 & -9.54 & -9.91 & 0.09 & 0.05 & -0.08 & 9.22 & 9.73 & 9.75 \\
        \bottomrule
    \end{tabular*}
\end{table}

\begin{figure}[H]
    \centering
    \includegraphics[width=\linewidth]{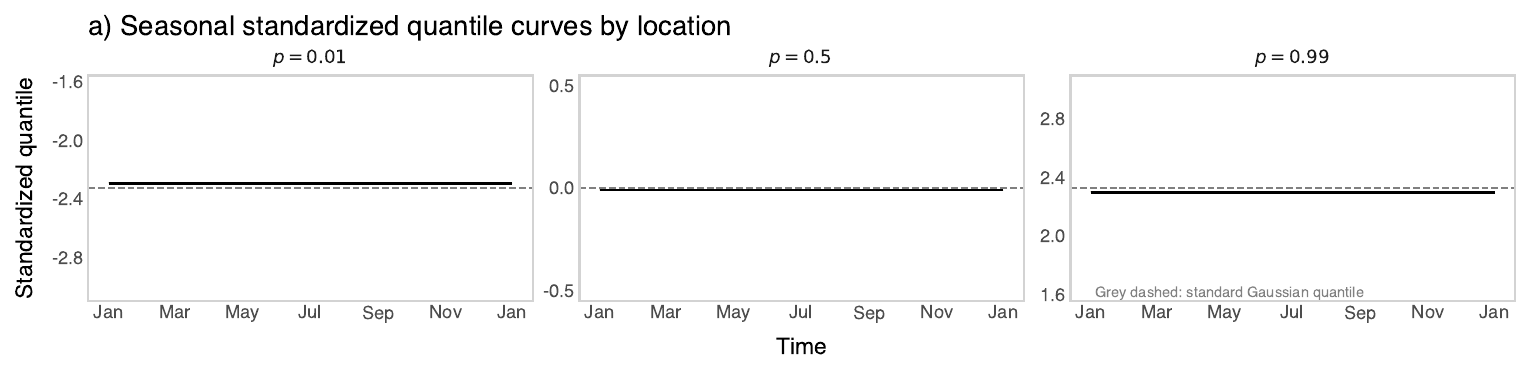}
    \par\vspace{0.35em}
    \includegraphics[width=\linewidth]{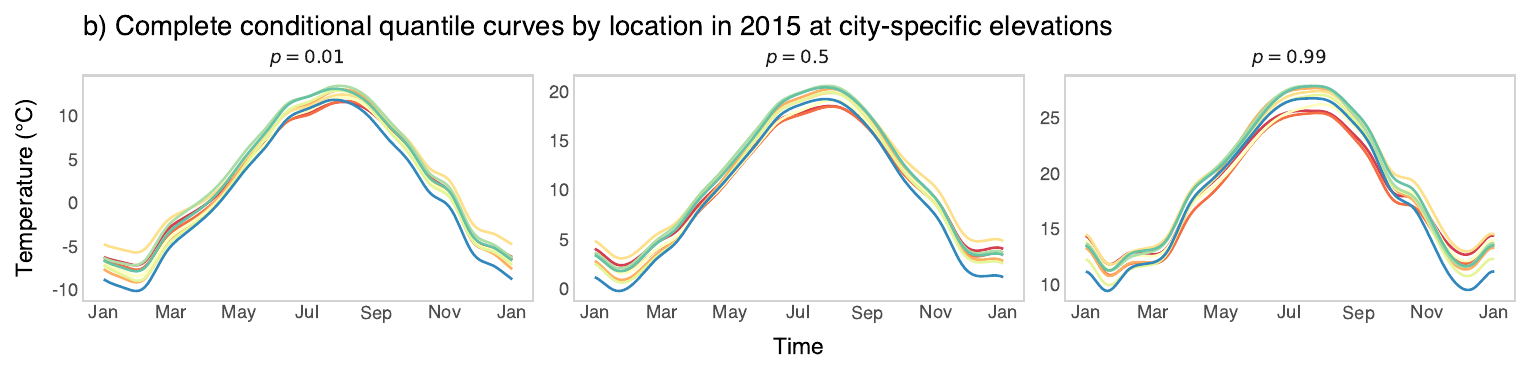}
    \par\vspace{0.35em}
    \includegraphics[width=\linewidth]{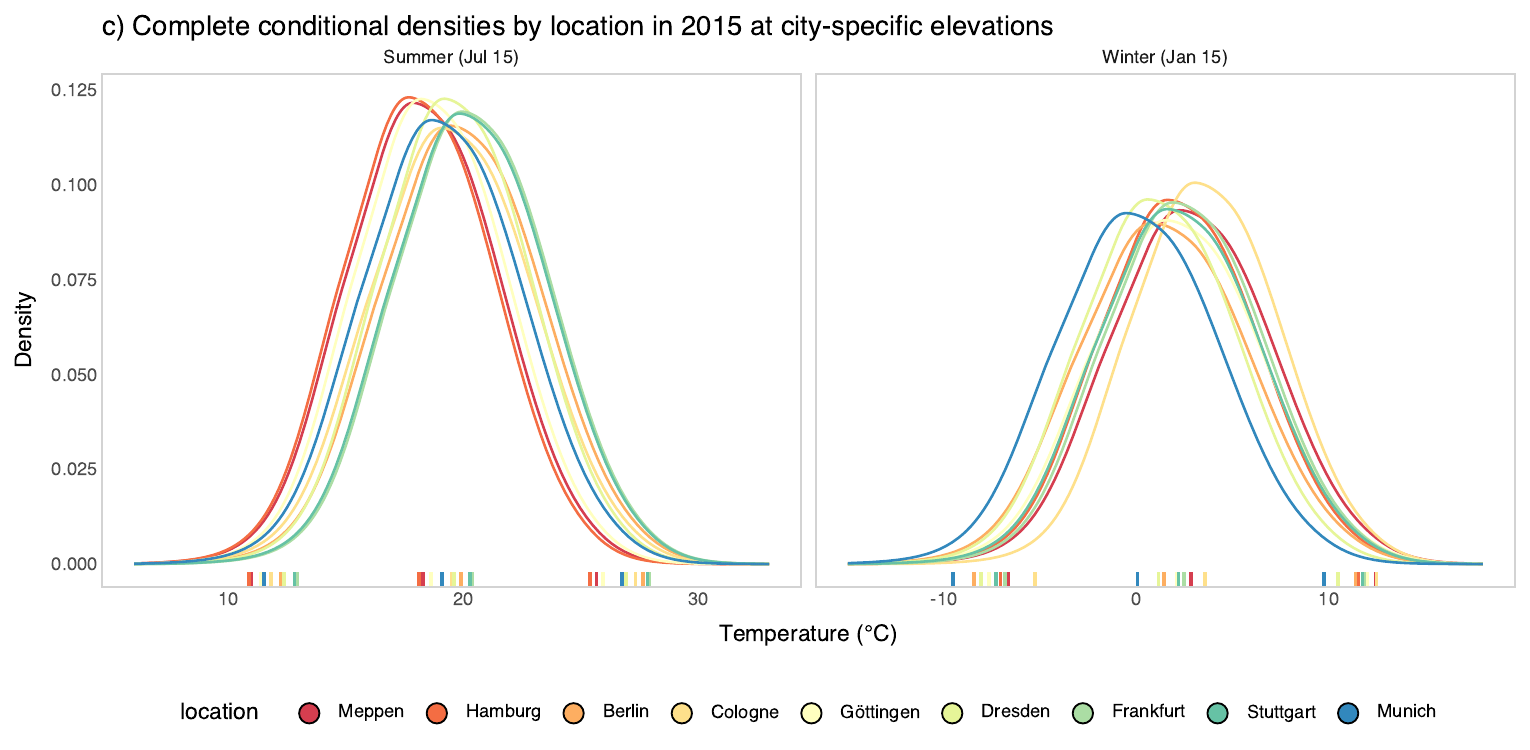}
    \caption{Seasonal distribution summaries for the best fixed-shape PTM in the
        German daily-temperature application ($D=20$ and onion-spline core
        $[-5,5]$), using the layout and locations of
        \autoref{fig:dwd-temperature-quantile-curves}. The top row shows
        posterior mean standardized quantiles at
        $p\in\{0.01,0.5,0.99\}$; they are constant over time and location
        because the fixed-shape PTM uses a covariate-invariant transformation, and
        gray dashed lines mark the corresponding standard Gaussian quantiles.
        The middle row shows complete conditional quantiles in degrees Celsius
        during 2015 at the city-specific elevations. The lower panels show the
        corresponding complete conditional densities on two
        representative days. Colored rug marks on the horizontal axes indicate the
        location-specific $p\in\{0.01,0.5,0.99\}$ quantiles. Colors identify
        locations using the shared legend below the density panels.}
    \label{fig:dwd-basic-ptm-quantile-curves}
\end{figure}

\begin{figure}[H]
    \centering
    \includegraphics[width=\linewidth]{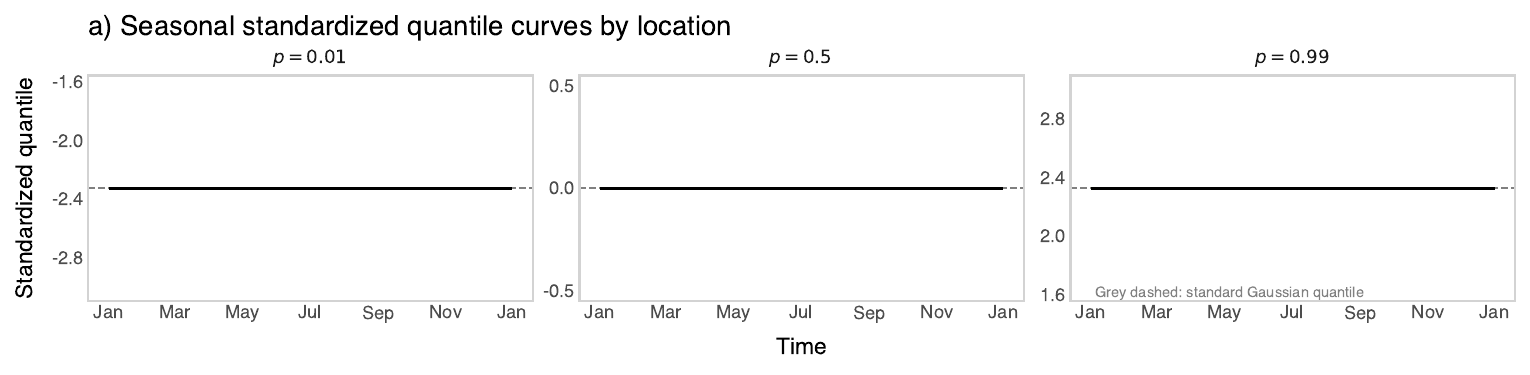}
    \par\vspace{0.35em}
    \includegraphics[width=\linewidth]{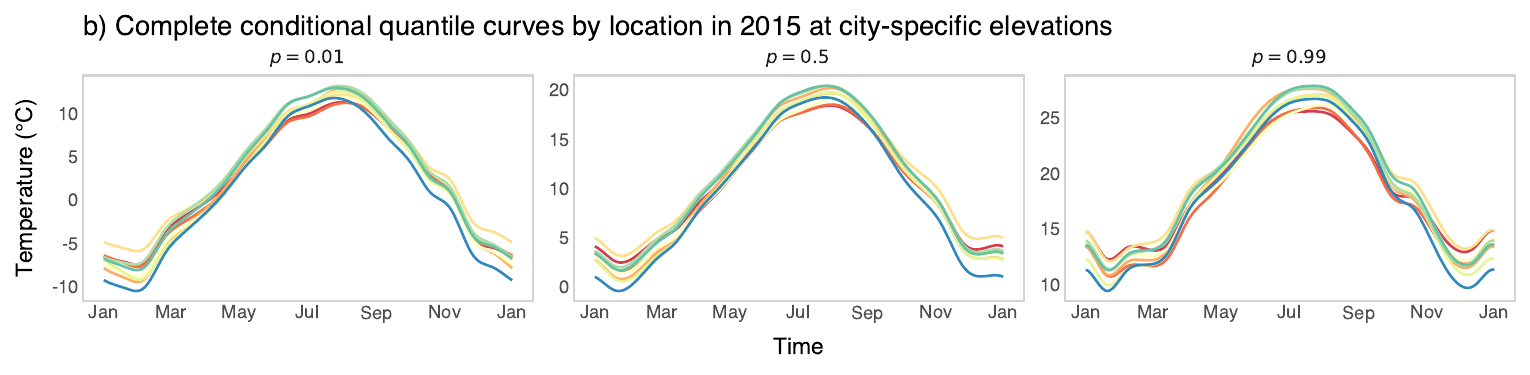}
    \par\vspace{0.35em}
    \includegraphics[width=\linewidth]{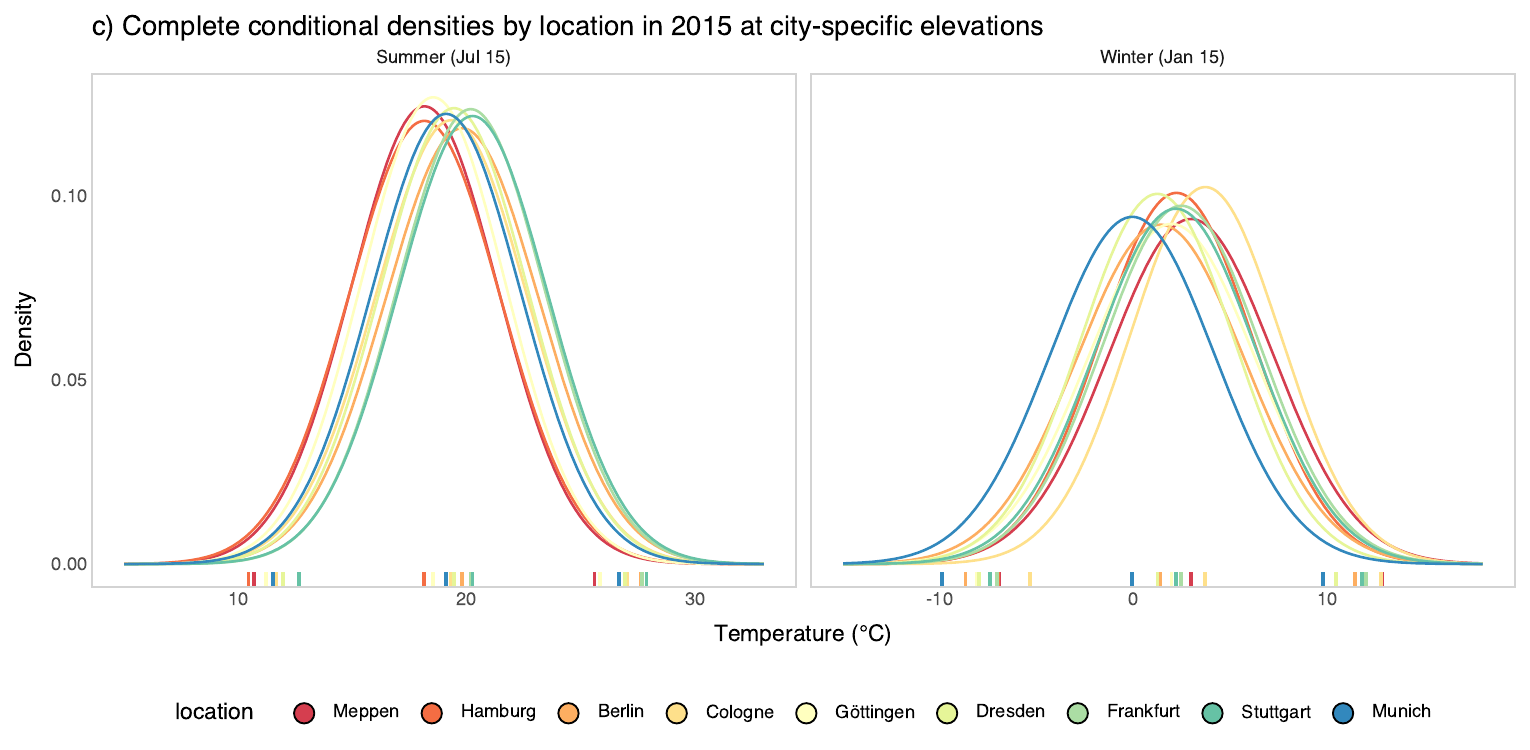}
    \caption{Seasonal distribution summaries for the Gaussian location--scale
        model in the German daily-temperature application, using the layout
        and locations of \autoref{fig:dwd-temperature-quantile-curves}. The top
        row shows posterior mean standardized quantiles at
        $p\in\{0.01,0.5,0.99\}$; they are constant over time and location and
        coincide with the corresponding standard Gaussian quantiles by
        construction. The middle row shows complete conditional quantiles in
        degrees Celsius during 2015 at the city-specific elevations. The lower
        panels show the corresponding complete conditional densities on
        two representative days.
        Colored rug marks on the horizontal axes indicate
        the location-specific $p\in\{0.01,0.5,0.99\}$ quantiles. Colors
        identify locations using the shared legend below the density panels.}
    \label{fig:dwd-gaussian-quantile-curves}
\end{figure}

\begin{figure}[H]
    \centering
    \includegraphics[width=\linewidth]{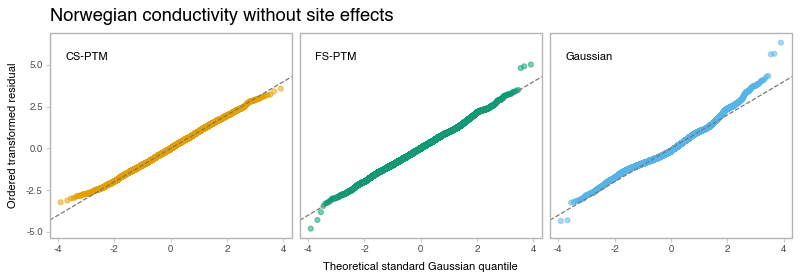}\par
    \includegraphics[width=\linewidth]{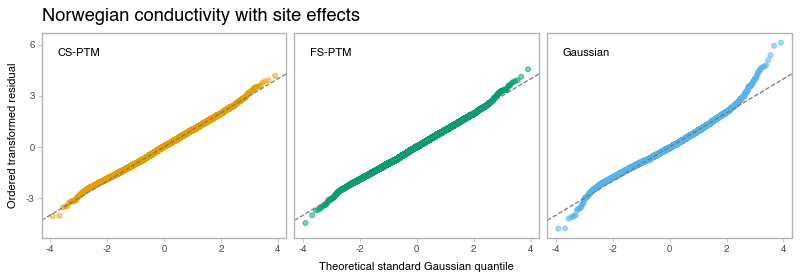}\par
    \includegraphics[width=\linewidth]{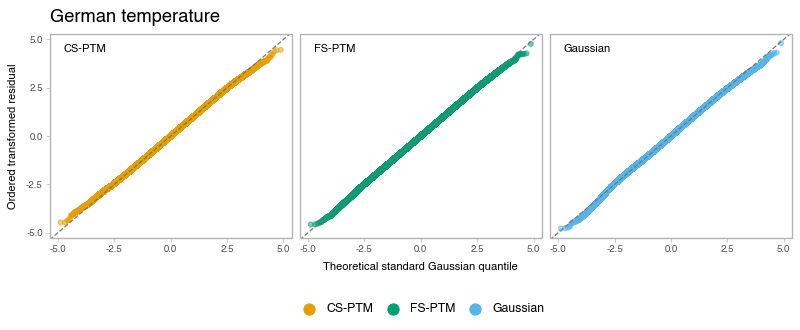}
    \caption{Q--Q plots of posterior mean transformed residuals for the best
        CS-PTM, FS-PTM, and Gaussian model within each application block
        of \autoref{tab:dwd-recent-fit-comparison}, ordered from top to bottom as
        Norwegian conductivity without monitoring-site effects, Norwegian
        conductivity with monitoring-site effects, and German temperature.
        Points show the complete
        ordered transformed residuals against their theoretical standard
        Gaussian quantiles; gray dashed lines mark the identity. The German comparison uses the CS-PTM with $D=10$ and $[a,b]=[-5,5]$, the FS-PTM with $D=20$ and
        $[a,b]=[-5,5]$, and the Gaussian model, all with
        $k_{\mathrm{sp}}=10$. Without monitoring-site effects, the Norwegian
        comparison uses the CS-PTM with $D=20$ and $[a,b]=[-5,5]$, the
        FS-PTM with $D=10$ and $[a,b]=[-5,5]$, and the Gaussian model. With
        monitoring-site effects, both Norwegian transformation models use
        $D=10$ and $[a,b]=[-4,4]$. CS-PTM and FS-PTM denote the
        conditional-shape and fixed-shape PTMs, respectively.}
    \label{fig:application-qq}
\end{figure}

\autoref{fig:norway-effects-site-ri} complements the focal Norwegian
water-conductivity results in \autoref{fig:norway-effects} by showing the same
effect summaries after monitoring-site random intercepts are added to the
location and log-scale predictors.

\begin{figure}[H]
    \centering
    \includegraphics[width=0.32\linewidth]{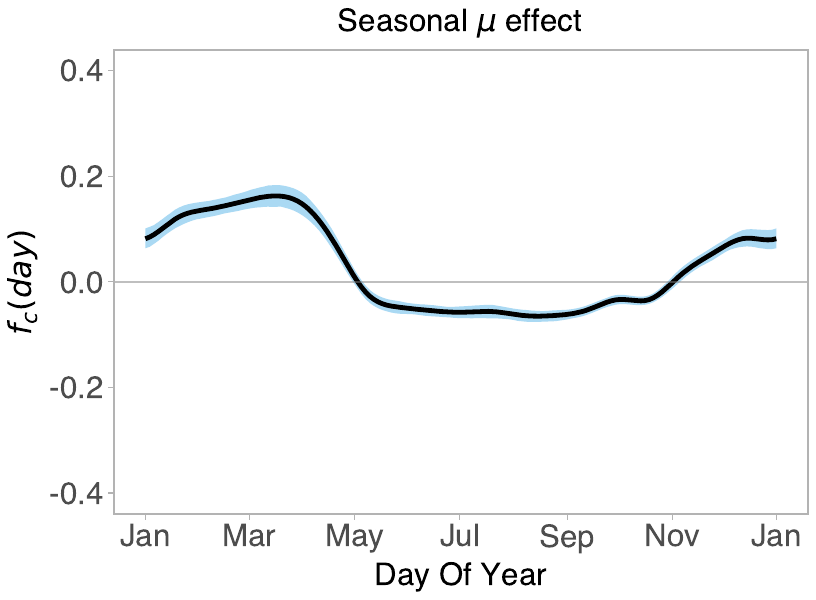}\hfill
    \includegraphics[width=0.32\linewidth]{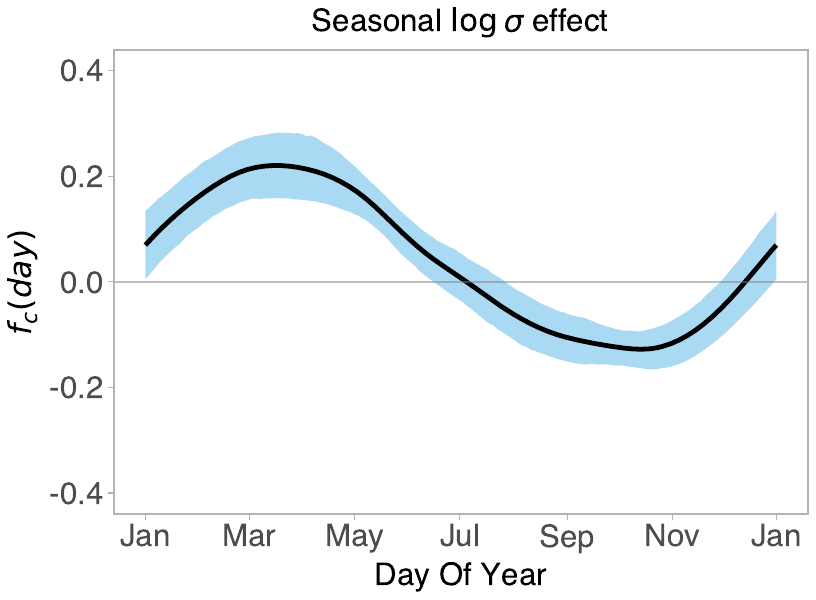}\hfill
    \includegraphics[width=0.32\linewidth]{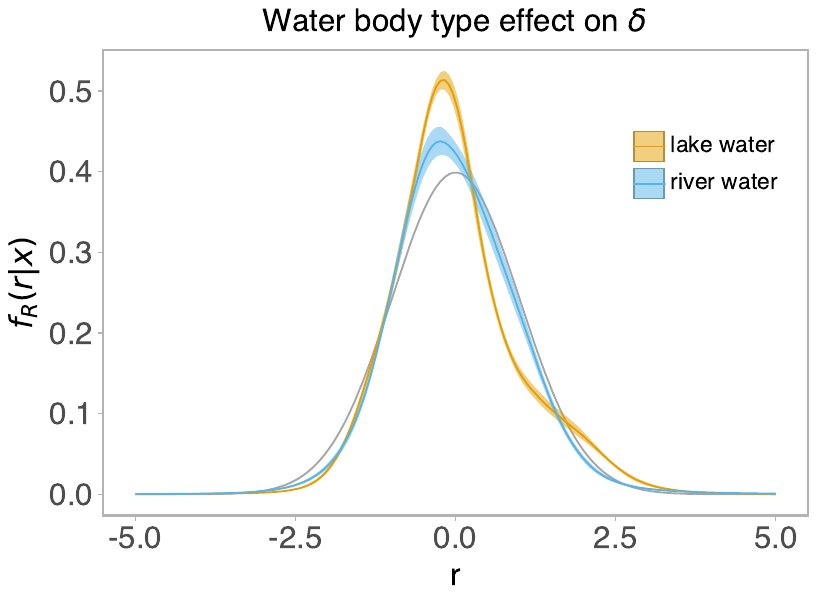}
    \par\vspace{0.75em}
    \includegraphics[width=0.32\linewidth]{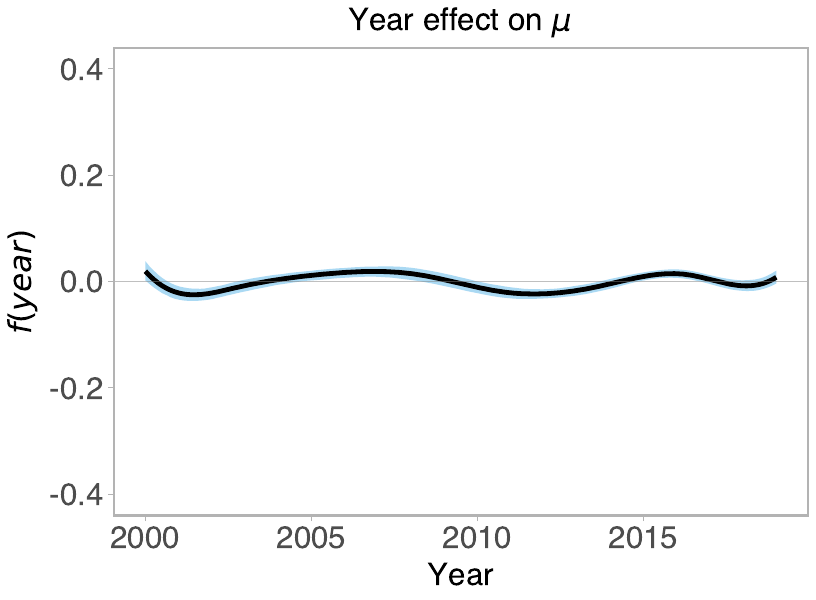}\hfill
    \includegraphics[width=0.32\linewidth]{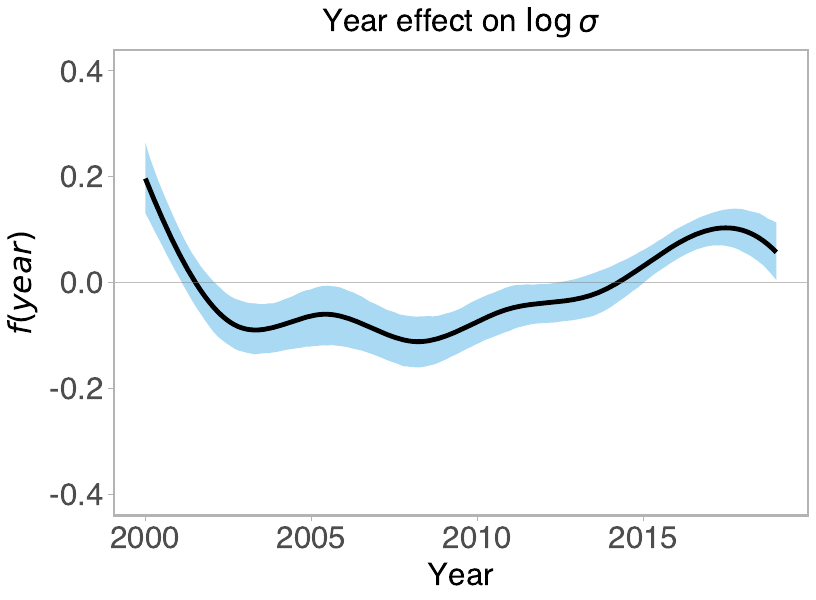}\hfill
    \includegraphics[width=0.32\linewidth]{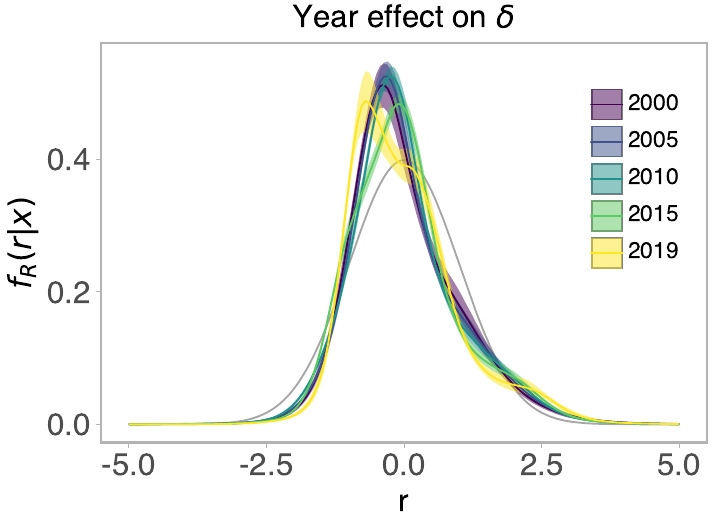}
    \par\vspace{0.75em}
    \includegraphics[width=0.32\linewidth]{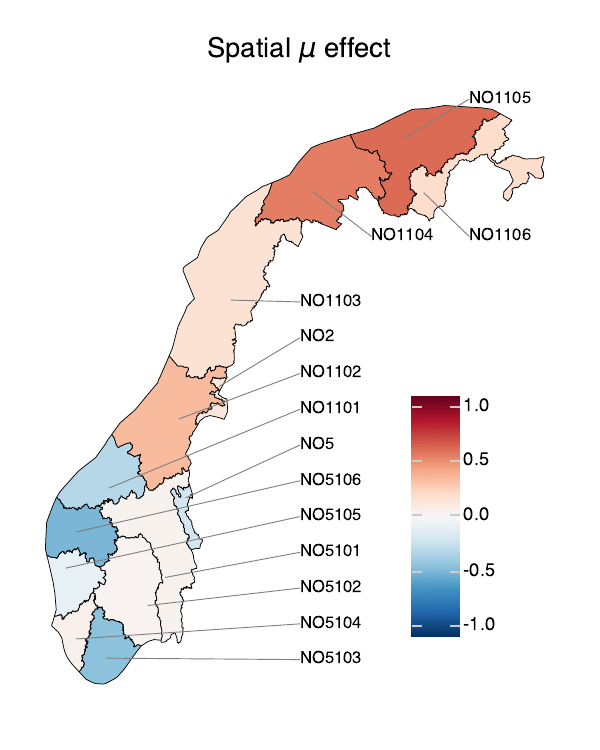}\hfill
    \includegraphics[width=0.32\linewidth]{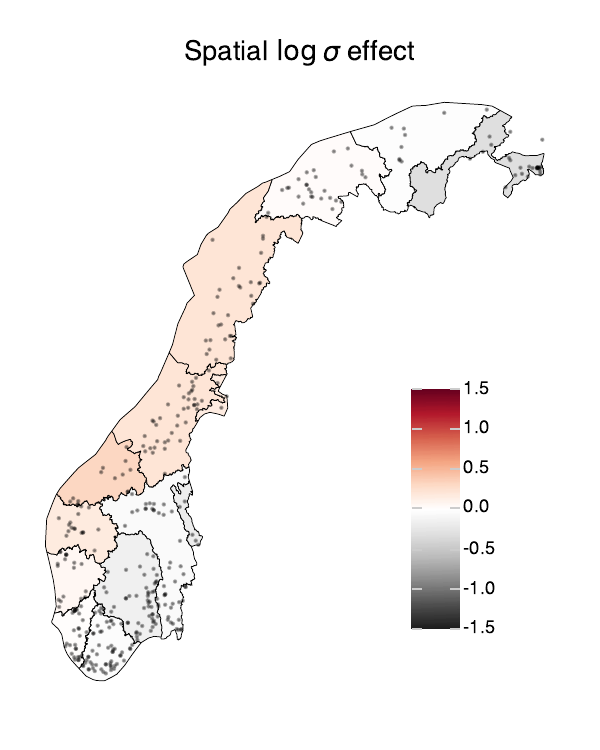}\hfill
    \includegraphics[width=0.32\linewidth]{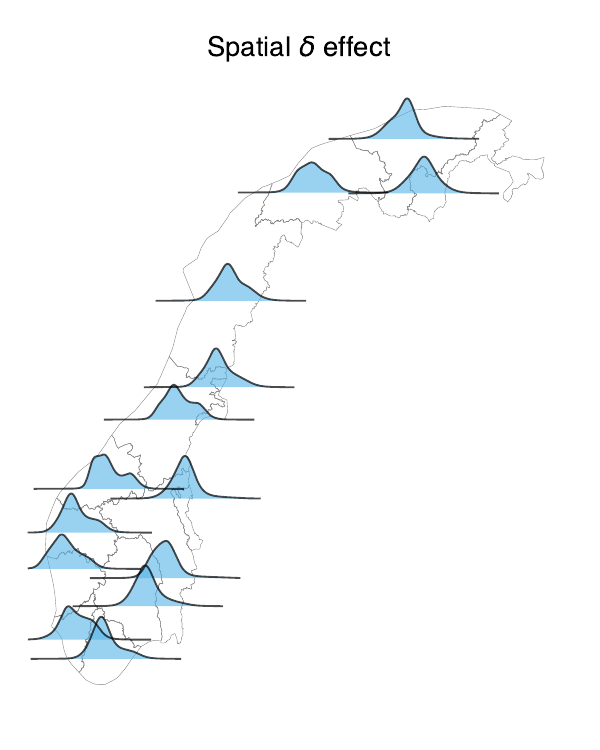}
    \caption{Estimated effects in the PTM with $D=10$ transformation
        parameters, onion-spline core $[a,b]=[-4,4]$, and monitoring-site random
        intercepts in the location and log-scale predictors. The layout mirrors
        \autoref{fig:norway-effects}: the top row shows day-of-year effects on
        location and log-scale and the water-body-type effect on shape; the
        middle row shows year effects on all three predictors; and the bottom
        row shows river-basin district effects. Day-of-year shape previews are
        omitted because their curves overlap almost exactly. Blue shaded
        ribbons are 90\% pointwise credible intervals. Shape effects are
        term-wise standardized-density previews with the remaining centered
        transformation terms held at their reference values; monitoring-site
        effects do not enter the shape predictor.}
    \label{fig:norway-effects-site-ri}
\end{figure}

\end{document}